\RequirePackage[l2tabu, orthodox]{nag}
\documentclass[11pt, a4paper]{article}

\usepackage[top=1truein, bottom=1truein, left=1truein, right=1truein]{geometry}
\usepackage{type1cm}
\usepackage[utf8]{inputenc}
\usepackage[T1]{fontenc}
\usepackage{lmodern}

\usepackage[bookmarksnumbered, pdfdisplaydoctitle, pdfusetitle, unicode]{hyperref}
\usepackage{url}

\usepackage{amsmath}
\usepackage[capitalize, noabbrev]{cleveref}

\usepackage{autonum}

\usepackage{booktabs}
\usepackage{multirow}

\usepackage{tikz}
\usetikzlibrary{arrows.meta}
\usetikzlibrary{positioning}

\usepackage{comment}
\usepackage[en-US]{datetime2}
\usepackage{subcaption}
\usepackage{xcolor}

\usepackage{okicmd}
\usepackage{okithm}

\DeclareMathOperator{\ch}{ch}
\DeclareMathOperator{\GF}{GF}

\newcommand{\setRpp}{\setR_{>0}}

\newcommand{\base}{\mathcal{B}}
\newcommand{\cbase}{\mathcal{B}}
\newcommand{\pbase}{\mathcal{B}}
\newcommand{\feas}{\mathcal{F}}

\newcommand{\matroid}{\mathbf{M}}
\newcommand{\dmatroid}{\mathbf{D}}

\DeclareMathOperator{\pf}{Pf}

\newcommand{\head}[1]{\partial^- #1}
\newcommand{\tail}[1]{\partial^+ #1}

\newcommand{\pbig}[1]{\prn[\big]{#1}}

\title{\texorpdfstring{%
  Pfaffian Pairs and Parities:\\
  Counting on Linear Matroid Intersection and Parity Problems
}{%
  Pfaffian Pairs and Parities: Counting on Linear Matroid Intersection and Parity Problems
}}
\author{
  Kazuki Matoya%
  \texorpdfstring{\thanks{
    Department of Mathematical Informatics, Graduate School of Information Science and Technology, University of Tokyo, Tokyo 113-8656, Japan.
    E-mail: \{\href{mailto:kazuki_matoya@mist.i.u-tokyo.ac.jp}{\nolinkurl{kazuki_matoya}}, \href{mailto:taihei_oki@mist.i.u-tokyo.ac.jp}{\nolinkurl{taihei_oki}}\}\texttt{@mist.i.u-tokyo.ac.jp}
  }}{}
  \and
  Taihei Oki%
  \texorpdfstring{\footnotemark[1]}{}
}
\newcommand{\mykeywords}{linear matroid intersection, linear matroid parity, Pfaffian, matrix-tree theorem, regular delta-matroid, arborescence, matching, Pfaffian orientation, spanning hypertree, \ST\ path, \calS-path, counting algorithm}
\hypersetup{pdfkeywords={\mykeywords}}

\begin{document}
\maketitle

\begin{abstract}
  % !TEX encoding = UTF-8 Unicode
% !TEX root = main.tex

Spanning trees are a representative example of linear matroid bases that are efficiently countable.
Perfect matchings of Pfaffian bipartite graphs are a countable example of common bases of two matrices.
Generalizing these two examples, Webb~(2004) introduced the notion of Pfaffian pairs as a pair of matrices for which counting of their common bases is tractable via the Cauchy--Binet formula.

This paper studies counting on linear matroid problems extending Webb's work.
We first introduce ``Pfaffian parities'' as an extension of Pfaffian pairs to the linear matroid parity problem, which is a common generalization of the linear matroid intersection problem and the matching problem.
We enumerate combinatorial examples of Pfaffian pairs and parities.
The variety of the examples illustrates that Pfaffian pairs and parities serve as a unified framework of efficiently countable discrete structures.
Based on this framework, we derive celebrated counting theorems, such as Kirchhoff's matrix-tree theorem, Tutte's directed matrix-tree theorem, the Pfaffian matrix-tree theorem, and the Lindström--Gessel--Viennot lemma.

Our study then turns to algorithmic aspects.
We observe that the fastest randomized algorithms for the linear matroid intersection and parity problems by Harvey~(2009) and Cheung--Lau--Leung~(2014) can be derandomized for Pfaffian pairs and parities.
We further present polynomial-time algorithms to count the number of minimum-weight solutions on weighted Pfaffian pairs and parities.
Our algorithms make use of Frank's weight splitting lemma for the weighted matroid intersection problem and the algebraic optimality criterion of the weighted linear matroid parity problem given by Iwata--Kobayashi~(2017).

  \bigskip\noindent\textbf{Keywords:} \mykeywords{}
\end{abstract}

\newpage
\tableofcontents

\newpage
% !TEX encoding = UTF-8 Unicode
% !TEX root = main.tex

\section{Introduction}
Let $A$ be a totally unimodular matrix of row-full rank; that is, any minor of $A$ is $0$ or $\pm 1$.
The (generalized) \emph{matrix-tree theorem}~\cite{Maurer1976} claims that the number of column bases of $A$ is equal to $\det A \trsp{A}$.
This can be observed by setting $A_1 = A_2 = A$ in the \emph{Cauchy--Binet formula}
\begin{align} \label{eq:cauchy_binet}
  \det A_1\trsp{A_2} = \sum_{J \subseteq E : \card{J} = r} \det A_1[J] \det A_2[J],
\end{align}
where $A_1, A_2$ are matrices of size $r \times n$ with common column set $E$ and $A_k[J]$ denotes the submatrix of $A_k$ indexed by columns $J \subseteq E$ for $k = 1, 2$.
In the case where $A$ comes from the incidence matrix of an undirected graph, the formula~\eqref{eq:cauchy_binet} provides the celebrated matrix-tree theorem due to Kirchhoff~\cite{Kirchhoff1847} for counting spanning trees.

From a matroidal point of view, the matrix-tree theorem is regarded as a theorem for counting bases of regular matroids, which are a subclass of linear matroids represented by totally unimodular matrices.
Regular matroids are recognized as the largest class of matroids for which base counting is exactly tractable.
For general matroids (even for binary or transversal matroids), base counting is \textsf{\#P-complete}~\cite{Colbourn1995,Snook2012} and hence approximation algorithms have been well-studied~\cite{Anari2018,Anari2019}.

Another example of a polynomial-time countable object is perfect matchings of graphs with \emph{Pfaffian orientation}~\cite{Kasteleyn1961}.
The \emph{Pfaffian} is a polynomial of matrix entries defined for a skew-symmetric matrix $S$ of even order.
If $S$ is the Tutte matrix of a graph $G$, its Pfaffian is the sum over all perfect matchings of $G$ except that each matching has an associated sign as well.
Suppose that edges of $G$ are oriented so that all terms in the Pfaffian become $+1$ by assigning $+1$ or $-1$ to each variable in the Tutte matrix according to the edge direction.
This means that there are no cancellations in the Pfaffian, and thus it coincides with the number of perfect matchings of $G$.
Such an orientation is called \emph{Pfaffian} and a graph that admits a Pfaffian orientation is also called \emph{Pfaffian}.
If $G$ is bipartite, we can consider the determinant of the Edmonds matrix instead of the Pfaffian of the Tutte matrix.
Whereas counting of perfect matchings is \textsf{\#P-complete} even for bipartite graphs~\cite{Valiant1977}, characterizations of Pfaffian graphs and polynomial-time algorithms to give a Pfaffian orientation have been intensively studied~\cite{Kasteleyn1961,Little1974,Robertson1999,Temperley1961,Vazirani1989}.

From the viewpoint of matroids again, the bipartite matching problem is generalized to the \emph{linear matroid intersection problem}~\cite{Edmonds1968,Edmonds1970}.
This is the problem to find a common column base of two matrices $A_1, A_2$ of the same size.
Besides the bipartite matching problem, the linear matroid intersection problem includes a large number of combinatorial optimization problems as special cases, such as problems of finding an arborescence, a colorful spanning tree, and a vertex-disjoint $S$--$T$ path~\cite{Tutte1965}.
The \emph{weighted linear matroid intersection problem} is to find a common base of $A_1, A_2$ that minimizes a given column weight $\funcdoms{w}{E}{\setR}$.
Various polynomial-time algorithms have been proposed for both the unweighted and the weighted linear matroid intersection problems~\cite{Edmonds1968,Edmonds1970,Frank2011,Harvey2009}.
However, the counting of common bases is intractable even for a pair of totally unimodular matrices, as it includes the counting of perfect bipartite matchings.

Commonly generalizing Pfaffian bipartite graphs and regular matroids, Webb~\cite{Webb2004} introduced the notion of a \emph{Pfaffian} (\emph{matrix}) \emph{pair} as a pair of totally unimodular matrices $A_1, A_2$ such that $\det A_1[B] \det A_2[B]$ is constant for any common base $B$ of $A_1$ and $A_2$.
This condition means due to the Cauchy--Binet formula~\eqref{eq:cauchy_binet} that the number of common bases of $(A_1, A_2)$ can be retrieved from $\det A_1 \trsp{A_2}$.
For example, bases of a totally unimodular matrix $A$ are clearly common bases of a Pfaffian pair $(A, A)$.
Webb~\cite{Webb2004} indicated that the set of perfect matchings of a Pfaffian bipartite graph can also be represented as common bases of a Pfaffian pair.
Although the Pfaffian pairs concept nicely integrates these two celebrated countable objects, its existence and importance do not seem to have been recognized besides the original thesis~\cite{Webb2004} of Webb.
We remark that one can remove the assumption of the total unimodularity on $A_1$ and $A_2$ for counting purpose.

The linear matroid intersection and the (nonbipartite) matching problems are commonly generalized to the \emph{linear matroid parity problem}~\cite{Lawler1976}, which is explained as follows.
Let $A$ be a $2r \times 2n$ matrix whose column set is partitioned into pairs, called \emph{lines}.
Let $L$ be the set of lines.
In this paper, we call $(A, L)$ a (\emph{linear}) \emph{matroid parity}.
The linear matroid parity problem on $(A, L)$ is to find a \emph{parity base} of $(A, L)$, which is a column base of $A$ consisting of lines.
Applications of the linear matroid parity problem include the maximum hyperforest problem on a 3-uniform hypergraph, the disjoint $A$-path problem and the disjoint \calS-path problem~\cite{Lovasz1980}.
The linear matroid parity problem is known to be solvable in polynomial time since the pioneering work of Lovász~\cite{Lovasz1980}.
Recently, Iwata--Kobayashi~\cite{Iwata2017} presented the first polynomial-time algorithm for the \emph{weighted linear matroid parity problem}, which is to find a parity base of $(A, L)$ that minimizes a given line weight $\funcdoms{w}{L}{\setR}$.

In this paper, we explore Pfaffian pairs and their generalization to the linear matroid parity problem, which we call \emph{Pfaffian} (\emph{linear matroid}) \emph{parities}.
The contributions of this paper are twofold: structural and algorithmic results.

\subsection{Structural Results}
We introduce a new concept ``Pfaffian parity'' as a matroid parity $(A, L)$ such that $\det A[B]$ is constant for all parity base $B$ of $(A, L)$.
As in the case of Pfaffian pairs, this condition ensures that the number of parity bases can be retrieved from the Pfaffian of a skew-symmetric matrix associated with $(A, L)$.
The proof of this fact relies on a generalization of the Cauchy--Binet formula~\eqref{eq:cauchy_binet} to the Pfaffian given by Ishikawa--Wakayama~\cite{Ishikawa1995} (see \cref{prop:ishikawa_pfaffian_cauchy_binet}).

We then consolidate a list of discrete structures that can be represented as common bases of Pfaffian pairs or parity bases of Pfaffian parities.
Some of them are already (explicitly or implicitly) known, and some are newly proved.
The variety of this list illustrates that Pfaffian pairs and parities serve as a unified framework of discrete structures for which counting is tractable.
Much of celebrated counting theorems are explained in this framework, such as Kirchhoff's matrix-tree theorem~\cite{Kirchhoff1847}, Tutte's directed matrix-tree theorem~\cite{Tutte1948}, the Pfaffian matrix-tree theorem due to Masbaum--Vaintrob~\cite{Masbaum2002}, and the Lindström--Gessel--Viennot (LGV) lemma~\cite{Gessel1985,Lindstrom1973}.
An overview of the list is as follows; see each linked section for the exact problem definitions.

\paragraph{Regular Matroids and Regular Delta-Matroids~\textmd{(\cref{sec:regular_matroids,sec:regular_delta_matroids})}.}
Bases of regular matroids are a trivial example of Pfaffian pairs, as explained above.
We obtain Kirchhoff's matrix-tree theorem~\cite{Kirchhoff1847} as a corollary.

Webb~\cite{Webb2004} showed that one can represent the set of nonsingular principal submatrices of a skew-symmetric totally unimodular matrix as common bases of a Pfaffian pair.
This can be slightly generalized to the feasible sets of \emph{regular delta-matroids}, which are a generalization of regular matroids introduced by Bouchet~\cite{Bouchet1995, Bouchet1998}.
A combinatorial example of regular delta-matroids is Euler tours in 4-regular directed graphs~\cite{Bouchet1995}.

\paragraph{Arborescences~\textmd{(\cref{sec:arborescences})}.}
An \emph{arborescence} of a directed graph $G$ is a rooted directed spanning tree.
It is well-known that arborescences of $G$ are characterized as a common base of two matrices $A_1, A_2$ associated with $G$.
Tutte's directed matrix-tree theorem~\cite{Tutte1948} claims that the number of arborescences of $G$ is equal to $\det A_1 \trsp{A_2}$.
Some known proofs of the directed matrix-tree theorem essentially show that $(A_1, A_2)$ is Pfaffian.
This means that the directed matrix-tree theorem can be treated in the framework of Pfaffian pairs.

\paragraph{Perfect Matchings of Pfaffian Graphs~\textmd{(\cref{sec:perfect_matchings})}.}
We show that the set of perfect matchings of a Pfaffian graph can be seen as parity bases of a Pfaffian parity.
This is an extension of the relationship between a Pfaffian bipartite graph and a Pfaffian pair observed by Webb~\cite{Webb2004}.

%Unfortunately, the counting of perfect matchings is \textsf{\#P-complete} even for bipartite graphs in general~\cite{Valiant1977}.
%Nevertheless, Kasteleyn~\cite{Kasteleyn1967} proved that every planar graph admits a Pfaffian orientation, and this was improved by Little~\cite{Little1974} for graphs that do not contain $K_{3,3}$ as a ``matching minor.''
%Polynomial-time algorithms to give a Pfaffian orientation have also been developed for planar graphs~\cite{Kasteleyn1961,Temperley1961} (known as the FKT algorithm) as well as for $K_{3,3}$-free graphs~\cite{Vazirani1989}.

\paragraph{Spanning Hypertrees of 3-Pfaffian 3-Uniform Hypergraphs~\textmd{(\cref{sec:spanning_hypertrees})}.}
Let $H = (V, \mathcal{E})$ be a 3-graph (3-uniform hypergraph) and $\vec{H} = (V, \vec{\mathcal{E}})$ an orientation of $H$.
Namely, each element in $\mathcal{E}$ is an unordered triple over $V$, and each element in $\vec{\mathcal{E}}$ is an ordered triple.
A spanning hypertree of $H$ has the sign according to the orientation $\vec{H}$.
The Pfaffian matrix-tree theorem~\cite{Masbaum2002} claims that the Pfaffian of a skew-symmetric matrix associated with $\vec{H}$ is the signed sum of all spanning hypertrees of $H$.
The orientation $\vec{H}$ is called \emph{3-Pfaffian} if the signs of all the spanning hypertrees are the same~\cite{Goodall2011}.
This means that the absolute value of the determinant of the Pfaffian turns out to be the number of spanning hypertrees of $H$.
The 3-graph $H$ is also called \emph{3-Pfaffian} if it admits a 3-Pfaffian orientation.
3-Pfaffian orientations for 3-graphs includes Pfaffian orientations for graphs as a special case~\cite{Goodall2011}.

Lovász~\cite{Lovasz1980} presented a reduction of the problem to find a spanning hypertree of $H$ to the linear (graphic) matroid parity problem.
This reduction yields a one-to-one correspondence between spanning hypertrees of $H$ and parity bases of a matroid parity $(A, L)$ constructed from $H$ (with $L = \mathcal{E}$ indeed).
Appropriately reflecting the information on the orientation to the matrix $A$, we show that $(A, L)$ becomes Pfaffian when the orientation is 3-Pfaffian.
More generally, we prove that the sign of a spanning hypertree $T$ is equal to $\det A[B]$, where $B$ is the parity base of $(A, L)$ corresponding to $T$.
Although we can easily confirm this fact by using the Pfaffian matrix-tree theorem, we also provide another proof without relying on the Pfaffian matrix-tree theorem.
This leads us to a new proof of the Pfaffian matrix-tree theorem.

\paragraph{Disjoint $S$--$T$ Paths of DAGs~\textmd{(\cref{sec:dag})}.}
Let $G$ be a directed acyclic graph (DAG) and take disjoint vertex subsets $S$ and $T$ with $\card{S} = \card{T} = k$.
An \emph{$S$--$T$ path} of $G$ is the union of $k$ directed paths, each from a vertex in $S$ to a vertex in $T$
\footnote{
  An $S$--$T$ path generally refers to a single path from $S$ to $T$ rather than the union of such paths.
  In some literature, an $S$--$T$ path in our definition is called a \emph{perfect Menger-type linking}~\cite[Section~2.2.4]{Murota2000}.
}.
An $S$--$T$ path is called (vertex-)\emph{disjoint} if the consisting directed paths are pairwise vertex-disjoint. % chktex 36
A disjoint $S$--$T$ path has the sign according to the pattern of which vertices in $S$ are connected to which vertices in $T$.
The LGV lemma~\cite{Gessel1985,Lindstrom1973} provides a formula on the sum of signs of disjoint $S$--$T$ paths in $G$ via the determinant.
The LGV lemma has various applications in combinatorics and linear algebra~\cite{Gessel1985}; the Cauchy--Binet formula~\eqref{eq:cauchy_binet} can be proved via the LGV lemma for example.

The disjoint $S$--$T$ path problem on $G$ reduces to the bipartite matching problem~\cite{Frank2011}, in which a disjoint $S$--$T$ path is mapped to a perfect bipartite matching bijectively.
We show that this map retains the signs as well.
This means that if all disjoint $S$--$T$ paths of $G$ have the same sign, the set of disjoint $S$--$T$ paths forms common bases of a Pfaffian pair.
We say that such $(S, T)$ is in the \emph{LGV position} on $G$ and illustrate two examples arising from planar graphs.
We further provide a new proof of the LGV lemma making use of this map.

\paragraph{Shortest Disjoint $S$--$T$ Paths and $S$--$T$--$U$ Paths~\textmd{(\cref{sec:st,sec:stu})}.}
We generalize the above arguments of the disjoint $S$--$T$ path problem on DAGs to \emph{Mader's disjoint \calS-path problem}~\cite{Gallai1964,Mader1978} on undirected graphs.
Let $G$ be an undirected graph and $\mathcal{S}$ a family of disjoint vertex subsets.
Suppose that $\card{\cup_{S \in \mathcal{S}} S} = 2k$.
An \emph{\calS-path} of $G$ is the union of $k$ paths, each of which connects vertices belonging to distinct parts in $\mathcal{S}$.
An \calS-path is called \emph{disjoint} if the consisting paths are pairwise vertex-disjoint.
The disjoint \calS-path problem on $G$ is to find a disjoint \calS-path of $G$.
As for disjoint $S$--$T$ paths of DAGs, the sign of a disjoint \calS-path is defined based on the connecting pattern on $\mathcal{S}$.
When $\card{\mathcal{S}} = 2$, we call an \calS-path an $S$--$T$ path (with $\mathcal{S} = \set{S, T}$).
When $\card{\mathcal{S}} = 3$, we refer to \calS-path as an $S$--$T$--$U$ path (with $\mathcal{S} = \set{S, T, U}$).

Tutte~\cite{Tutte1965} proposed a reduction of the disjoint $S$--$T$ path problem to the linear (graphic) matroid intersection problem.
Subsequently, Schrijver~\cite{Schrijver2003} presented a reduction of the disjoint \calS-path problem to the linear matroid parity problem based on the Lovász' reduction~\cite{Lovasz1980}.
Suppose also that $G$ is equipped with a positive edge length, and consider the shortest disjoint \calS-path problem, which is to find a disjoint \calS-path that minimizes the sum of edge lengths.
Yamaguchi~\cite{Yamaguchi2016} showed that the shortest disjoint \calS-path problem is reduced to the weighted linear matroid parity problem.
As a special case, the shortest disjoint $S$--$T$ path problem reduces to the weighted linear matroid intersection problem.
This is a generalization of the reduction from the disjoint $S$--$T$ path problem on a DAG to the bipartite matching problem.

We first deal with the disjoint $S$--$T$ path problem on $G$.
Unfortunately, Tutte's reduction provides only a one-to-many correspondence between disjoint $S$--$T$ paths and common bases of a matrix pair $(A_1, A_2)$.
Nevertheless, we show that the sign of a disjoint $S$--$T$ path $P$ coincides with $\det A_1[B] \det A_2[B]$, where $B$ is a common base corresponding to $P$.
In addition, the weighted reduction gives a one-to-one correspondence between optimal solutions.
As a consequence, if $(S, T)$ is in the \emph{LGV position}, i.e., the signs of all disjoint $S$--$T$ paths are the same, we can represent the set of shortest disjoint $S$--$T$ paths of $G$ as minimum-weight common bases of a Pfaffian pair.

We next consider the general disjoint \calS-path problem.
Like the $S$--$T$ case, Schrijver's reduction for unweighted problems constructs only a one-to-many correspondence, whereas Yamaguchi's reduction for weighted problems provides a one-to-one correspondence.
Unlike the $S$--$T$ case, however, it will be turned out that $\det A[B]$ depends on some factor other than the sign of a disjoint \calS-path $P$, where $A$ is the matrix in the reduced linear matroid parity problem and $B$ is a parity base corresponding to $P$.
Nonetheless, when $\card{\mathcal{S}} = 3$, i.e., in the $S$--$T$--$U$ case, the factor would be constant for any disjoint \calS-path.
This means the shortest disjoint $S$--$T$--$U$ paths are represented as minimum-weight parity bases of a weighted Pfaffian parity when $S, T, U$ are in the \emph{LGV position}.

\subsection{Algorithmic Results}

Let $(A_1, A_2)$ be an $r \times n$ Pfaffian pair and $(A, L)$ a $2r \times 2n$ Pfaffian parity.
The definitions of Pfaffian pairs and parities guarantee that one can count the number of common bases of $(A_1, A_2)$ and the number of parity bases of $(A, L)$ just by matrix computations.
We observe that these computations can be done in $\Order\prn{nr^{\omega-1}}$ or $\Order\prn{nr^{\omega-1} + r^3}$-time (see \cref{thm:complexity_c_known}), where we assume that arithmetic operations can be performed in constant time.
Here, $2 < \omega \le 3$ is the matrix multiplication exponent, i.e., the multiplication of two $r \times r$ matrices is performed in $\Order\prn{r^\omega}$-time.
The current best value of $\omega$ is $2.3728639$~\cite{LeGall2014}.
More generally and precisely, when the matrices are over a field $\setK$ of characteristic $\ch(\setK)$, we can compute the number of common or parity bases modulo $\ch(\setK)$ within these times.

In the above arguments, we implicitly assumed that we know the value of $\det A_1[B] \det A_2[B]$ with an arbitrary common base $B$ of $(A_1, A_2)$ and $\det A[B]$ with an arbitrary parity base $B$ of $(A, L)$.
These values are called \emph{constants}.
If we do not know the value of constants, then we need to obtain one common or parity base $B$ beforehand by executing linear matroid intersection and parity algorithms.
The current best time complexities for solving the linear matroid intersection problem is deterministic $\Order\prn{nr^{\frac{5-\omega}{4-\omega}} \log r}$-time due to Gabow--Xu~\cite{Gabow1996} and randomized $\Order\prn{nr^{\omega-1}}$-time due to Harvey~\cite{Harvey2009}.
Also, the current best time complexities for the linear matroid parity problem is deterministic $\Order\prn{nr^\omega}$-time due to Gabow--Stallmann~\cite{Gabow1986} and Orlin~\cite{Orlin2008}, and randomized $\Order\prn{nr^{\omega-1}}$-time due to Cheung--Lau--Leung~\cite{Cheung2014}.
Therefore, we are confronted with a choice of whether to stick to deterministic algorithms or to employ randomized algorithms to keep the running time.
We would face the same trade-off to find one common or parity base even if we know the constants.

For the case of $\ch(\setK) = 0$, we show that it is possible to cherry-pick good points of both choices as follows.

\begin{theorem}\label{thm:complexity_of_unweighted_counting}
  When $\ch(\setK) = 0$, we can count the number of common bases of an $r \times n$ Pfaffian pair and the number of parity bases of a $2r \times 2n$ Pfaffian parity in deterministic $\Order\prn{nr^{\omega-1}}$-time.
  In addition, we can construct one common or parity base in the same time complexity.
\end{theorem}

Intuitively, the algorithms of Harvey~\cite{Harvey2009} and Cheung--Lau--Leung~\cite{Cheung2014} make use of randomness to take a random vector avoiding numerical cancellations between common or parity bases.
We show that Pfaffian pairs and parities do not involve such numerical cancellations by their definitions.

We next consider the problems of counting the number of minimum-weight common bases of a column-weighted Pfaffian pair and the number of minimum-weight parity bases of a line-weighted Pfaffian parity.
These problems can be algebraically formulated by using a univariate polynomial matrix (assuming the weight to be integral).
In these formulations, the number of minimum-weight common or parity bases is obtained as the coefficient of the lowest degree term in the determinant or Pfaffian of the polynomial matrix.
While we can compute it by performing a symbolic computation, this yields only a pseudo-polynomial time algorithm.

Broder--Mayr~\cite{Broder1997} and Hayashi--Iwata~\cite{Hayashi2018} presented polynomial-time counting algorithms for minimum-weight spanning trees and minimum-weight arborescences, respectively.
Their algorithms first compute dual optimal solutions of linear programming (LP) formulations and then perform graphic operations on trees constructed from the dual optimal solutions.

Generalizing these algorithms from a matroidal perspective, we present a polynomial-time counting algorithm for minimum-weight common bases of a Pfaffian pair.
We make use of Frank's weight splitting lemma~\cite{Frank1981}, which reveals the dual structure of the weighted matroid intersection problem.
Applying a row operation, we reduce the counting on a weighted Pfaffian pair to the counting on an unweighted Pfaffian pair.
Our reduction can be seen as a succinct description of a known trick to represent minimum-weight common bases of a weighted matrix pair as the set of common bases of unweighted matrix pair.
The running time is estimated as follows.

\begin{theorem}\label{thm:complexity_of_counting_weighted_pairs}
  Let $(A_1, A_2)$ be a Pfaffian pair with column weight $\funcdoms{w}{E}{\setR}$.
  We can compute the number of minimum-weight common bases of $(A_1, A_2)$ modulo $\ch(\setK)$ in deterministic $\Order\prn{nr^\omega + nr \log n}$-time.
\end{theorem}

We also present a polynomial-time counting algorithm for weighted Pfaffian parities.
Although an LP formulation of the weighted linear matroid parity algorithm is not yet known, Iwata--Kobayashi~\cite{Iwata2017} gave an algebraic optimality criterion, which associates the minimum weight of a parity base with the maximum weight of a perfect matching of a graph.
Based on this association, we show that the number of minimum-weight parity bases coincides with the leading coefficient of a skew-symmetric polynomial matrix that the algorithm of Iwata--Kobayashi outputs as a byproduct.
We then apply Murota's upper-tightness testing algorithm~\cite{Murota1995a} to compute the leading coefficient.
Murota's algorithm was originally presented in the context of combinatorial relaxation, which is to compute the degree of the determinant (Pfaffian) of a skew-symmetric polynomial matrix.
The time complexity is summarized as follows.

\begin{theorem}\label{thm:counting_minimum_weight_parity_bases}
  Let $(A, L)$ be a Pfaffian parity with line weight $\funcdoms{w}{L}{\setR}$.
  We can compute the number of minimum-weight parity bases of $(A, L)$ modulo $\ch(\setK)$ in deterministic $\Order\prn{n^3r}$-time.
\end{theorem}

On describing time complexities, we have assumed that arithmetic operations on $\setK$ can be performed in constant time.
This assumption is reasonable when $\setK$ is a finite field of fixed order.
When $\setK$ is the field $\setQ$ of rational numbers, there is no guarantee that a direct application of the algorithm of Iwata--Kobayashi~\cite{Iwata2017} does not swell the bit-lengths of intermediate numbers.
Instead, they showed that one can solve the weight linear matroid parity problem over $\setQ$ by applying their algorithm over a sequence of finite fields.
We give a polynomial-time counting algorithm with $\setK = \setQ$ based on their reduction.

\begin{theorem}\label{thm:bit_complexity}
  Let $(A, L)$ be a Pfaffian parity over $\setQ$ with line weight $\funcdoms{w}{L}{\setR}$.
  We can deterministically compute the number of minimum-weight parity bases of $(A, L)$ in time polynomial in the binary encoding length of $A$.
\end{theorem}

As seen above, our counting algorithms for weighted Pfaffian pairs and parities are based on different approaches: that for pairs reduces to the unweighted counting and that for parities is via the matching problem.
We show in Appendix that the algorithm for Pfaffian pairs can also be derived by an approach based on the bipartite matching problem.

\subsection{Organization}
The rest of this paper is organized as follows.
After introducing some preliminaries, \cref{sec:pfaffian_matroid_pairs_and_parities} gives formal definitions of Pfaffian pairs and parities as well as their properties.
\Cref{sec:examples,sec:lgv} exhibit examples of Pfaffian pairs and parities.
The family of disjoint path problems are dealt with in \cref{sec:lgv} and others are in \cref{sec:examples}.
Finally, \cref{sec:algorithms} presents our counting algorithms for unweighted and weighted Pfaffian pairs and parities.

\section{Pfaffian Pairs and Pfaffian Parities}\label{sec:pfaffian_matroid_pairs_and_parities}

\subsection{Preliminaries}\label{sec:preliminaries}
Let $\setZ, \setQ$ and $\setR$ denote the set of all integers, rational and real numbers, respectively.
For a nonnegative integer $n$, we denote $\set{1, 2, \dotsc, n}$ by $\intset{n}$.
Let $\setK$ be a field of characteristic $\ch(\setK)$.
Unless otherwise stated, all matrices are over $\setK$ in \cref{sec:pfaffian_matroid_pairs_and_parities,sec:algorithms} and are over $\setQ$ in \cref{sec:examples,sec:lgv}.
For $n \in \setZ$, we define $n$ modulo $0$ as $n$ for convenience.
For $n,m \in \setZ$, ``$n$ is equal to $m$ over $\setK$'' means $n$ is congruent to $m$ modulo $\ch(\setK)$.
For a matrix $A$, we denote by $A[I, J]$ the submatrix of $A$ with row subset $I$ and column subset $J$.
If $I$ is all the rows of $A$, we denote $A[I, J]$ by $A[J]$.

While we can describe this paper without defining matroids, here we give a general definition.
A \emph{matroid} is the pair $\mathbf{M} = (E, \mathcal{B})$ of a finite set $E$ and a nonempty set family $\mathcal{B} \subseteq 2^E$ over $E$ satisfying the following: for any $B_1, B_2 \in \mathcal{B}$ and $x \in B_1 \setminus B_2$, there exists $y \in B_2 \setminus B_1$ such that $B_1 \setminus \set{x} \cup \set{y} \in \mathcal{B}$.
Each element of $\mathcal{B}$ is called a \emph{base} of $\mathbf{M}$ and $E$ is called the \emph{ground set} of $\mathbf{M}$.

Typical examples of matroids arise from matrices.
Let $A \in \setK^{r \times n}$ be a matrix with column set $E$.
Define
\begin{align}
  \base(A) \defeq \set{B \subseteq E}[\text{$\card{B} = r$, $A[B]$ is nonsingular}].
\end{align}
If $A$ is of row-full rank, then $\matroid(A) \defeq (E, \base(A))$ forms a matroid, called a \emph{linear matroid} represented by $A$.
We refer to each element of $\base(A)$ as a \emph{base} of $A$.
In this paper, we consider $A$ to have no base if $A$ is not of row-full rank.

Recall that the \emph{determinant} of a square matrix $A = \prn{A_{i,j}}_{i,j \in \intset{n}} \in \setK^{n \times n}$ is defined as
\begin{align}\label{def:determinant}
  \det A \defeq \sum_{\sigma \in \sym_n} \sgn \sigma \prod_{i=1}^n A_{i, \sigma(i)},
\end{align}
where $\sym_n$ is the set of all permutations on $\intset{n}$ and $\sgn \sigma$ denotes the sign of a permutation $\sigma \in \sym_n$.
A square matrix $S$ is said to be \emph{skew-symmetric} if $\trsp{S} = -S$ and all diagonal entries are zero, where the latter condition cares the case when $\ch(\setK) = 2$.
For a skew-symmetric matrix $S = \prn{S_{i,j}}_{i,j \in \intset{2n}} \in \setK^{2n \times 2n}$ of even order, the \emph{Pfaffian} of $S$ is defined as
\begin{align}\label{def:pfaffian}
  \pf S \defeq \sum_{\sigma \in F_{2n}} \sgn \sigma \prod_{i=1}^n S_{\sigma(2i-1), \sigma(2i)},
\end{align}
where $F_{2n}$ is the subset of $\sym_{2n}$ given by
\begin{align}\label{def:F}
  F_{2n} \defeq \set{\sigma \in \sym_{2n}}[\text{$\sigma(1) < \sigma(3) \dotsb < \sigma(2n-1)$ and $\sigma(2i-1) < \sigma(2i)$ for $i \in \intset{n}$}].
\end{align}

It is well-known that
\begin{align}
  \prn{\pf S}^2  & = \det S,\label{eq:pf_det}                                          \\
  \pf AS\trsp{A} & = \det A \pf S\label{eq:pfaffian_multiplicativity}
\end{align}
hold, where $A \in \setK^{2n \times 2n}$ is any square matrix.
The following formula is a generalization of the Cauchy--Binet formula~\eqref{eq:cauchy_binet} to Pfaffian given by Ishikawa--Wakayama~\cite{Ishikawa1995}.
\begin{proposition}[{\cite[Theorem~1]{Ishikawa1995}}]\label{prop:ishikawa_pfaffian_cauchy_binet}
  Let $S \in \setK^{2n \times 2n}$ be a skew-symmetric matrix and $A \in \setK^{2r \times 2n}$ a square matrix.
  Suppose that the row and column sets of $S$ and the column set of $A$ are indexed by $E$.
  Then it holds
  \begin{align}
    \pf AS\trsp{A} = \sum_{\condit{J \subseteq E}[\card{J} = 2r]} \det A[J] \pf S[J, J].\label{eq:pfaffian_cauchy_binet}
  \end{align}
\end{proposition}

Let $A$ be a square matrix partitioned as $A = \begin{psmallmatrix} X & Y \\ Z & W \end{psmallmatrix}$.
Suppose that $X$ and $W$ are square, and $W$ is nonsingular.
The \emph{Schur complement} of $A$ with respect to $W$ is the matrix $X - YW^{-1}Z$.
This matrix is the one that appears at the position of $X$ after eliminating $Y$ and $Z$ by elementary operations using $W$.
Since elementary operations do not change the determinant, it holds
\begin{align}\label{eq:shurt_det}
  \det A = \det W \det\prn[\big]{X - YW^{-1}Z}.
\end{align}
Elementary operations also retain the Pfaffian of a skew-symmetric matrix by~\eqref{eq:pfaffian_multiplicativity}.
Thus if $S = \begin{psmallmatrix} X & Y \\ -\trsp{Y} & W \end{psmallmatrix}$ is skew-symmetric and $W$ is nonsingular, we have
\begin{align}\label{eq:shurt_pf}
  \pf S = \pf W \pf\prn[\big]{X + YW^{-1}\trsp{Y}}.
\end{align}

\subsection{Linear Matroid Intersection Problem and Pfaffian Pairs}\label{sec:linear_matroid_intersection_problem_and_pfaffian_pairs}
The \emph{matroid intersection problem} introduced by Edmonds~\cite{Edmonds1968,Edmonds1970} is the following: given two matroids $\mathbf{M}_1 = (E, \mathcal{B}_1)$ and $\mathbf{M}_2 = (E, \mathcal{B}_2)$ over the same ground set $E$, we find a common base $B \in \mathcal{B}_1 \cap \mathcal{B}_2$ of $\mathbf{M}_1$ and $\mathbf{M}_2$.
The \emph{linear matroid intersection problem} is to find a common base of two linear matroids.
We regard the input of the linear matroid intersection problem as a \emph{matrix pair} $(A_1, A_2)$, which is the pair of matrices $A_1, A_2 \in \setK^{r \times n}$ of the same size over the same ground field $\setK$.
We denote the set of common bases of $A_1$ and $A_2$ by $\cbase(A_1, A_2) \defeq \base(A_1) \cap \base(A_2)$.

The linear matroid intersection problem can be algebraically formulated as follows.
For a vector $z = \prn{z_j}_{j \in E}$ indexed by the common column set $E$ of $(A_1, A_2)$, we denote the diagonal matrix $\diag\prn{z_j}_{j \in E}$ by $D(z)$.
We also define a block matrix
\begin{align}\label{def:Xi}
  \Xi_{A_1, A_2}(z)
  \defeq
  \begin{pmatrix}
    O & A_1 \\\trsp{A_2} &
    D(z)
  \end{pmatrix},
\end{align}
where $O$ denotes the zero matrix of appropriate size.
We henceforth omit the subscript $A_1, A_2$ of $\Xi$ as it will be clear from the context.

\begin{proposition}[{see~\cite{Geelen2001,Tomizawa1974}}]\label{prop:algebraic_formulation_of_intersection}
  Let $(A_1, A_2)$ be a matrix pair and $z = \prn{z_j}_{j \in E}$ a vector of distinct indeterminates indexed by the common column set $E$ of $(A_1, A_2)$.
  Then the following are equivalent:
  \begin{enumerate}
    \item $(A_1, A_2)$ has a common base.\label{item:algebraic_formulation_of_intersection_1}
    \item $A_1 D(z) \trsp{A_2}$ is nonsingular.\label{item:algebraic_formulation_of_intersection_2}
    \item $\Xi(z)$ is nonsingular.\label{item:algebraic_formulation_of_intersection_3}
  \end{enumerate}
\end{proposition}
Here, the nonsingularity in \cref{prop:algebraic_formulation_of_intersection} is in the sense of matrices over the rational function field $\setK(z) \defeq \setK\prn{\set{z_j}[j \in E]}$.
As indicated by Tomizawa--Iri~\cite{Tomizawa1974}, the equivalence of \cref{prop:algebraic_formulation_of_intersection}~\ref{item:algebraic_formulation_of_intersection_1} and~\ref{item:algebraic_formulation_of_intersection_2} can be seen from the Cauchy--Binet formula~\eqref{eq:cauchy_binet} because the formula expands $\det A_1 D(z) \trsp{A_2}$ as
\begin{align}\label{eq:cauchy_binet_1}
  \det A_1 D(z) \trsp{A_2}
  = \sum_{B \in \cbase(A_1, A_2)} \det A_1[B] \det A_2[B] \prod_{j \in B} z_j.
\end{align}
This equation means that $\det A_1 D(z) \trsp{A_2} \ne 0$ if and only if $\cbase(A_1, A_2) \ne \varnothing$ since the factor on $x$ avoids cancellations in the summation.
Considering the formula~\eqref{eq:shurt_det} on the Schur complement and~\eqref{eq:cauchy_binet_1}, we also have
\begin{align}\label{eq:cauchy_binet_2}
  \det \Xi(z)
  = \det A_1 {D(z)}^{-1} \trsp{A_2} \cdot \det D(z)
  = \sum_{B \in \cbase(A_1, A_2)} \det A_1[B] \det A_2[B] \prod_{j \in E \setminus B} z_j.
\end{align}
Hence all the claims in \cref{prop:algebraic_formulation_of_intersection} are equivalent.
%We can devise a simple randomized algorithm for the linear matroid intersection problem using \cref{prop:algebraic_formulation_of_intersection} and the Schwartz--Zippel lemma~\cite{Schwartz1980,Zippel1979}.
See also Harvey~\cite{Harvey2009} and Murota~\cite[Remark~2.3.37]{Murota2000}.

Now we define Pfaffian matrix pairs slightly generalizing that of Webb~\cite{Webb2004}.

\begin{definition}[{Pfaffian matrix pair; see~\cite{Webb2004}}]\label{def:pfaffian_pair}
  We say that a matrix pair $(A_1, A_2)$ is \emph{Pfaffian} if there exists $c \in \setK \setminus \set{0}$ such that $\det A_1[B] \det A_2[B] = c$ for all $B \in \cbase(A_1, A_2)$.
  The value $c$ is called the \emph{constant} of $(A_1, A_2)$.
\end{definition}

We abbreviate a Pfaffian matrix pair as a Pfaffian pair.
If $(A_1, A_2)$ is Pfaffian, nonzero terms in the summation of~\eqref{eq:cauchy_binet_1} and~\eqref{eq:cauchy_binet_2} do not cancel out.
Hence the following holds for Pfaffian pairs.

\begin{proposition}\label{prop:counting_formula_for_pfaffian_pairs}
  Let $(A_1, A_2)$ be a Pfaffian pair of constant $c$ and $z = \prn{z_j}_{j \in E}$ a vector indexed by the common column set $E$ of $(A_1, A_2)$.
  Then it holds
  \begin{align}
    \det A_1 D(z) \trsp{A_2} & = c \sum_{B \in \cbase(A_1, A_2)} \prod_{j \in B} z_j,\label{eq:symbolic_counting_pairs} \\
    \det \Xi(z)              & = c \sum_{B \in \cbase(A_1, A_2)} \prod_{j \in E \setminus B} z_j.
  \end{align}
  In particular, the number of common bases of $(A_1, A_2)$ is equal to $c^{-1} \det A_1\trsp{A_2} = c^{-1} \det \Xi(\onevec)$ over $\setK$, where $\onevec$ denotes the vector of ones with appropriate dimension.
\end{proposition}

Next, consider a column-weighted version of matrix pairs.
Let $(A_1, A_2)$ be a matrix pair and $\funcdoms{w}{E}{\setR}$ a weight on the common column set $E$ of $(A_1, A_2)$.
The \emph{weight} $w(J)$ of $J \subseteq E$ is defined as $w(J) \defeq \sum_{j \in J} w(j)$.
The \emph{weighted linear matroid intersection problem} is to find a common base of $(A_1, A_2)$ that minimizes the weight $w$ among all common bases.
It is well-known that one can algebraically encode the information on the weight $w$ by putting it to the power of an indeterminate $\theta$, as the following proposition shows.
We define $\theta^w \defeq \pbig{\theta^{w(j)}}_{j \in E}$.

\begin{proposition}\label{prop:algebraic_weighted_pair}
  Let $(A_1, A_2)$ be a matrix pair with column weight $\funcdoms{w}{E}{\setR}$ and let $\theta$ be an indeterminate.
  For $x \in \setR$, the coefficient of $\theta^x$ in $\det A_1 D\pbig{\theta^w} \trsp{A_2}$ and the coefficient of $\theta^{w(E)-x}$ in $\det \Xi\pbig{\theta^w}$ are equal to
  \begin{align}\label{eq:weighted_counting_pair}
    \sum_{B \in \cbase_x} \det A_1[B] \det A_2[B],
  \end{align}
  where $\cbase_x \defeq \set{B \in \cbase(A_1, A_2)}[w(B) = x]$.
  In particular, if $(A_1, A_2)$ is Pfaffian with constant $c$, the coefficients of $\theta^x$ in $\det A_1 D\pbig{\theta^w} \trsp{A_2}$ and of $\theta^{w(E)-x}$ in $\det \Xi\pbig{t\theta^w}$ are equal to $c^{-1}\card{\cbase_x}$ over $\setK$.
\end{proposition}
\begin{proof}
  By~\eqref{eq:cauchy_binet_1}, we have
  \begin{align}
    \det A_1 D\pbig{\theta^w} \trsp{A_2}
     & = \sum_{B \in \cbase(A_1, A_2)} \det A_1[B] \det A_2[B] \prod_{j \in B} \theta^{w(j)}                                             \\
     & = \sum_{B \in \cbase(A_1, A_2)} \det A_1[B] \det A_2[B] \theta^{w(B)}.\label{eq:expansion_of_weighted_pair}
  \end{align}
  Hence the coefficient of $\theta^x$ in~\eqref{eq:expansion_of_weighted_pair} is equal to~\eqref{eq:weighted_counting_pair} for $x \in \setR$.
  We can show the claim for $\det \Xi\pbig{\theta^w}$ in the same way via~\eqref{eq:cauchy_binet_1}.
\end{proof}

\subsection{Linear Matroid Parity Problem and Pfaffian Parities}
Let $\mathbf{M} = (E, \mathcal{B})$ be a matroid with $\card{E}$ being even.
The ground set $E$ is partitioned into pairs, called \emph{lines}.
Let $L$ be the set of lines.
The \emph{matroid parity problem} (also known as the \emph{matchoid problem} or the \emph{matroid matching problem}), introduced by Lawler~\cite{Lawler1976}, is to find a base of $\mathbf{M}$ consisting of lines.
Such a base is called a \emph{parity base} of $\mathbf{M}$ (with respect to $L$).
In the general case, the matroid parity problem requires exponential number of membership oracle calls of $\mathcal{B}$~\cite{Lovasz1980}.
Nevertheless, Lovász~\cite{Lovasz1980} showed that the \emph{linear matroid parity problem} admits a polynomial-time algorithm, in which the linear matroid is given as a matrix $A$.
Here, the numbers of rows and columns of $A$ are even, say, $A \in \setK^{2r \times 2n}$.
We call the pair $(A, L)$ a (linear) \emph{matroid parity}.
We regard parity bases of $(A, L)$ as a subset of $L$ and denote by $\pbase(A, L)$ the set of all parity bases of $\matroid(A)$ with respect to $L$.
For $J \subseteq L$, we denote by $A[J]$ the submatrix of $A$ consisting of columns corresponding to lines in $J$.

The linear matroid parity problem also has algebraic formulations.
For a vector $z = \prn{z_l}_{l \in L}$ indexed by $L$, we denote by $\Delta_L(z)$ the $2n \times 2n$ skew-symmetric block-diagonal matrix defined as follows: the row and column sets are indexed by $E$, and each block corresponding to a line $l \in L$ is a $2 \times 2$ skew-symmetric matrix $\begin{psmallmatrix} \phantom{+}0 & +z_l \\ -z_l & \phantom{+}0 \end{psmallmatrix}$.
Similarly to~\eqref{def:Xi}, we define a skew-symmetric block matrix
\begin{align}\label{def:Phi}
  \Phi_{A, L}(z)
  \defeq
  \begin{pmatrix}
    O         & A           \\
    -\trsp{A} & \Delta_L(z)
  \end{pmatrix}.
\end{align}
We also omit the subscripts $L$ of $\Delta$ and $A, L$ of $\Phi$ as they will be always clear.

\begin{proposition}[{\cite{Geelen2005,Lovasz1979}}]\label{prop:algebraic_formulation_of_parity}
  Let $(A, L)$ be a matroid parity and $z = \prn{z_l}_{l \in L}$ a vector of distinct indeterminates indexed by $L$.
  Then the following are equivalent:
  \begin{enumerate}
    \item $(A, L)$ has a parity base.\label{prop:algebraic_formulation_of_parity_1}
    \item $A \Delta(z) \trsp{A}$ is nonsingular.\label{prop:algebraic_formulation_of_parity_2}
    \item $\Phi(z)$ is nonsingular.\label{prop:algebraic_formulation_of_parity_3}
  \end{enumerate}
\end{proposition}

We note that the matrix $A \Delta(z) \trsp{A}$ in \cref{prop:algebraic_formulation_of_parity}~\ref{prop:algebraic_formulation_of_parity_2} can also be written as
\begin{align}\label{eq:rewrite_ADeltaA}
  A \Delta(z) \trsp{A}
  = A_1 D(z) \trsp{A_2} - A_2 D(z) \trsp{A_1}
  = \sum_{l = (v, \bar{v}) \in L} z_l \pbig{a_v\trsp{a_{\bar{v}}} - a_{\bar{v}} \trsp{a_v}},
\end{align}
where $a_v \defeq A[\set{v}]$ is the $v$th column of $A$ for $v \in E$ and $A_1, A_2$ are $2r \times n$ submatrices of $A$ consisting of column vectors $a_v$ and $a_{\bar{v}}$ for each line $(v, \bar{v}) \in L$, respectively.
Lovász~\cite[Theorem~3]{Lovasz1979} described the equivalence of \cref{prop:algebraic_formulation_of_parity}~\ref{prop:algebraic_formulation_of_parity_1} and~\ref{prop:algebraic_formulation_of_parity_2} representing $A \Delta(z) \trsp{A}$ in the rightmost form of~\eqref{eq:rewrite_ADeltaA}.
The equivalence of \cref{prop:algebraic_formulation_of_parity}~\ref{prop:algebraic_formulation_of_parity_1} and~\ref{prop:algebraic_formulation_of_parity_3} is due to Geelen--Iwata~\cite[Theorem~4.1]{Geelen2005}; see also Cheung--Lau--Leung~\cite{Cheung2014} and Murota~\cite[Remark~7.3.23]{Murota2000}.
These equivalences can also be observed from the following identities.
\begin{proposition}\label{prop:pfaffian_cauchy_binet}
  Let $(A, L)$ be a matroid parity and $z = \prn{z_l}_{l \in L}$ a vector indexed by $L$.
  Then it holds
  \begin{align}
    \pf A\Delta(z)\trsp{A} & = \sum_{B \in \pbase(A, L)} \det A[B] \prod_{l \in B} z_l,\label{eq:pfaffian_cauchy_binet_1}             \\
    \pf \Phi(z)            & = \sum_{B \in \pbase(A, L)} \det A[B] \prod_{l \in L \setminus B} z_l.\label{eq:pfaffian_cauchy_binet_2}
  \end{align}
\end{proposition}
\begin{proof}
  Applying the the expanding formula~\eqref{eq:pfaffian_cauchy_binet} to $\pf A\Delta(z)\trsp{A}$, we have
  \begin{align}
    \pf A\Delta(z)\trsp{A} = \sum_{\condit{J \subseteq E}[\card{J} = 2r]} \det A[J] \pf \Delta(z)[J, J],
  \end{align}
  where $2r$ is the cardinality of rows of $A$ and $E$ is the columns of $A$.
  From the definitions of $\Delta(z)$ and Pfaffian, $\Delta(z)[J, J]$ is nonsingular only if $J$ consists of lines.
  In this case, $\pf \Delta(z)[J, J]$ is equal to the product of $z_l$ for every line $l$ consisting $J$.
  Hence~\eqref{eq:pfaffian_cauchy_binet_1} is obtained.

  The latter identity~\eqref{eq:pfaffian_cauchy_binet_2} is obtained by applying the formula~\eqref{eq:shurt_pf} on the Schur complement to $\Phi(z)$.
  Note that $\Delta(z)$ can be regarded as nonsingular by seeing each $z_l$ as an indeterminate.
\end{proof}

Now we define Pfaffian matroid parities in the same manner as \cref{def:pfaffian_pair}.
\begin{definition}[{Pfaffian matroid parity}]\label{def:pfaffian_parity}
  We say that a matroid parity $(A, L)$ is \emph{Pfaffian} if there exists $c \in \setK \setminus \set{0}$ such that $\det A[B] = c$ for all $B \in \pbase(A, L)$.
  The value $c$ is called the \emph{constant} of $(A, L)$.
\end{definition}
We abbreviate Pfaffian matroid parity as Pfaffian parity for short.
The following is immediately obtained from \cref{prop:pfaffian_cauchy_binet} and \cref{def:pfaffian_parity}.

\begin{proposition}\label{prop:counting_formula_for_pfaffian_prities}
  Let $(A, L)$ be a Pfaffian parity of constant $c$ and $z = \prn{z_l}_{l \in L}$ a vector indexed by $L$.
  Then it holds
  \begin{align}
    \pf A\Delta(z)\trsp{A} & = c \sum_{B \in \pbase(A, L)} \prod_{l \in B} z_l,             \\
    \pf \Phi(z)            & = c \sum_{B \in \pbase(A, L)} \prod_{l \in L \setminus B} z_l.
  \end{align}
  In particular, the number of parity bases of $(A, L)$ is equal to $c^{-1} \pf A\Delta(\onevec)\trsp{A} = c^{-1} \pf \Phi(\onevec)$ over $\setK$.
\end{proposition}

We next consider the weighted linear matroid parity problem.
Let $(A, L)$ be a matroid parity and $\funcdoms{w}{L}{\setR}$ a weight on lines.
The \emph{weight} of $J \subseteq L$ is defined as $w(J) \defeq \sum_{j \in J} w(j)$.
The \emph{weighted linear matroid parity problem} is to find a parity base of $(A, L)$ having the minimum weight with respect to $w$ among all parity bases.
The following is obtained in the same way as \cref{prop:algebraic_weighted_pair} via \cref{prop:pfaffian_cauchy_binet}; see also Iwata--Kobayashi~\cite{Iwata2017}.

\begin{proposition}\label{prop:algebraic_weighted_parity}
  Let $(A, L)$ be a matroid parity equipped with a line weight $\funcdoms{w}{L}{\setR}$.
  Let $\theta$ be an indeterminate.
  For $x \in \setR$, the coefficient of $\theta^x$ in $\pf A \Delta\pbig{\theta^w} \trsp{A}$ and the coefficient of $\theta^{w(L)-x}$ in $\pf \Phi\pbig{\theta^w}$ are equal to
  \begin{align}\label{eq:weighted_counting_parity}
    \sum_{B \in \pbase_x} \det A[B],
  \end{align}
  where $\pbase_x \defeq \set{B \in \pbase(A, L)}[w(B) = x]$.
  In particular, if $(A, L)$ is Pfaffian with constant $c$, the coefficients of $\theta^x$ in $\pf A \Delta\pbig{\theta^w} \trsp{A}$ and of $\theta^{w(L)-x}$ in $\pf \Phi\pbig{\theta^w}$ are equal to $c^{-1}\card{\pbase_x}$ over $\setK$.
\end{proposition}

\subsection{Reducing Pfaffian Pairs to Pfaffian Parities}

Lawler~\cite{Lawler1976} presented the following reduction of the linear matroid intersection problem to the linear matroid parity problem.
Let $(A_1, A_2)$ be an $r \times n$ matrix pair with common column set $E$.
We define a $2r \times 2n$ matrix $A$ as follows: we associate each two columns of $A$ with $j \in E$ and set the $2r \times 2$ submatrix associated with $j$ as $\begin{psmallmatrix} A_1[\set{j}] & 0 \\ 0 & A_2[\set{j}] \end{psmallmatrix}$.
Through this association, we regard $E$ as the set of lines of $A$.
Then $B \subseteq E$ is a common base of $(A_1, A_2)$ if and only if $B$ is a parity base of $(A, E)$~\cite[Chapter~9.2]{Lawler1976}.
We show that when $(A_1, A_2)$ is Pfaffian, so is $(A, E)$.

\begin{proposition}
  Let $(A_1, A_2)$ be an $r \times n$ matrix pair and $(A, E)$ the $2r \times 2n$ matroid parity defined above.
  If $(A_1, A_2)$ is Pfaffian with constant $c$, then $(A, E)$ is Pfaffian with constant $\prn{-1}^{\frac{r(r-1)}{2}}c$.
\end{proposition}
\begin{proof}
  Let $B \subseteq E$ be a common base of $(A_1, A_2)$ as well as a parity base of $(A, E)$.
  By an appropriate column permutation, $A[B]$ is transformed into $\begin{psmallmatrix} A_1[B] & O \\ O & A_2[B] \end{psmallmatrix}$, whose the determinant is $\det A_1[B] \det A_2[B] = c$.
  The sign of this column permutation is $\prn{-1}^{1 + \cdots + (r-1)} = \prn{-1}^{\frac{r(r-1)}{2}}$.
  Hence the claim holds.
\end{proof}

\section{Examples}\label{sec:examples}

In this section, we enumerate discrete structures that can be represented as common bases of Pfaffian pairs or parity bases of Pfaffian parities.

\subsection{Regular Matroids and Spanning Trees}\label{sec:regular_matroids}
A matroid is called \emph{regular} if it is represented by a totally unimodular matrix, or equivalently, it is representable by a matrix over any field.
If $A$ is a totally unimodular matrix, a pair $(A, A)$ is clearly Pfaffian with constant 1.
Hence, as observed by Webb~\cite[Section~3.5]{Webb2004}, the number of bases of $A$ is equal to $\det A\trsp{A}$ by $\base(A) = \cbase(A, A)$ and \cref{prop:counting_formula_for_pfaffian_pairs}.
This well-known formula on regular matroids was initially indicated by Maurer~\cite{Maurer1976}.

Regular matroids typically arise from graphs.
Let $G = (V, E)$ be a connected undirected graph and $\mathcal{B}(G) \subseteq 2^E$ denote the set of all spanning trees of $G$.
Then $\matroid(G) \defeq (E, \mathcal{B}(G))$ forms a matroid, called the \emph{graphic matroid} of $G$.
Consider any orientation $\vec{G} = (V, \vec{E})$ of $G$.
Throughout this paper, we denote the directed edge in $\vec{E}$ corresponding to $e \in E$ by $\vec{e}$ and the directed edge set corresponding to $F \subseteq E$ by $\vec{F}$.
We define the \emph{incidence matrix} $A = \prn{A_{v,e}}_{v \in V, e \in E}$ of $\vec{G}$ as a matrix over $\setQ$ by
\begin{align}\label{def:incidence_matrix}
  A_{v,e} \defeq \begin{cases*}
    +1 & ($v = \tail{\vec{e}}$), \\
    -1 & ($v = \head{\vec{e}}$), \\
    \phantom{+}0 & (otherwise)
  \end{cases*}
\end{align}
for $v \in V$ and $e \in E$, where $\tail{\vec{e}}$ and $\head{\vec{e}}$ denote the tail and the head of $\vec{e}$, respectively.
The incidence matrix $A$ is known to be totally unimodular.
Let $A^{(r)}$ denote the submatrix of $A^{(r)}$ obtained by removing the $r$th row of $A$ for $r \in V$.
Then $A^{(r)}$ represents $\base(G)$, i.e., $\base(G) = \base\pbig{A^{(r)}}$.
Hence the number of spanning trees of $G$ is equal to $\det A^{(r)}\trsp{{A^{(r)}}}$, which is the $(r, r)$th cofactor of the \emph{Laplacian matrix} $A\trsp{A}$ of $G$.
This is exactly Kirchhoff's matrix-tree theorem~\cite{Kirchhoff1847}.
Refer to~\cite{Oxley2011} for details of regular and graphic matroids

\subsection{Regular Delta-Matroids and Euler Tours in 4-Regular Directed Graphs}\label{sec:regular_delta_matroids}
Let $S \in \setQ^{n \times n}$ be a skew-symmetric matrix whose rows and columns are indexed by a finite set $E$.
We also assume that $S$ is \emph{principally unimodular}; that is, any principal minor of $S$ is in $\set{+1, 0, -1}$~\cite{Bouchet1998}.
Since $S$ is skew-symmetric, all the principal minors of $S$ must be $0$ or $+1$.
Define
\begin{align}
  \feas(S) \defeq \set{F \subseteq E}[\text{$S[F, F]$ is nonsingular}]
\end{align}
and denote $(E, \feas(S))$ by $\dmatroid(S)$.
For $X \subseteq E$, we let $\dmatroid(S) \symdif X \defeq (E, \feas(S) \symdif X)$ with
\begin{align}
  \feas(S) \symdif X \defeq \set{F \symdif X}[F \in \feas(S)],
\end{align}
where $F \symdif X$ means the \emph{symmetric difference} of $F$ and $X$, that is, $F \symdif X \defeq (F \setminus X) \cup (X \setminus F)$.
Then $\dmatroid(S) \symdif X$ is called the \emph{regular delta-matroid} represented by $S$ (and $X$)~\cite{Bouchet1995,Geelen1995}.
Elements in $\feas(S) \symdif X$ are called \emph{feasible sets} of $\dmatroid(S) \symdif X$.
Regular delta-matroids are a generalization of regular matroids; see~\cite{Bouchet1998}.

Webb~\cite[Section~3.5]{Webb2004} indicated that the set of nonsingular principal submatrices of a skew-symmetric totally unimodular matrix can be represented by a Pfaffian pair.
This can be slightly generalized to the feasible sets of a regular delta-matroid as follows.
Define matrices $A_1 \defeq \begin{pmatrix} S & I_n\end{pmatrix}$ and $A_2 \defeq \begin{pmatrix} I_n & I_n \end{pmatrix}$ with common column set $E \cup \overline{E}$, where $I_n$ is the identity matrix of order $n = \card{E}$ and $\overline{E}$ is a disjoint copy of $E$ corresponding to the right blocks of $A_1$ and $A_2$.
Note that $A_1$ is not necessarily totally unimodular.
For $J \subseteq E$, denote by $\overline{J}$ the corresponding subset of $\overline{E}$ to $J$.

\begin{proposition}[{see~\cite[Section~3.5]{Webb2004}}]\label{prop:regular_delta_matroid_is_pfaffian_pair}
  The matrix pair $(A_1, A_2)$ is Pfaffian with constant 1.
  In addition, there is a one-to-one correspondence between common bases of $(A_1, A_2)$ and feasible sets of $\dmatroid(S) \symdif X$ given by $B \mapsto (B \cap E) \symdif X$ for $B \in \cbase(A_1, A_2)$.
\end{proposition}

\begin{proof}
  We first show that $(A_1, A_2)$ is Pfaffian.
  Note that $J \subseteq E \cup \overline{E}$ with $\card{J} = n$ is a base of $A_2$ if and only if $\overline{E \setminus J} = J \cap \overline{E}$.
  Taking such a column subset $J$, put $T_1 \defeq A_1[J]$ and $T_2 \defeq A_2[J]$.
  By a row permutation on $T_1$ and $T_2$, we transform $T_1$ and $T_2$ to
  \begin{align}
    T_1 = \begin{pmatrix}
      S[J \cap E, J \cap E]      & O       \\
      S[E \setminus J, J \cap E] & I_{n-k}
    \end{pmatrix},
    \quad
    T_2 = \begin{pmatrix}
      I_k & O       \\
      O   & I_{n-k}
    \end{pmatrix},
  \end{align}
  where $k \defeq \card{J \cap E}$.
  Note that $\det T_1 \det T_2$ does not change since the same row permutation is performed on both $T_1$ and $T_2$.
  Now we have
  \begin{align}\label{eq:regular_delta_matroid_as_pfaffian_pair}
    \det T_1 \det T_2 = \det S[J \cap E, J \cap E].
  \end{align}
  Since $S$ is skew symmetric and principally unimodular,~\eqref{eq:regular_delta_matroid_as_pfaffian_pair} is either 0 or 1.
  Hence $(A_1, A_2)$ is Pfaffian with constant 1.

  The equation~\eqref{eq:regular_delta_matroid_as_pfaffian_pair} also implies that $B \subseteq E \cup \overline{E}$ is a common base of $(A_1, A_2)$ if and only if $\overline{E \setminus B} = B \cap \overline{E}$ and $B \cap E \in \mathcal{F}(S)$.
  Hence $B \in \cbase(A_1, A_2)$ corresponds to $B \cap E \in \feas(S)$ one-to-one.
  The latter part of the proposition is obtained by taking the symmetric difference with $X$.
\end{proof}

\Cref{prop:regular_delta_matroid_is_pfaffian_pair} yields the following corollary.

\begin{corollary}
  The number of feasible sets of a regular delta-matroid $\dmatroid(S) \symdif X$ is equal to $\det (S+I_n)$.
\end{corollary}

Taking the symmetric difference with $X$ does not affect the number of feasible sets.
However, this changes the correspondence between elements of $E$ and columns of $(A_1, A_2)$.
Namely, labeling $t_j$ to an element $j \in E$ in $\dmatroid(S) \symdif X$ is equivalent to labeling $t_j$ in $(A_1, A_2)$ to $j \in E$ if $j \notin X$ and to $\overline{j} \in \overline{E}$ if $j \in X$.
Note this fact when applying the formula~\eqref{eq:symbolic_counting_pairs} to regular delta-matroids.

A combinatorial example of regular delta-matroids was given by Bouchet~\cite{Bouchet1995} as follows.
Let $G = (V, E)$ be a directed 4-regular Eulerian graph; that is, $G$ is strongly connected, and every vertex of $G$ is of in- and out-degree two.
A (directed) \emph{Euler tour} of $G$ is a tour that traverses every edge exactly once.
Any Euler tour $T$ of $G$ visits every vertex exactly twice as $G$ is 4-regular.
For each vertex $v \in V$ with incoming edges $e_1, e_2 \in E$ and outgoing edges $e_3, e_4 \in E$, there are exactly two possibilities of the way to visit $v$ twice; that is, $T$ traverses $e_3$ just after $e_1$ and $e_4$ just after $e_2$, or $e_4$ just after $e_1$ and $e_3$ just after $e_2$.
Therefore, fixing an Euler tour $U$ of $G$, we can represent every Euler tour $T$ of $G$ as a vertex subset $F_U(T) \subseteq V$ defined as follows: $v \in V$ is in $F_U(T)$ if and only if $T$ visits $v \in V$ in the different way from $U$.
The map $T \mapsto F_U(T)$ is injective.
Define $\dmatroid_U(G) \defeq (V, \mathcal{F}_U(G))$ with
$
  \mathcal{F}_U(G) \defeq \set{F_U(T)}[\text{$T$ is an Euler tour of $G$}].
$

For each vertex $v \in V$, we label the two directed edges leaving $v$ as $e_v^+$ and $e_v^-$.
Define a skew-symmetric matrix $S^U = \pbig{S^U_{u,v}}_{u,v \in V}$ over $\setQ$ as
\begin{align} \label{def:A_U}
  S^U_{u,v} \defeq \begin{cases*}
    +1 & ($U$ traverses edges in the order of $\cdots e_u^+ \cdots e_v^+ \cdots e_u^- \cdots e_v^- \cdots$), \\ % chktex 11
    -1 & ($U$ traverses edges in the order of $\cdots e_u^+ \cdots e_v^- \cdots e_u^- \cdots e_v^+ \cdots$), \\ % chktex 11
    \phantom{+}0 & (otherwise)
  \end{cases*}
\end{align}
for $u, v \in V$.
Then $S^U$ is principally unimodular~\cite[Theorem~11]{Bouchet1995} and $\dmatroid_U(G)$ coincides with $\dmatroid\pbig{S^U}$~\cite[Corollary~12]{Bouchet1995}.
Hence the set of Euler tours in a 4-regular directed graph can be represented as  common bases of a Pfaffian pair through \cref{prop:regular_delta_matroid_is_pfaffian_pair}.

\begin{remark}
  For an arbitrary directed graph $G = (V, E)$ each of whose vertex has the same in- and out-degree, there exists a formula, so-called the \emph{BEST theorem}, to count the number of Euler tours in $G$ (see, e.g.,~\cite[Theorem~6.36]{Cai2017}).
  This theorem states that the number of Euler tours in $G$ is
  \begin{align} \label{eq:BEST}
    T\prod_{v \in V} (d_v - 1)!,
  \end{align}
  where $d_v$ is the in-degree ($=\text{out-degree}$) of $v \in V$ and $T$ is the number of $r$-arborescences of $G$ with arbitrary root $r \in V$.

  In the case where $G$ is 4-regular, the BEST theorem~\eqref{eq:BEST} claims that the number of Euler tours in $G$ is equal to $T$, which can be computed by the directed matrix-tree theorem~\cite{Tutte1948}.
  Hence the Pfaffian-pair representation of Euler tours in 4-regular directed graph might seem useless.
  Nevertheless, this representation is needed when we apply the formula~\eqref{eq:symbolic_counting_pairs} that includes a variable $z$ because the corresponding ``symbolic'' version of the BEST theorem is yet unknown.
\end{remark}

\subsection{Arborescences}\label{sec:arborescences}
Let $G = (V, E)$ be a directed graph and take a vertex $r \in V$.
An $r$-\emph{arborescence}, or a \emph{directed tree} rooted at $r$, of $G$ is an edge subset $F \subseteq E$ satisfying the following:
\begin{enumerate}[label={(A\arabic*)}]
  \item $F$ is a spanning tree if the orientation is ignored.\label{item:A1}
  \item The in-degree of every $v \in V \setminus \set{r}$ is exactly one in $F$.\label{item:A2}
\end{enumerate}

It is well-known that $r$-arborescences can be represented as common bases of a matrix pair.
Let $A$ be the incidence matrix~\eqref{def:incidence_matrix} of $G$ and $R = \prn{R_{v,e}}_{v \in V, e \in E}$ a matrix over $\setQ$ defined by
\begin{align}
  R_{v,e} \defeq \begin{cases*}
    -1 & ($v = \partial^- e$), \\
    \phantom{+}0  & (otherwise)
  \end{cases*}
\end{align}
for $v \in V$ and $e \in E$.
The matrix $R$ is totally unimodular since each column has at most one nonzero entry.
Matroids represented by such matrices are called \emph{partition matroids}.
Put $A_1 \defeq A^{(r)}$ and $A_2 \defeq R^{(r)}$.
Then $B \subseteq E$ is a base of $A_1$ and $A_2$ if and only if $B$ satisfies~\ref{item:A1} and~\ref{item:A2}, respectively.
Hence common bases of $(A_1, A_2)$ correspond to $r$-arborescences of $G$.

The matrix $L \defeq A\trsp{R}$ is called the (directed) \emph{Laplacian matrix} of $G$.
The directed matrix-tree theorem due to Tutte~\cite{Tutte1948} claims that the $(r,r)$th cofactor of $L$, which is equal to $\det A_1 \trsp{A_2}$, coincides with the number of $r$-arborescences of $G$.
This implies:
\begin{proposition}\label{prop:arborescence_pfaffian_pair}
  The pair $(A_1, A_2)$ is Pfaffian with constant 1.
\end{proposition}
\begin{proof}
  Since both $A_1$ and $A_2$ are totally unimodular, $\det A_1[B] \det A_2[B]$ is $\pm1$ for all $B \in \cbase(A_1, A_2)$.
  If $(A_1, A_2)$ is not Pfaffian, $\det A_1 \trsp{A_2}$ is less than $\cbase(A_1, A_2)$ due to cancellations in the right-hand side of~\eqref{eq:cauchy_binet}.
  This contradicts the statement of the directed matrix-tree theorem.
\end{proof}

We can directly show that $(A_1, A_2)$ is Pfaffian without the directed matrix-tree theorem; see the proof of~\cite[Theorem~6.35]{Cai2017} for example.
Such a proof justifies the directed matrix-tree theorem via the Cauchy--Binet formula~\eqref{eq:cauchy_binet} the other way around.

\subsection{Perfect Matchings of Pfaffian Graphs}\label{sec:perfect_matchings}
A \emph{matching} of an undirected graph $G$ is an edge subset $M$ such that no two disjoint edges in $M$ share the same end.
We also define a matching for a directed graph by ignoring its orientation.
A matching $M$ is said to be \emph{perfect} if every vertex of $G$ is covered by some edge in $G$.
Matching theory has two faces depending on whether $G$ is bipartite or general.

First, let $G = (U \cup V, E)$ be a simple undirected bipartite graph.
The vertex set of $G$ is bipartitioned as $\set{U, V}$ with $n \defeq \card{U} = \card{V}$ and all edges are between $U$ and $V$.
We define totally unimodular matrices $A_U = \pbig{A^U_{u,e}}_{u \in U, e \in E}$ and $A_V = \pbig{A^V_{v,e}}_{v \in V, e \in E}$ as
\begin{align}\label{def:bipartite_matching_A1_A2}
  A^U_{u,e} \defeq \begin{cases}
    +1           & (u \in e),          \\
    \phantom{+}0 & (\text{otherwise}),
  \end{cases}
  \quad
  A^V_{v,e} \defeq \begin{cases}
    +1           & (v \in e),         \\
    \phantom{+}0 & (\text{otherwise})
  \end{cases}
\end{align}
for $u \in U, v \in V$ and $e \in E$.
Note that both $\matroid(A_U)$ and $\matroid(A_V)$ are partition matroids.
Then $M \subseteq E$ is a perfect matching of $G$ if and only if $M \in \cbase(A_U, A_V)$.

Suppose that vertices in $U$ and $V$ are ordered as $u_1, \ldots, u_n$ and $v_1, \ldots, v_n$.
For $i \in \intset{n}$, the $i$th rows of $A_U$ and $A_V$ are associated with $u_i$ and $v_i$, respectively.
A perfect matching $M$ of $G$ uniquely corresponds to a permutation $\sigma \in \sym_n$ on $\intset{n}$ such that $\set[\big]{u_i, v_{\sigma(i)}} \in M$ for all $i \in \intset{n}$.
Denote this permutation by $\sigma_M$.
We define the \emph{sign} of $M$ (with respect to the current ordering of $U$ and $V$) as $\sgn M \defeq \sgn \sigma_M$.

Let $z = \prn{z_e}_{e \in E}$ be a vector of distinct indeterminates indexed by $E$.
The matrix $A_1 D(z) \trsp{A_2}$ is called the \emph{Edmonds matrix} of $G$.
Its $(i,j)$th entry is $z_e$ if $e = \set{u_i, v_j} \in E$ and 0 otherwise for $i, j \in \intset{n}$.
By the definition~\eqref{def:determinant} of the determinant, it holds
\begin{align}\label{eq:edmonds_1}
  \det A_1 D(z) \trsp{A_2} = \sum_{M \in \cbase(A_1, A_2)} \sgn M \prod_{e \in M} z_e.
\end{align}
On the other hand, by~\eqref{eq:cauchy_binet_1}, we have
\begin{align}\label{eq:edmonds_2}
  \det A_1 D(z) \trsp{A_2} = \sum_{M \in \cbase(A_1, A_2)} \det A_1[M] \det A_2[M] \prod_{e \in M} z_e.
\end{align}
Comparing the coefficients of~\eqref{eq:edmonds_1} and~\eqref{eq:edmonds_2}, we have the following.

\begin{lemma}\label{lem:bipartite_matching_sign}
  The sign of a perfect matching $M$ of $G$ is equal to $\det A_1[M] \det A_2[M]$.
\end{lemma}

Consider an orientation $\vec{G} = (U \cup V, \vec{E})$ of $G$.
We define a vector $s = \prn{s_e}_{e \in E}$ indexed by $E$ as
\begin{align}
  s_e \defeq \begin{cases}
    +1 & (\vec{e} = (u, v)), \\
    -1 & (\vec{e} = (v, u))
  \end{cases}
\end{align}
for $e = \set{u, v} \in E$ with $u \in U$ and $v \in V$.
Put $\vec{A}_1 \defeq A_1 D(s)$ and $\vec{A}_2 \defeq A_2$.
Namely, $\vec{A}_1$ is a matrix obtained from $A_1$ by reversing the sign of every column corresponding to an edge from $V$ to $U$.
Note that $\cbase(A_1, A_2) = \cbase(\vec{A}_1, \vec{A}_2)$.
The matrix $N = \pbig{N_{i,j}}_{i,j \in \intset{n}} \defeq \vec{A}_1 \trsp{{\vec{A}_2}} = A_1 D(s) \trsp{A}_2$ is called the (\emph{directed}) \emph{bipartite adjacency matrix} of $\vec{G}$.
Its $(i,j)$th entry $N_{i,j}$ is
\begin{align}\label{eq:bipartite_adjacency_matrix}
  N_{i,j} = \begin{cases}
              +1 & ((u_i, v_j) \in \vec{E}), \\
              -1 & ((v_j, u_i) \in \vec{E}), \\
    \phantom{+}0 & (\text{otherwise})
  \end{cases}
\end{align}
for $i, j \in \intset{n}$.
Recall that $\vec{M}$ is also called a matching of $\vec{G}$ for a matching $M$ of $G$.
We define the \emph{sign} of a perfect matching $\vec{M}$ of $\vec{G}$ as
\begin{align}\label{def:sign_of_directed_perfect_bipartite_matching}
  \sgn \vec{M}
  \defeq \sgn M \prod_{e \in M} s_e
  = \sgn M \prod_{i=1}^n N_{i, \sigma_M(i)}
  \in \set{+1, -1}.
\end{align}

\begin{lemma}\label{lem:bipartite_matching_sign_directed}
  The sign of a perfect matching $\vec{M}$ of $\vec{G}$ is equal to $\det \vec{A}_1[M] \det \vec{A}_2[M]$.
\end{lemma}
\begin{proof}
  By \cref{lem:bipartite_matching_sign}, we have $\sgn M = \det A_1[M] \det A_2[M]$.
  We also have $\det \vec{A}_1[M] = \det A_1[M] \prod_{e \in M} s_e$ and $\det \vec{A}_2[M] = \det A_2[M]$ by the definitions of $\vec{A}_1$ and $\vec{A}_2$.
  Hence the claim holds.
\end{proof}

An orientation $\vec{G}$ of $G$ is called \emph{Pfaffian} if the signs of all perfect matchings of $\vec{G}$ are the same~\cite{Kasteleyn1961,Temperley1961}.
The following proposition, which was observed by Webb~\cite{Webb2004}, holds from \cref{lem:bipartite_matching_sign_directed}.

\begin{theorem}[{\cite[Observation~3.7]{Webb2004}}]\label{thm:bipartite_matching_pfaffian}
  Let $G$ be a bipartite graph, $\vec{G}$ an orientation of $G$, and $(\vec{A}_1, \vec{A}_2)$ the matrix pair defined above from $\vec{G}$.
  Then $\cbase(\vec{A}_1, \vec{A}_2)$ coincides with the set of perfect matchings of $G$.
  In addition, if $\vec{G}$ is Pfaffian, $(\vec{A}_1, \vec{A}_2)$ is also Pfaffian with constant $\sgn \vec{M}$, where $M$ is an arbitrary perfect matching of $G$.
\end{theorem}

We extend the above arguments to nonbipartite graphs.
Let $G = (V, E)$ be a simple undirected graph that is not necessarily bipartite.
Suppose that $\card{V}$ is even and vertices are ordered as $v_1, \ldots, v_{2n}$.
We define a totally unimodular matrix $A \in \setR^{\card{V} \times 2\card{E}}$ as follows: each row is indexed by $v \in V$ and each two columns are associated with an edge $e \in E$.
%Denote by $e^+$ and $e^-$ the column indices of $A$ associated with $e \in E$.
For $v \in V$ and $e = \set{v_i, v_j} \in E$ with $i < j$, the corresponding $1 \times 2$ submatrix of $A$ to $v$ and $e$ is defined to be $\begin{pmatrix} +1 & 0 \end{pmatrix}$ if $v = v_i$, $\begin{pmatrix} 0 & +1 \end{pmatrix}$ if $v = v_j$ and $O$ otherwise.
We regard each $e \in E$ as a line of $A$.
Then $M \subseteq E$ is a perfect matching of $G$ if and only if $M \in \pbase(A, E)$~\cite[Chapter~9.2]{Lawler1976}.

Recall that $F_{2n}$ is the subset of $\sym_{2n}$ defined by~\eqref{def:F}.
A perfect matching $M$ of $G$ uniquely corresponds to a permutation $\sigma \in F_{2n}$ such that $\set[\big]{v_{\sigma(2i-1)}, v_{\sigma(2i)}} \in M$ for all $i \in \intset{n}$.
Denote this permutation by $\sigma_M$.
We define the \emph{sign} of $M$ as $\sgn M \defeq \sgn \sigma_M$.

Let $z = \prn{z_e}_{e \in E}$ be a vector of distinct indeterminates indexed by $E$.
The skew-symmetric matrix $A \Delta(z) \trsp{A}$ is called the \emph{Tutte matrix} of $G$.
Its $(i,j)$th entry is equal to $z_e$ if $e = \set{u_i, v_j} \in E$ and $i < j$, to $-z_e$ if $e = \set{u_i, v_j} \in E$ and $i > j$ and to 0 otherwise for $i, j \in \intset{n}$.
By the definition~\eqref{def:pfaffian} of the Pfaffian, it holds
\begin{align}\label{eq:tutte_1}
  \pf A \Delta(z) \trsp{A} = \sum_{M \in \pbase(A, E)} \sgn M \prod_{e \in M} z_e.
\end{align}
We also have
\begin{align}\label{eq:tutte_2}
  \pf A \Delta(z) \trsp{A} = \sum_{M \in \pbase(A, E)} \det A[M] \prod_{e \in M} z_e.
\end{align}
by~\eqref{eq:pfaffian_cauchy_binet_1}.
Hence the following holds as an extension of \cref{lem:bipartite_matching_sign}.

\begin{lemma}\label{lem:matching_sign}
  The sign of a perfect matching $M$ of $G$ is equal to $\det A[M]$.
\end{lemma}

In the same way as the bipartite case, we next consider an orientation $\vec{G} = (V, \vec{E})$ of $G$.
Define a vector $s = \prn{s_e}_{e \in E}$ indexed by $E$ as follows: for each $\vec{e} = (v_i, v_j) \in \vec{E}$, we set
\begin{align}
  s_e \defeq \begin{cases}
    +1 & (i < j), \\
    -1 & (i > j).
  \end{cases}
\end{align}
We also construct a symmetric block diagonal matrix $X = \diag\prn{X_e}_{e \in E}$, where $X_e$ is a $2 \times 2$ matrix defined by $X_e \defeq I_2$ if $s_e = +1$ and $X_e \defeq \begin{psmallmatrix} \phantom{+}0 & +1 \\ +1 & \phantom{+}0 \end{psmallmatrix}$ if $s_e = -1$ for $e \in E$.
Put $\vec{A} \defeq AX$, i.e., $\vec{A}$ is obtained from $A$ by interchanging two columns associated with each $(v_i, v_j) \in \vec{E}$ with $i > j$.
Note that $X\Delta(\onevec)X = \Delta(s)$ and $\pbase(A, E) = \pbase(\vec{A}, E)$ hold.
The skew-symmetric matrix $S = \pbig{S_{i,j}}_{i,j \in \intset{2n}} \defeq \vec{A} \Delta(\onevec) \trsp{\vec{A}} = A \Delta(s) \trsp{A}$ is called the (\emph{directed}) \emph{skew-symmetric adjacency matrix} of $\vec{G}$.
It can be confirmed that
\begin{align}\label{eq:skew_symmetric_adjacency_matrix}
  S_{i,j} = \begin{cases}
              +1 & ((v_i, v_j) \in \vec{E}), \\
              -1 & ((v_j, v_i) \in \vec{E}), \\
    \phantom{+}0 & (\text{otherwise})
  \end{cases}
\end{align}
holds for $i, j \in \intset{2n}$.
For a perfect matching $M$ of $G$, it holds $s_e = S_{\sigma_M(2i-1), \sigma_M(2i)}$ for every $e = \set[\big]{v_{\sigma_M(2i-1)}, v_{\sigma_M(2i)}} \in M$ since $\sigma_M(2i-1) < \sigma_M(2i)$.
We define the \emph{sign} of a perfect matching $\vec{M}$ of $\vec{G}$ as
\begin{align}
  \sgn \vec{M}
  \defeq \sgn M \prod_{e \in M} s_e
  = \sgn M \prod_{i=1}^n S_{\sigma_M(2i-1), \sigma_M(2i)}
  \in \set{+1, -1}.
\end{align}

\begin{lemma}\label{lem:matching_sign_directed}
  The sign of a perfect matching $\vec{M}$ of $\vec{G}$ is equal to $\det \vec{A}[M]$.
\end{lemma}
\begin{proof}
  The claim follows from $\det \vec{A}[M] = \det A[M] \prod_{e \in M} s_e$ and $\sgn M = \det A[M]$ by \cref{lem:matching_sign}.
\end{proof}

An orientation $\vec{G}$ of $G$ is also called \emph{Pfaffian} if the signs of all perfect matchings of $\vec{G}$ are constant.
The following holds from \cref{lem:matching_sign_directed} as a generalization of \cref{thm:bipartite_matching_pfaffian}.

\begin{theorem}\label{thm:matching_pfaffian}
  Let $G = (V, E)$ be a graph, $\vec{G}$ an orientation of $G$ and $\vec{A}$ the matrix defined above from $\vec{G}$.
  Then $\pbase(\vec{A}, E)$ coincides with the set of perfect matchings of $G$.
  In addition, if $\vec{G}$ is Pfaffian, $(\vec{A}, E)$ is also Pfaffian with constant $\sgn \vec{M}$, where $M$ is an arbitrary perfect matching of $G$.
\end{theorem}

\subsection{Spanning Hypertrees of 3-Pfaffian 3-Uniform Hypergraphs}\label{sec:spanning_hypertrees}

We introduce basic notions of hypergraphs.
A \emph{3-graph}, or a \emph{3-uniform hypergraph} is a pair $H = (V, \mathcal{E})$, where $V$ is a finite set and $\mathcal{E} \subseteq \binom{V}{3}$ is a subset of unordered triples of elements in $V$.
Elements in $V$ and $\mathcal{E}$ are called \emph{vertices} and \emph{hyperedges} of $H$, respectively.
An $s$--$t$ \emph{path} of $H$ is an alternating sequence $s = v_0, e_1, v_1, \ldots, e_l, v_l = t$ of disjoint vertices $v_0, \ldots, v_l \subseteq V$ and disjoint hyperedges $e_1, \ldots, e_l \subseteq \mathcal{E}$ such that $\set{v_{i-1}, v_i} \subseteq e_i$ for $i \in \intset{l}$.
A \emph{cycle} is an alternating sequence satisfying the above condition with $v_0 = v_l$ and $l \ge 2$.
We regard paths and cycles as hyperedge subsets.
A hyperedge subset $F \subseteq \mathcal{E}$ is called \emph{connected} if there exists a path in $F$ between any two vertices contained in some hyperedge in $F$.
A \emph{hyperforest} of $H$ is a hyperedge subset having no cycles, and a \emph{hypertree} is a connected hyperforest.
A hypertree $T$ is said to be \emph{spanning} if any vertex of $H$ is contained in some hyperedge in $T$.
It holds $\card{V} = 2\card{T} + 1$ for any spanning hypertree $T$.
Namely, if $H$ has a spanning hypertree, $\card{V}$ must be odd.

Lovász~\cite{Lovasz1980} showed that one can find a spanning hypertree of a 3-graph by solving the linear (graphic) matroid parity problem as follows.
Let $H = (V, \mathcal{E})$ be a 3-graph with $\card{V}$ being odd.
We construct a multiple graph $G = (V, E)$ as follows: starting from $E = \varnothing$, for each hyperedge $\set{u, v, w} \in \mathcal{E}$, we arbitrary add two edges out of $\set{u, v}, \set{v, w}, \set{w, u}$ to $E$.
Now each hyperedge in $\mathcal{E}$ is associated with two edges in $E$.
By this association, we regard each hyperedge in $\mathcal{E}$ as a line of $E$.
It is easily checked that $T \subseteq \mathcal{E}$ is a spanning hypertree of $H$ if and only if the edge subset of $E$ corresponding to $T$ is a spanning tree of $G$.
Hence spanning hypertrees of $H$ correspond to parity bases of $\matroid(G)$ with respect to the line set $\mathcal{E}$.

Masbaum--Vaintrob~\cite{Masbaum2002} established the \emph{Pfaffian matrix-tree theorem} for enumerating hypertrees of a 3-graph.
To describe the theorem, we orient each hyperedge and then define the \emph{sign} of a spanning hypertree based on the description of~\cite{Goodall2011}.
Let $H = (V, \mathcal{E})$ be a 3-graph and suppose that vertices are ordered as $v_1, \ldots, v_{2n+1}$, where $\card{V} = 2n+1$.
Let $\vec{H} = (V, \vec{\mathcal{E}})$ be an orientation of $H$, i.e., $\vec{\mathcal{E}}$ is the set of ordered triples.
Elements in $\vec{\mathcal{E}}$ are called \emph{directed hyperedges}.
As in the orientation of graphs, we denote the directed hyperedge corresponding to $e \in \mathcal{E}$ by $\vec{e}$ and the set of directed hyperedges corresponding  to $F \subseteq \mathcal{E}$ by $\vec{F}$.
For each directed hyperedge $\vec{e} = (v_i, v_j, v_k) \in \vec{\mathcal{E}}$, we associate a cyclic permutation $\sigma_{\vec{e}} \defeq (i\; j\; k)$ with $\vec{e}$.
Arbitrary fixing an ordering of hyperedges in a spanning hypertree $T$, we put
\begin{align}\label{eq:prod_sym}
  \sigma \defeq \prod_{e \in T} \sigma_{\vec{e}}.
\end{align}
Then $\sigma$ forms a cyclic permutation $(s_1\; s_2\; \cdots\; s_{2n+1})$ over $\intset{2n+1}$ as shown in~\cite[Proposition~3.4]{Masbaum2002}.
Let $s \in \sym_{2n+1}$ be a permutation defined by $s(i) \defeq s_i$ for $i \in \intset{2n+1}$.
We define the \emph{sign} of $\vec{T}$ as $\sgn \vec{T} \defeq \sgn s$.
While $\sigma$ and $s$ depend on the order of the product in~\eqref{eq:prod_sym}, $\sgn \vec{T}$ is well-defined~\cite[Proposition~3.8]{Masbaum2002}.

Now we are ready to describe the Pfaffian matrix-tree theorem.
Let $z = \prn{z_e}_{e \in \mathcal{E}}$ be a vector of distinct indeterminates indexed by $\mathcal{E}$.
For distinct $i,j,k \in \intset{2n+1}$, let $m_{ijk}$ be the inversion number of $(i,j,k)$ and put $\epsilon_{ijk} \defeq \prn{-1}^{m_{ijk}}$.
Namely, $\epsilon_{ijk} = +1$ if $i < j < k$, $j < k < i$, or $k < i < j$ and $\epsilon_{ijk} = -1$ otherwise.
We define a skew-symmetric matrix $\Lambda(z) = \prn{\Lambda_{i,j}(z)}_{i,j \in \intset{2n+1}}$ by
\begin{align}\label{def:Lambda-hypertree}
  \Lambda_{i,j}(z) \defeq
  \sum_{\condit{k \in \intset{2n+1}}[e = \set{v_i, v_j, v_k} \in \mathcal{E}]} \epsilon_{ijk} z_e.
\end{align}
For $r \in \intset{2n+1}$, let $\Lambda^{(r)}(z)$ denote the submatrix of $\Lambda(z)$ obtained by removing the $r$th row and column.

\begin{theorem}[{Pfaffian matrix-tree theorem~\cite[Theorem~5.3]{Masbaum2002}}]\label{thm:pfaffian_matrix_tree_theorem}
  In the above setting, it holds
  \begin{align}\label{eq:pfaffian_matrix_tree_theorem}
    \pf \Lambda^{(r)}(z) = \prn{-1}^{r-1} \sum_{T \in \mathcal{T}} \sgn \vec{T} \prod_{e \in T} z_e
  \end{align}
  for $r \in \intset{2n+1}$, where $\mathcal{T}$ is the set of all spanning hypertrees of $H$.
\end{theorem}

Recall the construction of the multiple undirected graph $G$ from $H$ described above.
We shall construct a multiple directed graph $\vec{G} = (V, \vec{E})$ reflecting the orientation $\vec{H}$ as follows: starting from $\vec{E} = \varnothing$, for each $(v_i, v_j, v_k) \in \vec{\mathcal{E}}$, we add $(v_i, v_j)$ and $(v_i, v_k)$ to $\vec{E}$.
Let $A$ be the incidence matrix~\eqref{def:incidence_matrix} of $\vec{G}$.
We arrange the rows and columns of $A$ so that the $i$th row corresponds to $v_i$ for $i \in \intset{2n+1}$, and in the two columns associated with each $(v_i, v_j, v_k) \in \vec{\mathcal{E}}$, the column indexed by $(v_i, v_j)$ is in the left of that by $(v_i, v_k)$.
Then the following holds.

\begin{lemma}\label{lem:AxA_Lambda}
  It holds $A \Delta(z) \trsp{A} = \Lambda(z)$.
\end{lemma}
\begin{proof}
  Take $\vec{e} = \prn{v_i, v_j, v_k} \in \vec{\mathcal{E}}$ and suppose that $i < j < k$.
  The $3 \times 3$ submatrix of $A \Delta(z) \trsp{A}$ indexed by $\set{i, j, k}$ is
  \begin{align}
    \begin{pmatrix}
      +1           & +1           \\
      -1           & \phantom{+}0 \\
      \phantom{+}0 & -1
    \end{pmatrix}
    \begin{pmatrix}
      0    & z_e \\
      -z_e & 0
    \end{pmatrix}
    \begin{pmatrix}
      +1 & -1           & \phantom{+}0 \\
      +1 & \phantom{+}0 & -1
    \end{pmatrix}
    =
    \begin{pmatrix}
      0    & +z_e & -z_e \\
      -z_e & 0    & +z_e \\
      +z_e & -z_e & 0
    \end{pmatrix}
    =
    \begin{pmatrix}
      0                 & \epsilon_{ijk} z_e & \epsilon_{ikj}z_e \\
      \epsilon_{jik}z_e & 0                  & \epsilon_{jki}z_e \\
      \epsilon_{kij}z_e & \epsilon_{kji}z_e  & 0
    \end{pmatrix},
  \end{align}
  which agrees with the contribution of $z_e$ in $\Lambda(z)$.
  Cases of other orderings of $i,j,k$ are the same.
\end{proof}

Recall that $A^{(r)}$ denotes the submatrix of $A$ obtained by removing the $r$th row for $r \in \intset{2n+1}$.
Note that $A^{(r)}$ represents $\matroid(G)$ and hence $\pbase\pbig{A^{(r)}, \mathcal{E}}$ coincides with the set of spanning hypertrees of $H$ one-to-one.
Using the Pfaffian matrix-tree theorem, we can associate $\sgn \vec{T}$ with $\det A^{(r)}[T]$ as follows.

\begin{lemma}\label{lem:3_pfaffian_sign}
  For a spanning hypertree $T$ of $H$ and $r \in \intset{2n+1}$, the sign of $\vec{T}$ is equal to $\prn{-1}^{r-1} \det A^{(r)}[T]$.
\end{lemma}
\begin{proof}
  By \cref{lem:AxA_Lambda} and~\eqref{eq:pfaffian_cauchy_binet_1}, we have
  \begin{align}\label{eq:3_pfaffian_sign_mid}
    \pf \Lambda^{(r)}(z)
    = \pf A^{(r)} \Delta(z) \trsp{{A^{(r)}}}
    = \sum_{T \in \pbase(A^{(r)}, \mathcal{E})} \det A^{(r)}[T] \prod_{e \in T} z_e.
  \end{align}
  Comparing the coefficients of~\eqref{eq:pfaffian_matrix_tree_theorem} and~\eqref{eq:3_pfaffian_sign_mid}, we obtain the claim.
\end{proof}

An orientation of a 3-graph $H$ is said to be \emph{3-Pfaffian} if $\sgn \vec{T}$ is constant for all spanning hypertree $T$ of $H$~\cite{Goodall2011}.
A 3-graph is called \emph{3-Pfaffian} if it admits a 3-Pfaffian orientation.
By \cref{lem:3_pfaffian_sign}, the following holds.

\begin{theorem}
  Let $H = (V, \mathcal{E})$ be a 3-graph, $\vec{H}$ an orientation of $H$ and $A$ the matrix defined above from $\vec{H}$.
  Then $\pbase\pbig{A^{(r)}, \mathcal{E}}$ corresponds to spanning hypertrees of $H$ one-to-one for $r \in \intset{\card{V}}$.
  In addition, when $\vec{H}$ is 3-Pfaffian, then $\pbig{A^{(r)}, \mathcal{E}}$ is Pfaffian with constant $\prn{-1}^{r-1} \sgn \vec{T}$, where $T$ is an arbitrary spanning hypertree of $H$.
\end{theorem}

The above proof of \cref{lem:3_pfaffian_sign} relies on the Pfaffian matrix-tree theorem.
In the following, we give an alternative proof of \cref{lem:3_pfaffian_sign} without the Pfaffian matrix-tree theorem.
It provides a new proof of the Pfaffian matrix-tree theorem via \cref{lem:AxA_Lambda} and~\eqref{eq:pfaffian_cauchy_binet_1}.

\begin{proof}[{of \cref{lem:3_pfaffian_sign}}]
  We regard $T$ as a hypertree rooted at $r$.
  We arrange hyperedges in $T$ as $e_1, \ldots, e_n$ so that $e_1$ contains $v_r$ and $|e_j \cap (e_1 \cup \cdots \cup e_{j-1})| = 1$ for every $2 \le j \le n$.
  Such an ordering is obtained by traversing $T$ in pre-order or breadth-first order.
  For $j \in \intset{n}$, we define $t_j, p_j, q_j \in \intset{2n+1}$ as follows.
  Put $t_1 \defeq r$ and take $p_1$ and $q_1$ so that $(t_1, p_1, q_1)$ is equal to $\vec{e}_1$ up to an even permutation.
  For $j \ge 1$, the hyperedge $\vec{e}_j$ consists of one vertex contained in $e_1 \cup \cdots \cup e_{j-1}$ and two vertices which do not occur in $e_1 \cup \cdots \cup e_{j-1}$.
  We denote the former vertex by $v_{t_j}$ and the latter two vertices by $v_{p_j}$ and $v_{q_j}$, where $(v_{t_j},v_{p_j},v_{q_j})$ is equal to $\vec{e}_j$ up to an even permutation.

  Permuting rows and columns, we transform $A^{(r)}[T]$ into a matrix $X$ as follows.
  The lines (columns) are arranged in the same order as hyperedges.
  The row permutation is performed so that $v_{p_j}$ corresponds to the $(2j-1)$st row and $v_{q_j}$ corresponds to the $2j$th row for $j \in \intset{n}$.
  Then $X$ is a block upper-triangular matrix whose the $j$th diagonal block is $\begin{psmallmatrix} -1 & \phantom{+}0 \\ \phantom{+}0 & -1 \end{psmallmatrix}$ if $\vec{e}_j = (v_{t_j}, v_{p_j}, v_{q_j})$, $\begin{psmallmatrix} +1 & +1 \\ -1 & \phantom{+}0 \end{psmallmatrix}$ if $\vec{e}_j = (v_{p_j}, v_{q_j}, v_{t_j})$, and $\begin{psmallmatrix} \phantom{+}0 & -1 \\ +1 & +1 \end{psmallmatrix}$ if $\vec{e}_j = (v_{q_j}, v_{t_j}, v_{p_j})$.
  The determinants of those three matrices are 1, and thus we have $\det X = 1$.
  The sign of the line permutation is 1 and the sign of the row permutation is $\prn{-1}^{r-1}\sgn \tau$, where
  \begin{align}
    \tau \defeq
    \begin{pmatrix}
      1  &   2 &   3 &   4 & \cdots & 2n-1 & 2n  & 2n+1 \\
      p_1& q_1 & p_2 & q_2 & \cdots & p_n  & q_n & r
    \end{pmatrix}.
  \end{align}
  Hence we have $\det A^{(r)}[T] = \prn{-1}^{r-1} \sgn \tau$.

  We next show $\sgn \tau = \sgn \vec{T}$, which implies the claim.
  For $i \in \intset{2n+1}$, let $T_i$ denote the \emph{subhypertree} of $T$ rooted at $v_i$, i.e., the set of hyperedges $e$ in $T$ such that the unique $v$--$v_r$ path in $T$ contains $v_i$ for every $v \in e$.
  Note that $T_i$ might be empty.
  %It holds $i \ne p_j$ and $i \ne q_j$ for every $j \in \intset{n}$ with $e_j \in T_i$.
  Put $R_i \defeq \set{j \in \intset{n}}[v_i \in e_j \in T_i]$.
  We recursively define a function $f: [2n+1] \rightarrow \sym_{2n+1}$ by
  \begin{align}
    f(i) \defeq \prod_{j \in R_i} f(q_j) (t_j\; p_j\; q_j) f(p_j).
  \end{align}
  Then $\sigma \defeq f(r)$ is the product of $(t_j\; p_j\ q_j)$ for $e_j \in T$ in some order.
  As we have mentioned in defining $\sgn \vec{T}$, the permutation $\sigma$ is cyclic.
  Suppose $\sigma = (s_1\;\cdots\;s_{2n+1})$.
  Note that $\set{s_1, \ldots ,s_{2n+1}} = V = \set{p_1, q_1, \ldots, p_n, q_n, r}$.
  We have $\sigma(p_j) = q_j$ for $j \in \intset{n}$ because every hyperedge containing $q_j$ appears the right of $(t_j\;p_j\;q_j)$ in the sequence of products of permutations and $(t_j\;p_j\;q_j)$ maps $p_j$ to $q_j$.
  This means that $(s_1, s_2, \ldots, s_{2n+1})$ is equal to $(p_1,q_1,p_2,\ldots,q_n,r)$ up to an even permutation.
  Hence $\sgn \vec{T}$ is equal to $\sgn \tau$ by the definition of $\sgn \vec{T}$.
\end{proof}

\begin{remark}
  As mentioned by Hirschiman--Reiner~\cite{Hirschman2004}, Kirchhoff's matrix-tree theorem has at least three different well-known proofs: one via the Cauchy-Binet formula~\eqref{eq:cauchy_binet}, one via deletion-contraction induction, and one via the sign-reversing involution.
  The original proof of the Pfaffian matrix-tree theorem by Masbaum--Vaintrob~\cite{Masbaum2002} is based on an analog of deletion-contraction induction.
  Hirschiman--Reiner~\cite{Hirschman2004} provided the second proof using the sign-reversing involution.
  Our proof above is the third one, which is analogous to the proof of Kirchhoff's matrix-tree theorem via the Cauchy-Binet formula.
\end{remark}

\section{Examples from Disjoint \calS-Path Problem}\label{sec:lgv}

Continued from \cref{sec:examples}, this section provides further examples of Pfaffian pairs and parities in the framework of Mader's disjoint \calS-path problem.
%In \cref{sec:dag}, we show that the known reduction of the disjoint $S$--$T$ path problem on a directed acyclic graph (DAG) to the bipartite matching problem retains the sign.
%This observation yield an example of Pfaffian pairs (bipartite graphs) as well as a new proof of the Lindström--Gessel--Viennot (LGV) lemma.
%\Cref{sec:st} generalizes this relation to the shortest disjoint $S$--$T$ path problem on an undirected graph.
%We further consider the $S$--$T$--$U$ path problem in \cref{sec:stu}.

\subsection{Disjoint \ST\ Paths on Directed Acyclic Graphs}\label{sec:dag}

Let $G = (V, E)$ be a directed acyclic graph (DAG) and take disjoint vertex subsets $S, T \subseteq V$ with $k \defeq \card{S} = \card{V}$.
Suppose that vertices in $S$ and $T$ are ordered as $s_1, \ldots, s_k$ and $t_1, \ldots, t_k$.
We assume that the in-degree of $s_i$ and the out-degree of $t_j$ are zero for all $i,j \in \intset{k}$.
We regard directed paths of $G$ as edge subsets.
A (\emph{directed}) \emph{$S$--$T$ path} $P \subseteq E$ of $G$ is the union of $k$ directed paths $P_1, \ldots, P_k$ of $G$ satisfying the following:
\begin{enumerate}[label={(P1)}]
  \item There exists a permutation $\sigma \in \sym_k$ of $\intset{k}$ such that every $P_i$ is a path from $s_i$ to $t_{\sigma(i)}$ for $i \in \intset{k}$.\label{item:P1}
\end{enumerate}
We call an $S$--$T$ path $P$ (\emph{vertex-})\emph{disjoint} if $P_i$ and $P_j$ have no common vertices for distinct $i, j \in \intset{k}$.
We denote the permutation in~\ref{item:P1} by $\sigma_P$.
Note that $\sigma_P$ is well-defined for disjoint $S$--$T$ paths.
We define the \emph{sign} of a disjoint $S$--$T$ path $P$ as $\sgn P \defeq \sgn \sigma_P$.

We introduce the Lindström--Gessel--Viennot (LGV) lemma, which was provided by Gessel--Viennot~\cite{Gessel1985} based on the work of Lindström~\cite{Lindstrom1973}.
Let $z = \prn{z_e}_{e \in E}$ be a vector of distinct indeterminates indexed by $E$.
Define a $k \times k$ matrix $\Omega(z) = \prn{\Omega_{i,j}(z)}_{i,j \in \intset{k}}$ by
\begin{align}\label{def:Omega_entry}
  \Omega_{i,j}(z) \defeq \sum_{P \in \mathcal{P}_{i,j}} \prod_{e \in P} z_e
\end{align}
for $i,j \in \intset{k}$, where $\mathcal{P}_{i,j}$ is the set of all disjoint $s_i$--$t_j$ paths of $G$.

\begin{lemma}[{LGV lemma~\cite{Gessel1985,Lindstrom1973}}]\label{lem:normal-lgv}
  It holds
  \begin{align}\label{eq:normal-lgv}
    \det \Omega(z) = \sum_{P \in \mathcal{P}} \sgn P \prod_{e \in P} z_e,
  \end{align}
  where $\mathcal{P}$ is the set of all disjoint $S$--$T$ paths of $G$.
\end{lemma}

We say that $(S,T)$ is in the \emph{LGV position} on $G$ if $\sgn P$ is constant for any disjoint $S$--$T$ path $P$ of $G$.
When $(S, T)$ is in the LGV position, the number of disjoint $S$--$T$ paths of $G$ can be computed through the LGV lemma.

% regarding S-T paths on DAG as bipartite matching.

The \emph{disjoint $S$--$T$ path problem} on a DAG $G = (V, E)$, which is to find a disjoint $S$--$T$ path of $G$, can be reduced to the bipartite matching problem.
We review the reduction presented in the proof of~\cite[Theorem~2.5.9]{Frank2011}.
Let $\tilde{V}_S$ and $\tilde{V}_T$ be disjoint copies of $\tilde{V} \defeq V \setminus (S \cup T)$.
For $v \in \tilde{V}$, we denote the corresponding vertices to $v$ in $\tilde{V}_S$ and $\tilde{V}_T$ by $v_s$ and $v_t$, respectively.
Also $v_s$ and $v_t$ indicate $v$ itself for $v \in S$ and $v \in T$, respectively.
We construct a bipartite graph $\Gamma$ as follows.
The vertex set of $\Gamma$ is the disjoint union of $V_S \defeq S \cup \tilde{V}_S$ and $V_T \defeq T \cup \tilde{V}_T$.
The edge set of $\Gamma$ is $F_1 \cup F_2$, where
\begin{align}
  F_1 &\defeq \set{\set{u_s, v_t}}[(u, v) \in E], \\
  F_2 &\defeq \set[\big]{\set{v_s, v_t}}[v \in \tilde{V}].
\end{align}

\begin{lemma}[{see~\cite[Theorem~2.5.9]{Frank2011}}]\label{lem:dag_correspondense}
  There is a one-to-one correspondence between disjoint $S$--$T$ paths of $G$ and perfect matchings of $\Gamma$.
\end{lemma}
\begin{proof}
  Let $P$ be a disjoint $S$--$T$ path of $G$ and $U \subseteq V$ the set of vertices covered by $P$.
  Let $M$ be the union of $M_1$ and $M_2$, where
  \begin{align}
    M_1 &\defeq \set{\set{u_s, v_t}}[(u, v) \in P] \subseteq F_1, \\
    M_2 &\defeq \set[\big]{\set{v_s, v_t}}[v \in \tilde{V} \setminus U] \subseteq F_2.
  \end{align}
  Then each $u_s \in V_S$ is covered by $\set{u_s, v_t} \in M_1$ for some $v_t \in V_T$ if $u \in U$ and by $\set{u_s, u_t} \in M_2$ if $u \notin U$.
  In addition, such an edge in $M$ is unique since $P$ is disjoint.
  The same argument holds for vertices in $V_T$.
  Hence $M$ is a perfect matching of $\Gamma$.

  Conversely, let $M$ be a perfect matching of $\Gamma$.
  Define $P \defeq \{(u, v) \mid \set{u_s, v_t} \in M \cap F_1\}$.
  Then the in- and out-degrees of $v \in \tilde{V}$ are the same (zero or one) and the in-degree of $t_j \in T$ and the out-degree of $s_i \in S$ are one in $P$.
  Thus $P$ is the disjoint union of a disjoint $S$--$T$ path $P'$ of $G$ and cycles on $\tilde{V}$.
  In particular, since $G$ is acyclic, it holds $P = P'$.
  It is easily confirmed that the former and the latter correspondences are inversely mapped.
\end{proof}

Orienting every edge of $\Gamma$ appropriately, we show that the correspondence given in \cref{lem:dag_correspondense} preserves the signs of disjoint $S$--$T$ paths and perfect matchings up to the factor of $\prn{-1}^k$.
Suppose that vertices in $V_S$ and $V_T$ are ordered so that the first $k$ vertices are $s_1, \ldots, s_k$ and $t_1, \ldots, t_k$, respectively.
Construct a directed bipartite graph $\vec{\Gamma} = (V_S \cup V_T, \vec{F}_1 \cup \vec{F}_2)$, where $\vec{F}_1$ is the orientation of $F_1$ from $V_T$ to $V_S$ and $\vec{F}_2$ is the orientation of $F_2$ from $V_S$ to $V_T$.
Recall from \cref{sec:perfect_matchings} that $\sgn \vec{M}$ is defined by~\eqref{def:sign_of_directed_perfect_bipartite_matching} for a perfect matching $M$ of $\Gamma$.

% another proof of LGV-lemma
\begin{lemma}\label{lem:normal-sgn}
  Let $P$ be a disjoint $S$--$T$ path of $G$ and $M$ the perfect matching of $\Gamma$ corresponding to $P$.
  Then it holds $\sgn P = \prn{-1}^k \sgn \vec{M}$.
\end{lemma}
\begin{proof}
  Let $\vec{\Gamma}^*$ be the directed bipartite graph obtained from $\vec{\Gamma}$ by appending $k$ directed edges
  \begin{align}
    \vec{F}_3 \defeq \set[\big]{\pbig{t_{\sigma_P(i)}, s_i}}[i \in \intset{k}].
  \end{align}
  Then $\vec{M}' \defeq \vec{F}_2 \cup \vec{F}_3$ is a perfect matching of $\vec{\Gamma}^*$.
  It is clear that $\sgn M' = \sgn P$.
  The number of edges in $\vec{M}'$ from $V_T$ to $V_S$ is $\card[\big]{\vec{F}_3} = k$.
  Hence the sign of $\vec{M}'$ is $\prn{-1}^k \sgn P$.

  We show $\sgn \vec{M} = \sgn \vec{M}$, which implies the claim.
  For $i \in \intset{k}$, let $P_i$ be the $s_i$--$t_{\sigma_P(i)}$ path in $P$ and $U_i$ the set of vertices in $\tilde{V}$ covered by $P_i$.
  Let $\vec{C}_i$ be the disjoint union of $\vec{C}_{i,1}, \vec{C}_{i,2}$ and $\vec{C}_{i,3}$ defined by
  \begin{align}
    \vec{C}_{i,1} &\defeq \set{(v_t, u_s)}[(u, v) \in P_i] \subseteq \vec{F}_1, \\
    \vec{C}_{i,2} &\defeq \set{(v_s, v_t)}[v \in U_i] \subseteq \vec{F}_2, \\
    \vec{C}_{i,3} &\defeq \set[\big]{(t_{\sigma_P(i)}, s_i)} \subseteq \vec{F}_3.
  \end{align}
  Then $\set[\big]{\vec{C}_1, \ldots, \vec{C}_k}$ is the set of alternating cycles of $M \symdif M'$.
  Note that $M$ is also a perfect matching of $\vec{\Gamma}^*$.

  An even cycle of directed edges is said to be \emph{oddly oriented} if the number of edges consistent with the direction of a traversal is odd for either choice of traversals.
  Every alternating cycle $\vec{C}_i$ is oddly oriented because the unique edge in $\vec{C}_{i,3}$ is opposite to all the other edges of $\vec{C}_i$.
  This means that the signs of $M$ and $M'$ are the same (see, e.g.,~\cite[Lemma~8.3.1]{Plummer1986}).
  Hence the claim holds.
\end{proof}

By \cref{lem:normal-sgn,lem:matching_sign_directed,lem:dag_correspondense}, we have the following.

\begin{theorem}
  Let $G$ be a DAG and take disjoint vertex subsets $S, T$ with $\card{S} = \card{T}$.
  Let $\vec{\Gamma}$ be the directed bipartite graph defined above and $(\vec{A}_1, \vec{A}_2)$ the matrix pair representing perfect matchings of $\vec{\Gamma}$ given in \cref{sec:perfect_matchings}.
  Then $\cbase(\vec{A}_1, \vec{A}_2)$ and disjoint $S$--$T$ paths of $G$ correspond one-to-one.
  In addition, if $(S, T)$ is in the LGV position, then $(\vec{A}_1, \vec{A}_2)$ is Pfaffian with constant $\pm 1$.
\end{theorem}

\Cref{fig:LGVDAG} illustrates two examples of $(S, T)$ that is in the LGV position.
If $\sigma_P$ is identity for every disjoint $S$--$T$ path $P$ of $G$, then $(S, T)$ is clearly in the LGV position.
Such $(S, T)$ is called \emph{nonpermutable}~\cite{Gessel1985} and one famous example arises from a planar DAG $G$.
Suppose that $s_1, \ldots, s_k, t_k, \ldots, t_1$ are aligned clockwise on the boundary of one face of $G$.
Then $(S, T)$ is nonpermutable because an $S$--$T$ path $P$ cannot be disjoint if $\sigma_P$ is not identity.
Another example of the LGV position also arises from a planar DAG $G$.
Suppose that all the vertices in $S$ are on the boundary of a face $F$ of $G$, and all the vertices in $T$ are on the boundary of another face that does not adjoin $F$.
Then $\sigma_P$ must be a power of a cyclic permutation over $\intset{k}$, whose sign is always $+1$ when $k$ is odd.

\tikzset{s/.style={fill, circle, inner sep=0pt, minimum size=4pt}}
\tikzset{t/.style={draw, circle, inner sep=0pt, minimum size=4pt}}
\tikzset{arc/.style={-{Stealth[length=2mm, width=2mm, angle'=40]}}}

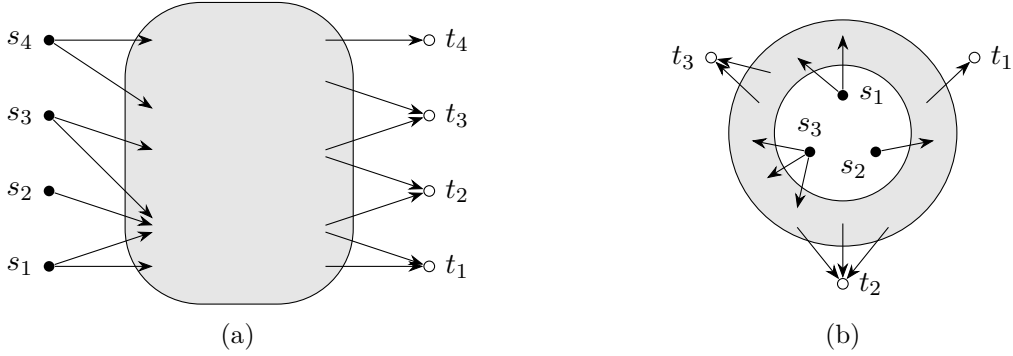
\begin{figure}[tbp]
  \begin{minipage}[b]{0.5\textwidth}
    \centering
    \begin{tikzpicture}
      \node (s1) [s] at (0, 1) {}; \node [left=0mm of s1] {$s_1$};
      \node (s2) [s] at (0, 2) {}; \node [left=0mm of s2] {$s_2$};
      \node (s3) [s] at (0, 3) {}; \node [left=0mm of s3] {$s_3$};
      \node (s4) [s] at (0, 4) {}; \node [left=0mm of s4] {$s_4$};
      \node (t1) [t] at (5, 1) {}; \node [right=0mm of t1] {$t_1$};
      \node (t2) [t] at (5, 2) {}; \node [right=0mm of t2] {$t_2$};
      \node (t3) [t] at (5, 3) {}; \node [right=0mm of t3] {$t_3$};
      \node (t4) [t] at (5, 4) {}; \node [right=0mm of t4] {$t_4$};

      \node (e) at (1.5,1.0) {};
      \node (f) at (1.5,1.5) {};
      \node (g) at (1.5,2.5) {};
      \node (h) at (1.5,3.0) {};
      \node (i) at (1.5,4.0) {};
      \node (e2) at (3.5,1.0) {};
      \node (f2) at (3.5,1.5) {};
      \node (g2) at (3.5,2.5) {};
      \node (h2) at (3.5,3.5) {};
      \node (i2) at (3.5,4) {};

      \draw[rounded corners=1cm,fill=gray!20] (1,0.5) rectangle ++(3,4) node [midway] {} ;

      \draw[arc] (s1) -- (e);
      \draw[arc] (s1) -- (f);
      \draw[arc] (s2) -- (f);
      \draw[arc] (s3) -- (f);
      \draw[arc] (s3) -- (g);
      \draw[arc] (s4) -- (h);
      \draw[arc] (s4) -- (i);

      \draw[arc] (e2) -- (t1);
      \draw[arc] (f2) -- (t1);
      \draw[arc] (f2) -- (t2);
      \draw[arc] (g2) -- (t2);
      \draw[arc] (g2) -- (t3);
      \draw[arc] (h2) -- (t3);
      \draw[arc] (i2) -- (t4);
    \end{tikzpicture}
    \subcaption{}
  \end{minipage}%
  \begin{minipage}[b]{0.5\textwidth}
    \centering
    \begin{tikzpicture}
      \draw[fill = gray!20] (0, 0) circle (1.5cm);
      \draw[fill = white] (0, 0) circle (0.9cm);

      \node (s1) [s] at ( 0    ,  0.5 ) {}; \node [right=0mm of s1] {$s_1$};
      \node (s2) [s] at ( 0.433, -0.25) {}; \node [below left=-1mm of s2] {$s_2$};
      \node (s3) [s] at (-0.433, -0.25) {}; \node [above=0mm of s3] {$s_3$};

      \node (t1) [t] at ( 1.732,  1) {}; \node [right=0mm of t1] {$t_1$};
      \node (t2) [t] at (     0, -2) {}; \node [right=0mm of t2] {$t_2$};
      \node (t3) [t] at (-1.732,  1) {}; \node [left=0mm of t3] {$t_3$};

      \draw[arc] (s1) -- (0, 1.3);
      \draw[arc] (s1) -- (-0.6, 1);
      \draw[arc] (s2) -- (1.2, -0.1);
      \draw[arc] (s3) -- (-1.2, -0.1);
      \draw[arc] (s3) -- (-0.6, -1);
      \draw[arc] (s3) -- (-1, -0.6);

      \draw[arc] (1.1, 0.4) -- (t1);
      \draw[arc] (0, -1.2) -- (t2);
      \draw[arc] (0.6, -1.25) -- (t2);
      \draw[arc] (-0.6, -1.25) -- (t2);
      \draw[arc] (-1.1, 0.4) -- (t3);
      \draw[arc] (-0.95, 0.8) -- (t3);
    \end{tikzpicture}
    \subcaption{}
  \end{minipage}
  \caption{%
    Examples of $(S, T)$ that are in the LGV position.
    Gray areas represent planar DAGs.
    (a) $s_1, \ldots, s_4, t_4, \ldots, t_1$ are aligned clockwise on the boundary of an (outer) face.
    (b) $S$ adjoins a face, $T$ adjoins another face, and $\card{S} = \card{T}$ is odd.
  }\label{fig:LGVDAG}
\end{figure}

Lastly, we present another application of \cref{lem:normal-sgn}.
Let $(\vec{A}_1, \vec{A}_2)$ be the matrix pair representing perfect matchings of $\vec{\Gamma}$.
Let $z = \prn{z_e}_{e \in F_1 \cup F_2}$ be a vector of distinct indeterminates indexed by $F_1 \cup F_2$.
We substitute $z_e \defeq 1$ for $e \in F_2$.
On indexing an component of $z$, we identify $\set{u_s, v_t} \in F_1$ with $(u, v) \in E$.
Put $N(z) \defeq \vec{A}_1 D(z) \trsp{{\vec{A}_2}}$.
Denote by $\mathcal{M}$ the set of perfect matchings of $\Gamma$, by $\mathcal{P}$ the set of disjoint $S$--$T$ paths of $G$, and by $\Omega(z)$ the matrix defined in the LGV lemma for $G$.
Then we have
\begin{align}\label{eq:derive_LGV}
  \det N(z)
  = \sum_{M \in \mathcal{M}} \sgn \vec{M} \prod_{e \in M} z_e
  = \sum_{P \in \mathcal{P}} \prn{-1}^k \sgn P \prod_{e \in P} z_e
  = \det \prn{-\Omega(z)},
\end{align}
where the first equality follows from~\eqref{eq:cauchy_binet_1}, the second one follows from \cref{lem:dag_correspondense,lem:normal-sgn}, and the last one is due to the LGV lemma.
Indeed, we can prove the equality $\det N(z) = \det \prn{-\Omega(z)}$ without using the LGV lemma, which turns to be a new proof of the lemma.

\begin{proof}[{of \cref{lem:normal-lgv}}]
  Recall that entries of $N(z)$ satisfy~\eqref{eq:bipartite_adjacency_matrix} up to the factor of $x$.
  The matrix $N(z)$ can be partitioned as
  \begin{align}
    N(z) =
    \begin{pmatrix}
      X & Y \\
      Z & W
    \end{pmatrix},
  \end{align}
  where the row and column sets of $X$ correspond to $S$ and $T$, respectively and other rows and columns are indexed by $\tilde{V}$.
  By sorting $\tilde{V}$ in a topological order with respect to $\vec{G}$, we can transform $W$ into an upper triangular matrix, as $G$ is acyclic.
  In addition, diagonal entries of $W$ are 1 since $(v_s, v_t) \in \vec{E}_2$ for $v \in \tilde{V}$.
  Hence we have $\det W(z) = 1$.
  By the formula~\eqref{eq:shurt_det} on the Schur complement, we have
  \begin{align}
    \det N(z)
    = \det W \det \pbig{X - Y{W}^{-1}Z}
    = \det \pbig{X - YW^{-1}Z}.
  \end{align}

  We show $X - YW^{-1}Z = -\Omega(z)$.
  Fix $i, j \in \intset{k}$ and let $G^{(i,j)}$ be the subgraph of $G$ obtained by deleting $S \setminus \set{s_i}$ and $T \setminus \set{t_j}$.
  We consider the disjoint $\set{s_i}$--$\set{t_j}$ path problem on $G^{(i,j)}$.
  The directed bipartite graph $\vec{\Gamma}^{(i,j)}$ on this problem is the subgraph of $\vec{\Gamma}$ induced by $\pbig{\tilde{V}_S \cup \set{s_i}} \cup \pbig{\tilde{V}_T \cup \set{t_j}}$.
  Similarly, the counterpart $N^{(i,j)}(z)$ of $N(z)$ on this problem is $N(z)[\tilde{V}_S \cup \set{s_i}, \tilde{V}_T \cup \set{t_j}]$, which means
  \begin{align}
    N^{(i,j)}(z) =
    \begin{pmatrix}
      X[\set{s_i}, \set{t_j}] & Y[\set{s_i}, \tilde{V}] \\
      Z[\tilde{V}, \set{t_j}] & W
    \end{pmatrix}.
  \end{align}
  By~\eqref{eq:shurt_det}, it holds
  \begin{align}
    \det N^{(i,j)}(z) = X[\set{s_i}, \set{t_j}] - Y[\set{s_i}, \tilde{V}] W^{-1} Z[\tilde{V}, \set{t_j}],
  \end{align}
  which coincides with the $(i,j)$th entry of $X - YW^{-1}Z$.

  Denote by $\mathcal{M}^{(i,j)}$ the set of perfect matchings of $\Gamma^{(i,j)}$ and by $\mathcal{P}^{(i,j)}$ the set of disjoint $\set{s_i}$--$\set{t_j}$ paths of $G^{(i,j)}$, which are exactly $s_i$--$t_i$ paths of $G$.
  In the same was as~\eqref{eq:derive_LGV}, we have
  \begin{align}
    \det N^{(i,j)}(z)
    &= \sum_{M \in \mathcal{M}^{(i,j)}} \sgn \vec{M} \prod_{e \in M} z_e \\
    &= \sum_{P \in \mathcal{P}^{(i,j)}} \prn{-1}^1 \sgn P \prod_{e \in P} z_e \\
    &= -\sum_{P \in \mathcal{P}^{(i,j)}} \prod_{e \in P} z_e
    = -\Omega_{i,j}(z).\label{eq:LGV_proof_last}
  \end{align}
  Note that we did not apply the LGV lemma on the last equality of~\eqref{eq:LGV_proof_last}; it is just the definition~\eqref{def:Omega_entry} of $\Omega_{i,j}(z)$.
  Hence $X - YW^{-1}Z = -\Omega(z)$ holds.
  The LGV lemma follows from this fact together with the first two equalities in~\eqref{eq:derive_LGV}.
\end{proof}

\subsection{Shortest Disjoint \ST\ Paths on Undirected Graphs}\label{sec:st}

%We extend the LGV lemma to the counting of shortest vertex-disjoint $S$--$T$ paths on a weighted undirected graph.
Let $G=(V,E)$ be an undirected graph and $S = \set{s_1, \ldots, s_k}$ and $T = \set{t_1, \ldots, t_k}$ disjoint vertex subsets of cardinality $k$.
An \emph{$S$--$T$ path} $P \subseteq E$ of $G$ is the union of $k$ paths $P_1, \ldots, P_k$ of $G$ satisfying~\ref{item:P1} (with direction ignored).
We denote the permutation in~\ref{item:P1} by $\sigma_P$.
An $S$--$T$ path $P$ is said to be (\emph{vertex}-)\emph{disjoint} if $P_i$ and $P_j$ do not share the same vertices for all distinct $i, j \in \intset{k}$.
The \emph{disjoint $S$--$T$ path problem} on $G$ is to find a disjoint $S$--$T$ path of $G$.
We also equip $G$ with a positive edge length $\funcdoms{l}{E}{\setRpp}$.
The \emph{length} $l(P)$ of an disjoint $S$--$T$ path $P$ of $G$ is defined as the sum of lengths of all edges in $P$.
The \emph{shortest disjoint $S$--$T$ path problem} on $G$ is to find a disjoint $S$--$T$ path of $G$ with minimum length.

We first show that the shortest disjoint $S$--$T$ path problem on an undirected graph is a generalization of the disjoint $S$--$T$ path problem on a DAG\@.
Let $\vec{G} = (V, \vec{E})$ be a DAG and take disjoint vertex subsets $S, T \subseteq V$ of cardinality $k$.
Let $v_1, \ldots, v_n$ be a topological ordering of $V$ with respect to $\vec{G}$, i.e., $(v_i, v_j) \notin \vec{E}$ for all $i, j \in \intset{n}$ with $i \ge j$.
Let $G = (V, E)$ be the undirected graph obtained from $\vec{G}$ by ignoring the orientation.
We set an edge length for $G$ as $l(e) \defeq j-i$ for $e = \set{v_i, v_j} \in E$ with $i < j$.

\begin{proposition}
  Suppose that $\vec{G}$ has at least one disjoint directed $S$--$T$ path.
  Then $P \subseteq E$ is a shortest disjoint $S$--$T$ path of $G$ if and only if $\vec{P}$ is a disjoint directed $S$--$T$ path of $\vec{G}$.
\end{proposition}

\begin{proof}
  Let $P = P_1 \cup \cdots \cup P_k$ be a disjoint $S$--$T$ path of $G$, where $P_i$ is the $s_i$--$t_{\sigma_P(i)}$ path contained in $P$ for $i \in [k]$.
  For each $e \in P$, we denote by $\partial^s e$ and $\partial^t e$ the ends of $e$ which are in sides closer to $s_i$ and to $t_{\sigma_P(i)}$ on $P_i$, respectively.
  Let $f(v_i) \defeq i$ for $i \in [n]$.
  Then we have
  \begin{align}\label{eq:l_P_i}
    l(P_i)
    = \sum_{e \in P_i} l(e)
    = \sum_{\{u,v\} \in P_i} |f(v) -f(u)|
    \ge \sum_{e \in P_i} \prn{f\pbig{\partial^t e} - f\pbig{\partial^s e}}
    = f\pbig{t_{{\sigma_P}(i)}} - f(s_i)
  \end{align}
  for $i \in \intset{k}$.
  Summing~\eqref{eq:l_P_i} up over all $i \in \intset{k}$, we have
  \begin{align}\label{eq:l_P}
    l(P)
    = \sum_{i=1}^k l(P_i)
    \ge \sum_{i=1}^k \prn{f\pbig{t_{{\sigma_P}(i)}} - f(s_i)}
    = \sum_{v \in T} f(v) - \sum_{u \in S} f(u).
  \end{align}
  The equality of~\eqref{eq:l_P} is attained if and only if $f\pbig{\partial^s e} \le f\pbig{\partial^t e}$ for every $e \in P$.
  Since $V$ is topologically ordered, this is equivalent to the condition that $\vec{P}$ is a disjoint directed $S$--$T$ path of $\vec{G}$.
  By assumption, such a disjoint $S$--$T$ path exists on $G$.
  Hence the set of shortest disjoint $S$--$T$ path of $G$ corresponds to the set of disjoint directed $S$--$T$ paths of $\vec{G}$.
\end{proof}

% Psuedopoly algorithm

% graphical matroid intersection
We next describe a reduction from the disjoint $S$--$T$ path problem on an undirected graph $G = (V, E)$ to the graphic matroid intersection problem, given by Tutte~\cite{Tutte1965}.
Let $G_S$ and $G_T$ be the graphs obtained from $G$ by shrinking $S$ and $T$ into a single vertex $v^*$, respectively.
Recall from \cref{sec:regular_matroids} that $\matroid(G_S)$ and $\matroid(G_T)$ denote the graphic matroids of $G_S$ and $G_T$.
An edge subset $B \subseteq E$ is a common base of $\matroid(G_S)$ and $\matroid(G_T)$ if and only if $B$ is a spanning forest of $G$ consisting of $k$ connected components each of which covers exactly one vertex belonging to $S$ and exactly one vertex belonging to $T$ in $G$.
This means that any common base $B$ contains a unique disjoint $S$--$T$ path $P$ of $G$.
Conversely, given a disjoint $S$--$T$ path $P$, we can construct a common base $B$ by adding some of the remaining edges to $P$ so that it forms a spanning forest with $k$ connected components.
This means that $\matroid(G_S)$ and $\matroid(G_T)$ have a common base if and only if $G$ has a disjoint $S$--$T$ path.
Note that $B \mapsto P$ is injective in this correspondence but $P \mapsto B$ is not so.

Analogously to Tutte's reduction, the shortest disjoint $S$--$T$ path problem can be reduced to the weighted graphic matroid intersection problem.
This is a special case of Yamaguchi's reduction~\cite{Yamaguchi2016} of the shortest disjoint \calS-path problem, which will be described in \cref{sec:stu}, on our $S$--$T$ path setting.
While this is a simple extension of Tutte's reduction, it provides a one-to-one correspondence between optimal solutions of these problems.
Let $G^*$ be the graph obtained from $G$ by appending $v^*$ as a new vertex and adding the edge set $E' \defeq \set[\big]{\set{v, v^*}}[v \in \tilde{V}]$, where $\tilde{V} := V \setminus (S \cup T)$.
Similarly, let $G_S^*$ and $G_T^*$ be the graphs obtained from $G_S$ and $G_T$ by adding $E'$, respectively.
Denote $E \cup E'$ by $E^*$.
We set a column weight $\funcdoms{w}{E^*}{\setR}$ of $\matroid\pbig{G_S^*}$ and $\matroid\pbig{G_T^*}$ as $w(e) \defeq l(e)$ for $e \in E$ and as $q(e) \defeq 0$ for $e \in E'$.

\begin{lemma}\label{lem:st-one2one}
  The minimum length of a disjoint $S$--$T$ path of $G$ with respect to $l$ is equal to the minimum weight of a common base of $\matroid\pbig{G_S^*}$ and $\matroid\pbig{G_T^*}$ with respect to $w$.
  In addition, there is a one-to-one correspondence between shortest disjoint $S$--$T$ paths of $G$ and minimum-weight common bases of $\matroid\pbig{G_S^*}$ and $\matroid\pbig{G_T^*}$.
\end{lemma}

\begin{proof}
  As in the unweighted case, $B \subseteq E^*$ is a common base of $\matroid\pbig{G_S^*}$ and $\matroid\pbig{G_T^*}$ if and only if $B$ is a spanning forest of $G^*$ consisting of $k+1$ connected components $B_0 = B \cap E', B_1, \ldots, B_k \subseteq B$ such that $v^*$ is covered only by $B_0$ and each of $B_1, \ldots, B_k$ covers exactly one vertex belonging to $S$ and exactly one vertex belonging to $T$.
  Thus any common base $B$ contains a unique disjoint $S$--$T$ path $P$ of $G$.
  Since $w$ is nonnegative, we have $w(B) \ge l(P)$.
  Conversely, given a disjoint $S$--$T$ path $P$, we can construct a common base $B$ by adding $\set{v, v^*} \in E'$ for each $v \in \tilde{V}$ that is not covered by $P$.
  We have $w(B) = l(P)$ as $w(e) = 0$ for $e \in E'$.
  Hence the shortest length of a disjoint $S$--$T$ path of $G$ is equal to the minimum weight of a common base of $\matroid\pbig{G_S^*}$ and $\matroid\pbig{G_T^*}$.

  As we discussed above, every minimum-weight common base contains a unique shortest disjoint $S$--$T$ path, and for every shortest disjoint $S$--$T$ path $P$, there exists a minimum-weight common base $B$ containing $P$.
  Since $l(e) > 0$ for all $e \in E$, this map $P \mapsto B$ is injective.
\end{proof}

We next consider the correspondence in \cref{lem:st-one2one} by means of matrices.
Let $A$ be the incidence matrix~\eqref{def:incidence_matrix} of any orientation of $G^*$.
We define $A_1 \defeq A[S \cup \tilde{V}, E^*]$ and $A_2 \defeq A[T \cup \tilde{V}, E^*]$.
Then $A_1$ and $A_2$ represent $\matroid(G_S^*)$ and $\matroid(G_T^*)$, respectively.
We assume that for each $i \in \intset{k}$, the $i$th rows of $A_1$ and $A_2$ correspond to $s_i$ and $t_i$, respectively.
Then the sign of a disjoint $S$--$T$ path $P$ can be written using the subdeterminants of $A_1$ and $A_2$ as follows, even if $P$ is not the shortest.

\begin{lemma}\label{lem:st_sign}
  Let $P$ be a disjoint $S$--$T$ path of $G$ and $B$ a common base of $(A_1, A_2)$ containing $P$.
  Then we have
  \begin{align}
    \sgn P = \prn{-1}^k \det A_1[B] \det A_2[B].
  \end{align}
\end{lemma}

\begin{proof}
  As shown in the proof of \cref{lem:st-one2one}, $B$ consists of $k+1$ connected components $B_0$, $B_1, \ldots, B_k$ in $G^*$.
  Here, $B_0$ covers $v^*$ and $B_i$ contains the $s_i$--$t_{\sigma_P(i)}$ path contained in $P$ for $i \in \intset{k}$.
  Let $V_i$ be the vertex subset covered by $B_i$ for $i = 0, \ldots, k$.
  By row and column permutations, $A_1[B]$ is transformed into a block diagonal matrix $X = \diag(X_0, X_1, \ldots, X_k)$, where $X_0$ has the row set $V_0 \setminus \set{v^*}$ and the column set $B_0$ and $X_i$ has the row set $V_i \setminus \set[\big]{t_{\sigma_P(i)}}$ and the column set $B_i$ for $i \in \intset{k}$.
  Similarly, we permute rows and columns of $A_2[B]$ so that it becomes a block diagonal matrix $Y = \diag(Y_0, Y_1, \ldots, Y_k)$, where $Y_0$ has the row set $V_0 \setminus \set{v^*}$ and the column set $B_0$ and $Y_i$ has the row set $V_i \setminus \set[\big]{s_i}$ and the column set $B_i$ for $i \in \intset{k}$.
  We can assume that column permutations of these two transformations are the same.
  We can also assume that $X_0$ and $Y_0$ have the same ordering of rows, and for each $i \in \intset{k}$, the first rows of $X_i$ and $Y_i$ correspond to $s_i$ and $t_{\sigma_P(i)}$, respectively, and other rows are in the same ordering.
  This implies that the product of signs of these two row permutations on $A_1$ and $A_2$ are $\sgn P$.
  Hence we have
  \begin{align}\label{eq:st_sign_mid}
    \det A_1[B] \det A_2[B]
    = \sgn P \det X \det Y
    = \sgn P \prod_{i=0}^k \det X_i \det Y_i.
  \end{align}

  Here, we evaluate $\det X_i \det Y_i$ for $i = 0, \ldots, k$.
  Note that $\det X_i \det Y_i$ is in $\set{+1, -1}$ since the incidence matrix $A$ is totally unimodular and $A_1[B]$ and $A_2[B]$ are nonsingular.
  For $i = 0$, we have $\det X_0 \det Y_0 = 1$ by $X_0 = Y_0 = A[V_0 \setminus \set{v^*}, B_0]$.
  For $i \in \intset{k}$, let $Z_i$ be the matrix such that the first row is the sum of those of $X_i$ and $Y_i$ and other rows are $X_i[V_i \cap \tilde{V}, B_i] = Y_i[V_i \cap \tilde{V}, B_i] = A[V_i \cap \tilde{V}, B_i]$.
  Note that $\det X_i + \det Y_i = \det Z_i$.
  Since both the ends of every edge in $B_i$ are in $V_i$, every column of $Z_i$ contains exactly one $+1$ and one $-1$ and other entries are zero.
  Hence $Z_i$ is singular and thus we have $\det X_i = -\det Y_i$.
  The claim of the proposition holds by~\eqref{eq:st_sign_mid}.
\end{proof}

We have the following conclusion by \cref{lem:st-one2one,lem:st_sign}.

\begin{theorem}
  Let $G$ be an undirected graph and take disjoint vertex subsets $S, T$ with $k \defeq \card{S} = \card{T}$.
  Let $(A_1, A_2)$ be the matrix pair associated with $G, S$ and $T$.
  When $(S, T)$ is in the LGV position, then $(A_1, A_2)$ is Pfaffian with constant ${(-1)}^k \sgn P$, where $P$ is an arbitrary disjoint $S$--$T$ path of $G$.
  In addition, when $G$ is equipped with a positive edge length $l$, there is a one-to-one correspondence between shortest disjoint $S$--$T$ paths of $G$ and minimum-weight common bases of $(A_1, A_2)$ with respect to the column weight $w$ defined above.
\end{theorem}

We say that $(S, T)$ is in the \emph{LGV position} on $G$ if $\sgn P$ is constant for any disjoint $S$--$T$ path $P$ of $G$.
\Cref{fig:LGV-ST} illustrates two examples of LGV-position, which are obtained by ignoring edge directions in \cref{fig:LGVDAG}.

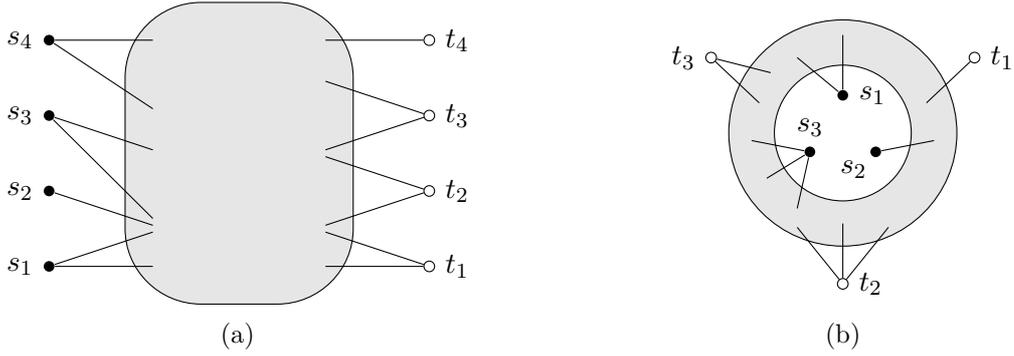
\begin{figure}[tbp]
  \begin{minipage}[b]{0.5\textwidth}
    \centering
    \begin{tikzpicture}
      \node (s1) [s] at (0, 1) {}; \node [left=0mm of s1] {$s_1$};
      \node (s2) [s] at (0, 2) {}; \node [left=0mm of s2] {$s_2$};
      \node (s3) [s] at (0, 3) {}; \node [left=0mm of s3] {$s_3$};
      \node (s4) [s] at (0, 4) {}; \node [left=0mm of s4] {$s_4$};
      \node (t1) [t] at (5, 1) {}; \node [right=0mm of t1] {$t_1$};
      \node (t2) [t] at (5, 2) {}; \node [right=0mm of t2] {$t_2$};
      \node (t3) [t] at (5, 3) {}; \node [right=0mm of t3] {$t_3$};
      \node (t4) [t] at (5, 4) {}; \node [right=0mm of t4] {$t_4$};

      \node (e) at (1.5,1.0) {};
      \node (f) at (1.5,1.5) {};
      \node (g) at (1.5,2.5) {};
      \node (h) at (1.5,3.0) {};
      \node (i) at (1.5,4.0) {};
      \node (e2) at (3.5,1.0) {};
      \node (f2) at (3.5,1.5) {};
      \node (g2) at (3.5,2.5) {};
      \node (h2) at (3.5,3.5) {};
      \node (i2) at (3.5,4) {};

      \draw[rounded corners=1cm,fill=gray!20] (1,0.5) rectangle ++(3,4) node [midway] {} ;

      \draw (s1) -- (e);
      \draw (s1) -- (f);
      \draw (s2) -- (f);
      \draw (s3) -- (f);
      \draw (s3) -- (g);
      \draw (s4) -- (h);
      \draw (s4) -- (i);

      \draw (e2) -- (t1);
      \draw (f2) -- (t1);
      \draw (f2) -- (t2);
      \draw (g2) -- (t2);
      \draw (g2) -- (t3);
      \draw (h2) -- (t3);
      \draw (i2) -- (t4);
    \end{tikzpicture}
    \subcaption{}
  \end{minipage}%
  \begin{minipage}[b]{0.5\textwidth}
    \centering
    \begin{tikzpicture}
      \draw[fill = gray!20] (0, 0) circle (1.5cm);
      \draw[fill = white] (0, 0) circle (0.9cm);

      \node (s1) [s] at ( 0    ,  0.5 ) {}; \node [right=0mm of s1] {$s_1$};
      \node (s2) [s] at ( 0.433, -0.25) {}; \node [below left=-1mm of s2] {$s_2$};
      \node (s3) [s] at (-0.433, -0.25) {}; \node [above=0mm of s3] {$s_3$};

      \node (t1) [t] at ( 1.732,  1) {}; \node [right=0mm of t1] {$t_1$};
      \node (t2) [t] at (     0, -2) {}; \node [right=0mm of t2] {$t_2$};
      \node (t3) [t] at (-1.732,  1) {}; \node [left=0mm of t3] {$t_3$};

      \draw (s1) -- (0, 1.3);
      \draw (s1) -- (-0.6, 1);
      \draw (s2) -- (1.2, -0.1);
      \draw (s3) -- (-1.2, -0.1);
      \draw (s3) -- (-0.6, -1);
      \draw (s3) -- (-1, -0.6);

      \draw (1.1, 0.4) -- (t1);
      \draw (0, -1.2) -- (t2);
      \draw (0.6, -1.25) -- (t2);
      \draw (-0.6, -1.25) -- (t2);
      \draw (-1.1, 0.4) -- (t3);
      \draw (-0.95, 0.8) -- (t3);
    \end{tikzpicture}
    \subcaption{}
  \end{minipage}
  \caption{%
    Examples of $(S, T)$ that are in the LGV position.
    Gray areas represent planar DAGs.
  }\label{fig:LGV-ST}
\end{figure}

\subsection{Shortest Disjoint \STU\ Paths on Undirected Graphs}\label{sec:stu}

We further extend the shortest disjoint $S$--$T$ path problem to the shortest disjoint $S$--$T$--$U$ path problem, which is a special case of the shortest disjoint \calS-path problem.

% matroid parity
We first introduce Mader's disjoint \calS-path problem~\cite{Gallai1964,Mader1978}.
Let $G = (V, E)$ be an undirected graph and $\mathcal{S} = \{S_1, \ldots, S_s\}$ a family of disjoint nonempty subsets of $V$.
Suppose that $\Sigma \defeq S_1 \cup \cdots \cup S_s$ is of cardinality $2k$ and ordered as $u_1, \ldots, u_{2k}$ so that $i \le j$ means $\alpha \le \beta$, where $u_i \in S_{\alpha}$ and $u_j \in S_{\beta}$ for $i, j \in \intset{2k}$.
Vertices in $\Sigma$ are called \emph{terminals}.
Recall that $F_{2k}$ is the subset of $\sym_{2k}$ defined by~\eqref{def:F}.
An \emph{\calS-path} $P$ of $G$ is the union of $k$ paths $P_1, \ldots, P_k \subseteq E$ of $G$ satisfying the following:
\begin{enumerate}[label={(P2)}]
  \item There exists a permutation $\sigma \in F_{2k}$ such that $P_i$ is a path between $u_{\sigma(2i-1)} \in S_\alpha$ and $u_{\sigma(2i)} \in S_\beta$ with $\alpha \ne \beta$ for each $i \in \intset{k}$.\label{item:P2}
\end{enumerate}
Namely, $P$ is an \calS-path if the ends of each $P_i$ belong to distinct parts in $\mathcal{S}$ and the ends of $P_i$ and $P_j$ are disjoint for all distinct $i, j \in \intset{k}$.
We call an \calS-path $P$ \emph{disjoint} if $P_i$ and $P_j$ have no common vertices for all distinct $i,j \in [k]$.
For a disjoint \calS-path $P$, a permutation satisfying~\ref{item:P2} uniquely exists in $F_{2k}$ and we denote it by $\sigma_P$.
The \emph{sign} of a disjoint \calS-path $P$ is defined as $\sgn P \defeq \sgn \sigma_P$.
The \emph{disjoint \calS-path problem} on $G$ is to find a disjoint \calS-path of $G$.
We also consider the situation when $G$ is equipped with a positive edge length $l \colon E \to \setRpp$.
Then the \emph{length} of an \calS-path $P$ is defined as $l(P) \defeq \sum_{e \in P} l(e)$.
The \emph{shortest disjoint \calS-path problem} is to find a disjoint \calS-path of $G$ with minimum length.

We next describe a reduction of the disjoint \calS-path problem to the linear matroid parity problem, based on Schrijver's linear representation~\cite{Schrijver2003} of Lov\'asz' reduction~\cite{Lovasz1980}.
We assume that there are no edges connecting terminals.
Put $\tilde{V} \defeq V \setminus \Sigma$, $\tilde{E} \defeq \set[\big]{\set{u,v} \in E}[u,v \in \tilde{V}]$, $m \defeq \card{E}$ and $\tilde{m} \defeq \card[\big]{\tilde{E}}$.
Fix two-dimensional row vectors $b_1, \ldots, b_s$ which are pairwise linearly independent.
We construct a matrix
\begin{align}
  X = \begin{pmatrix}
    X_1 & O   \\
    X_2 & X_3
  \end{pmatrix}
  %\in \setR^{(2n-2k) \times 2m}
\end{align}
from $G$ as follows.
The size of each block is $2k \times 2\tilde{m}$ for $X_1$, $(2n-4k) \times 2\tilde{m}$ for $X_2$, and $(2n-4k) \times 2(m-\tilde{m})$ for $X_3$.
Each edge $e \in \tilde{E}$ is associated with two columns of $\begin{psmallmatrix} X_1 \\ X_2 \end{psmallmatrix}$ and each $e \in E \setminus \tilde{E}$ is associated with two columns of $\begin{psmallmatrix} O \\ X_3 \end{psmallmatrix}$.
Each terminal $u_i \in \Sigma$ corresponds to the $i$th row of $\begin{pmatrix} X_1 & O \end{pmatrix}$ for $i \in \intset{2k}$ and each $v \in \tilde{V}$ is associated with two rows of $\begin{pmatrix} X_2 & X_3 \end{pmatrix}$.
Entries of each block are determined as follows.
\begin{itemize}
  \item The $1 \times 2$ submatrix of $X_1$ associated with $u_i \in U$ and $e \in E \setminus \tilde{E}$ is $b_\alpha$ if $e \cap S_{\alpha} = \set{u_i}$ and $O$ otherwise.
  \item The $2 \times 2$ submatrix of $X_2$ associated with $v \in \tilde{V}$ and $e \in E \setminus \tilde{E}$ is the identify matrix $I_2$ of order two if $v \in e$ and $O$ otherwise.
  \item The matrix $X_3$ is defined to be the Kronecker product $H[\tilde{V}, \tilde{E}] \otimes I_2$, where $H$ is the incidence matrix~\eqref{def:incidence_matrix} of any orientation of $G$.
        Namely, $X_3$ is obtained from $H[\tilde{V}, \tilde{E}]$ by replacing $+1$ with $+I_2$, $-1$ with $-I_2$, and $0$ with $O$.
\end{itemize}
We regard each edge $e \in E$ as a line of $X$, which consists of the two columns associated with $e$.

\begin{lemma}[{\cite[Lemma~4]{Yamaguchi2016}}]\label{lem:Yamaguchi2017}
  An edge subset $B \subseteq E$ is a parity base of $(X, E)$ if and only if $B$ is a spanning forest of $G$ such that every connected component covers exactly two terminals belonging to distinct parts of $\mathcal{S}$.
\end{lemma}

Note that if $B$ is a parity base of $(X, E)$, the number of connected components must be $k$ since $B$ covers all the vertices of $G$ by \cref{lem:Yamaguchi2017}.
Hence $(X, E)$ has a parity base if and only if $G$ has a disjoint \calS-path.
Unfortunately, as in the $S$--$T$ path case described in \cref{sec:st}, this reduction does not provide a one-to-one correspondence between $\pbase(X, E)$ and the set of disjoint \calS-paths of $G$.

Yamaguchi~\cite{Yamaguchi2016} showed that the shortest disjoint \calS-path problem can be reduced to the weighted linear matroid parity problem.
Here we present a simplified reduction for our setting (where an \calS-path covers all terminals), together with a one-to-one correspondence of optimal solutions.
Let $G^*$ be the graph obtained from $G$ by adding a new vertex $v^*$ and an edge set $E' \defeq \set[\big]{\set{v, v^*}}[v \in \tilde{V}]$.
We construct a matrix $X^*$ from $G^*$ in the same way as the construction of $X$ from $G$.
Let $A$ be the matrix obtained from $X^*$ by removing the two rows corresponding to $v^*$.
Namely, by an appropriate column permutation and an edge orientation on $E'$, the matrix $A$ is written as
\begin{align}\label{eq:form_of_A}
  \begin{pmatrix}
    X_1 & O   & O         \\
    X_2 & X_3 & I_{2n-4k}
  \end{pmatrix},
\end{align}
where each two columns of the left, middle and right blocks correspond to an edge in $\tilde{E}$, $E \setminus \tilde{E}$ and $E'$, respectively.
Regarding $E^* \defeq E \cup E'$ as the set of lines on $A$, we set a line weight $\funcdoms{w}{E^*}{\setR}$ as $w(e) \defeq l(e)$ for $e \in E$ and as $w(e) = 0$ for $e \in E'$.

\begin{lemma}\label{lem:stu_correspondence}
  The minimum length of a disjoint \calS-path of $G$ with respect to $l$ is equal to the minimum weight of a parity base of $(A, E^*)$ with respect to $w$.
  In addition, there is a one-to-one correspondence between shortest disjoint \calS-paths of $G$ and minimum-weight parity bases of $(A, E^*)$.
\end{lemma}

\begin{proof}
  The line of this discussion is almost the same as the proof of \cref{lem:st-one2one}.
  Since $A$ can be transformed into~\eqref{eq:form_of_A}, $B \subseteq E^*$ is a parity base of $(A, E^*)$ if and only if a submatrix of $X$ is nonsingular, where the rows and columns of the submatrix are ones associated with $V \setminus \set{v}[\set{v, v^*} \in B \cap E']$ and $B \cap E$, respectively.
  By \cref{lem:Yamaguchi2017}, this is equivalent to the condition that $B$ consists of $k+1$ connected components $B_0 = B \cap E', B_1, \ldots, B_k$ in $G^*$ such that $B_0$ contains $v^*$ and each of $B_1, \ldots, B_k$ is a tree covering exactly two terminals belonging to distinct parts of $\mathcal{S}$.
  Thus any parity base $B \in \pbase(A, E^*)$ contains a unique disjoint \calS-path $P$ of $G$.
  We have $l(P) \le w(B)$ by the nonnegativity of $w$.
  Conversely, given a disjoint \calS-path $P$, one can construct a parity base $B$ by adding $\set{v, v^*} \in E'$ for each $v \in \tilde{V}$ that is not covered by $P$.
  Then we have $w(B) = l(P)$ as $w(e) = 0$ for $e \in E'$.
  Hence the minimum length of a disjoint \calS-path of $G$ is equal to the minimum weight of a parity base of $(A, E^*)$.

  Any minimum-weight common base $B \in \pbase(A, E^*)$ contains a unique shortest disjoint \calS-path $P$ of $G$, as we have seen above.
  Conversely, for any shortest disjoint \calS-path, there exists a unique minimum-weight common base $B$ because $l(e)$ is positive for all $e \in E$.
  Hence the correspondence is one-to-one.
\end{proof}

We next give a formula that connects $\sgn P$ and $\det A[B]$.
For a disjoint \calS-path $P$, define
\begin{align}\label{def:c_P}
  c_P \defeq \prod_{i=1}^k \det
  \begin{pmatrix}
    b_{\alpha_{2i-1}} \\
    b_{\alpha_{2i}}
  \end{pmatrix}
  \in \setR \setminus \set{0},
\end{align}
where $\alpha_i$ is the element in $\intset{s}$ such that $u_{\sigma_P(i)} \in S_{\alpha_i}$ for $i \in \intset{2k}$.

\begin{lemma}\label{lem:stu_sign_correspondence}
  Let $P$ be a disjoint \calS-path of $G$ and $B$ a parity base of $(A, E^*)$ containing $P$.
  Then we have
  \begin{align}\label{eq:stu_sgn_P}
    c_P \sgn P = \det A[B].
  \end{align}
\end{lemma}

\begin{proof}
  As shown in the proof of \cref{lem:stu_correspondence}, $B$ consists of $k+1$ connected components $B_0 = B \cap E', B_1, \ldots, B_k$ in $G^*$ such that $B_0$ covers $v^*$ and $B_i$ contains the path $P_i \subseteq P$ between $u_{\sigma_P(2i-1)}$ and $u_{\sigma_P(2i)}$ for each $i \in \intset{k}$.
  Let $V_i$ be the vertex subset covered by $B_i$ for $i = 0, \ldots, k$.
  By row and column permutations, $A[B]$ is transformed into a block diagonal matrix $Z = \diag(Z_0, Z_1, \ldots, Z_k)$.
  For $i = 0, \ldots, k$, the columns of $Z_i$ correspond to $B_i$ and the rows of $Z_i$ are associated with $V_i$ if $i \ne 0$ and with $V_0 \setminus \set{v^*}$ if $i = 0$.
  We can assume that the permutations are taken so that they preserve the two rows and the two columns associated with each $v \in \tilde{V}$ and $e \in E^*$.
  This means that the column permutation is even.
  We also arrange the rows corresponding to $u_{\sigma_P(2i-1)}$ and $u_{\sigma_P(2i)}$ on the first and second rows of $Z_i$, respectively, for $i \in \intset{k}$.
  Since the first $2k$ rows of $A$ correspond to $u_1, \ldots, u_{2k}$ from top to bottom, the sign of the row permutation coincides with $\sgn P$.
  Hence we have
  \begin{align}\label{eq:eq:stu_sgn_P_mid}
    \det A[B] = \sgn P \prod_{i=0}^k \det Z_k.
  \end{align}

  Applying row and column permutations preserving each two rows and columns, $Z_0$ becomes a block diagonal matrix whose diagonal blocks are $I_2$ or $-I_2$.
  Thus $\det Z_0 = 1$.
  Consider $Z_i$ for $i \in \intset{k}$.
  Let $\alpha, \beta \in \intset{s}$ with $u_{\sigma_P(2i-1)} \in S_\alpha$, $u_{\sigma_P(2i)} \in S_\beta$.
  Permuting row and columns and reversing the signs of some lines (two consecutive columns) in $B_i \cap \tilde{E}$, we can transform $Z_i$ into
  \begin{align}\label{eq:Zi}
    \prn{\begin{array}{ccccc|c}
        b_\alpha &      &        &      &         & \multirow{6}{*}{$C$} \\
                 &      &        &      & b_\beta & \\
        I_2      & -I_2 &        &      &         & \\
                 & I_2  & \ddots &      &         & \\
                 &      & \ddots & -I_2 &         & \\
                 &      &        & I_2  & I_2     & \\\hline
                 &      &        &      &         & I_{2p}
      \end{array}},
  \end{align}
  for some matrix $C$, where empty cells indicate zero and $p \defeq \card{B_i \setminus P_i}$.
  Note that these permutations and change of sign retain the determinant again.
  The determinant of~\eqref{eq:Zi} is equal to $\det \begin{psmallmatrix} b_\alpha \\ b_\beta \end{psmallmatrix}$.
  Hence~\eqref{eq:stu_sgn_P} holds via~\eqref{eq:eq:stu_sgn_P_mid}.
\end{proof}

We say that $\mathcal{S}$ is in the \emph{LGV position} on $G$ if $\sgn P$ is constant for all disjoint \calS-path $P$ of $G$.
Since $\det A[B]$ depends not only on $\sgn P$ but on $c_P$ by \cref{lem:stu_sign_correspondence}, the matroid parity $(A, E^*)$ might not be Pfaffian even if $\mathcal{S}$ is in the LGV position.
Nevertheless, $c_P$ is constant when $\card{\mathcal{S}} = 3$, as claimed in the following.
We refer to the (shortest) disjoint \calS-path problem with $\mathcal{S} = \set{S = S_1, T = S_2, U = S_3}$ as the (\emph{shortest}) \emph{disjoint $S$--$T$--$U$ path problem}.
An \emph{$S$--$T$--$U$ path} means an $\set{S, T, U}$-path.

\begin{lemma}\label{lem:stu_c_P_constant}
  Let $G = (V, E)$ be an undirected graph and $S, T, U \subseteq V$ disjoint vertex subsets.
  Then $c_P$ defined by~\eqref{def:c_P} is constant for all disjoint $S$--$T$--$U$ path $P$ of $G$.
\end{lemma}
\begin{proof}
  Let $P$ be a disjoint $S$--$T$--$U$ path of $G$ and $x$, $y$, and $z$ the numbers of paths in $P$ connecting $S$ to $T$, $T$ to $U$, and $U$ to $S$, respectively.
  We have $x + y = \card{T}$, $y + z = \card{U}$, and $z + x = \card{S}$.
  This means that $x, y, z$ are uniquely determined from $\card{S}, \card{T}, \card{U}$.
  The value of $c_P$ is equal to
  \begin{align}
    c_P =
    \prn{\det \begin{pmatrix} b_1 \\ b_2 \end{pmatrix}}^x
    \prn{\det \begin{pmatrix} b_2 \\ b_3 \end{pmatrix}}^y
    \prn{\det \begin{pmatrix} b_1 \\ b_3 \end{pmatrix}}^z,
  \end{align}
  where $b_1, b_2$ and $b_3$ are two-dimensional row vectors corresponding to $S$, $T$, and $U$, respectively.
  Hence $c_P$ does not depend on the choice of $P$.
\end{proof}

We have the following conclusion by \cref{lem:stu_correspondence,lem:stu_sign_correspondence,lem:stu_c_P_constant}.

\begin{theorem}
  Let $G$ be an undirected graph and take disjoint vertex subsets $S, T, U$.
  Let $(A, E^*)$ be the matroid parity defined above.
  When $\set{S, T, U}$ is in the LGV position, then $(A, E^*)$ is Pfaffian with constant $c_P \sgn P$, where $P$ is an arbitrary disjoint $S$--$T$--$U$ path of $G$.
  In addition, when $G$ is equipped with a positive edge length $l$, there is a one-to-one correspondence between shortest disjoint $S$--$T$--$U$ paths of $G$ and minimum-weight parity bases of $(A, E^*)$ with respect to the line weight $w$ defined above.
\end{theorem}

We show one example of the LGV position for the $S$--$T$--$U$ case.
Let $G$ be a planar graph and suppose that terminals are aligned on the boundary of one face of $G$ in the order of $S$, $U$, $T$ clockwise.
Then as shown in the proof of \cref{lem:stu_c_P_constant}, the connecting pattern of terminals are uniquely determined from $\card{S}$, $\card{T}$ and $\card{U}$.
Hence $\sigma_P$ is constant for all disjoint $S$--$T$--$U$ path $P$ of $G$, which means that $\set{S, T, U}$ is in the LGV position.
See \cref{fig:LGV-STU} for an illustration.

\tikzset{u/.style={draw, fill=gray!50, circle, inner sep=0pt, minimum size=4pt}}

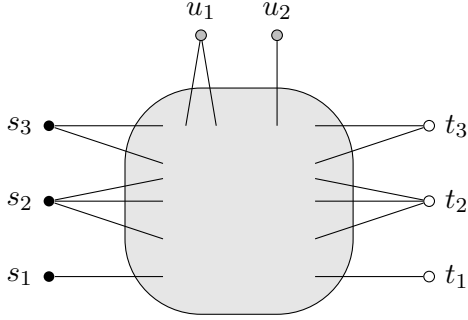
\begin{figure}
  \centering
  \begin{tikzpicture}
    \node (s1) [s] at (0, 0.5) {}; \node [left=0mm of s1] {$s_1$};
    \node (s2) [s] at (0, 1.5) {}; \node [left=0mm of s2] {$s_2$};
    \node (s3) [s] at (0, 2.5) {}; \node [left=0mm of s3] {$s_3$};
    \node (t1) [t] at (5, 0.5) {}; \node [right=0mm of t1] {$t_1$};
    \node (t2) [t] at (5, 1.5) {}; \node [right=0mm of t2] {$t_2$};
    \node (t3) [t] at (5, 2.5) {}; \node [right=0mm of t3] {$t_3$};

    \node (u1) [u] at (2, 3.7) {}; \node [above=0mm of u1] {$u_1$};
    \node (u2) [u] at (3, 3.7) {}; \node [above=0mm of u2] {$u_2$};

    \draw [rounded corners=1cm,fill=gray!20] (1,0) rectangle ++(3,3) node [midway] {} ;

    \draw (s1) -- (1.5, 0.5);
    \draw (s2) -- (1.5, 1);
    \draw (s2) -- (1.5, 1.5);
    \draw (s2) -- (1.5, 1.8);
    \draw (s3) -- (1.5, 2);
    \draw (s3) -- (1.5, 2.5);

    \draw (t1) -- (3.5, 0.5);
    \draw (t2) -- (3.5, 1);
    \draw (t2) -- (3.5, 1.5);
    \draw (t2) -- (3.5, 1.8);
    \draw (t3) -- (3.5, 2);
    \draw (t3) -- (3.5, 2.5);

    \draw (u1) -- (1.8, 2.5);
    \draw (u1) -- (2.2, 2.5);
    \draw (u2) -- (3, 2.5);
  \end{tikzpicture}
  \caption{%
    Example of $S, T, U$ that are in the LGV position, where $S = \set{s_1, s_2, s_3}$, $T = \set{t_1, t_2, t_3}$ and $U = \set{u_1, u_2}$.
    The gray area represents a planar graph.
  }\label{fig:LGV-STU}
\end{figure}

\section{Algorithms}\label{sec:algorithms}

In this section, we present algorithms for Pfaffian pairs and parities.
In \cref{sec:unweighted_counting}, we see that the current fastest randomized algorithms~\cite{Cheung2014,Harvey2009} for the linear matroid intersection and parity problems can be derandomized for Pfaffian pairs and parities with $\ch(\setK) = 0$.
\Cref{sec:counting_on_weighted_pfaffian_pairs,sec:counting_on_weighted_pfaffian_parities} describe counting algorithms for minimum-weight common bases of Pfaffian pairs and minimum-weight parity bases of Pfaffian parities, respectively.
Unless otherwise stated, we deal with matrices over a field $\setK$ of characteristic $\ch(\setK)$.
We assume that we can perform arithmetic operations on $\setK$ in constant time.

\subsection{Counting on Unweighted Pfaffian Pairs and Parities}\label{sec:unweighted_counting}

\Cref{prop:counting_formula_for_pfaffian_pairs,prop:counting_formula_for_pfaffian_prities} claim that the number of common bases of a Pfaffian pair $(A_1, A_2)$ of constant $c$ is equal to $c^{-1} \det A_1 \trsp{A_2}$ and the number of parity bases of a Pfaffian parity $(A, L)$ of constant $c$ is equal to $c^{-1} \pf A \Delta(\onevec) \trsp{A}$, both over $\setK$.
Therefore, if the value of $c$ is already known, we can compute these quantities just by performing matrix computations.

\begin{theorem}\label{thm:complexity_c_known}
  Suppose that we are given the value of constants.
  Then the following hold:
  \begin{enumerate}
    \item We can compute the number of common bases of an $r \times n$ Pfaffian pair modulo $\ch(\setK)$ in deterministic $\Order\prn{nr^{\omega-1}}$-time.\label{item:complexity_c_known_1}
    \item We can compute the number of parity bases of a $2r \times 2n$ Pfaffian parity modulo $\ch(\setK)$ in deterministic $\Order\prn{nr^{\omega-1} + r^3}$-time.
          When $\ch(\setK) = 0$, the running time can be improved to $\Order\prn{nr^{\omega-1}}$.\label{item:complexity_c_known_2}
  \end{enumerate}
\end{theorem}
\begin{proof}
  \ref{item:complexity_c_known_1} % chktex 2
  For an $r \times n$ Pfaffian pair $(A_1, A_2)$, we can compute the matrix multiplication $A_1 \trsp{A_2}$ in $\Order\prn{nr^{\omega-1}}$-time and its determinant in $\Order\prn{r^\omega}$-time~\cite[Theorem~6.6]{Aho1974}.

  \ref{item:complexity_c_known_2} % chktex 2
  Let $(A, L)$ be a $2r \times 2n$ Pfaffian parity.
  We rewrite $A \Delta(\onevec) \trsp{A}$ as
  \begin{align}\label{eq:reqrite_ADeltaA_onevec}
    A \Delta(\onevec) \trsp{A} = A_1 \trsp{A_2} - A_2 \trsp{A_1}
  \end{align}
  using~\eqref{eq:rewrite_ADeltaA}, where $A_1, A_2 \in \setK^{2r \times n}$ are the submatrices of $A$ defined in the same way as~\eqref{eq:rewrite_ADeltaA}.
  Then we can compute $A \Delta(\onevec) \trsp{A}$ in $\Order\prn{nr^{\omega-1}}$-time through~\eqref{eq:reqrite_ADeltaA_onevec}.
  The computation of Pfaffian requires $\Order(r^3)$-time via the naive Gaussian elimination.
  While we can compute the square of the Pfaffian using the fast determinant computation through~\eqref{eq:pf_det}, in general we cannot determine which of two square roots of $\det A \Delta(\onevec) \trsp{A}$ is $\pf A \Delta(\onevec) \trsp{A}$.
  When $\ch(\setK) = 0$, we have $c^{-1} \pf A \Delta(\onevec) \trsp{A} \ge 0$ since it is the cardinality of $\pbase(A, L)$.
  Thus we can compute the Pfaffian in $\Order(r^\omega)$-time.
\end{proof}

In practice, the value of $c$ can be typically retrieved from the reduction of discrete structures to Pfaffian pairs or parities, as seen in \cref{sec:examples,sec:lgv}.
However, if this is not the case, one common or parity base $B$ should be obtained to compute $c$.
Solving linear matroid intersection and parity problems will make the running time larger than $\Order\prn{nr^{\omega-1}}$ if we stick to deterministic algorithms~\cite{Gabow1996,Orlin2008}, or we can employ randomized algorithms~\cite{Cheung2014,Harvey2009} to keep the running time.
We face the same trade-off to find one common or parity base even if we know $c$.

Indeed, for Pfaffian pairs and parities with $\ch(\setK) = 0$, we can derandomize the linear matroid intersection algorithm of Harvey~\cite{Harvey2009} and the linear matroid parity algorithm of Cheung--Lau--Leung~\cite{Cheung2014}.
In these algorithms, randomness is used only to find a vector over $\setK$ satisfying some genericity conditions, summarized below.
See~\cite[Section~4]{Harvey2009} and~\cite[Section~6]{Cheung2014} for details on their algorithms.

Let $(A_1, A_2)$ be a matrix pair with common column set $E$.
A column subset $J \subseteq E$ is said to be \emph{extensible} if there exists a common base of $(A_1, A_2)$ containing $J$.
Similarly, for a matroid parity $(A, L)$, we call $J \subseteq L$ \emph{extensible}%
\footnote{
  Cheung--Lau--Leung~\cite{Cheung2014} call such $J$ ``growable.''
}
if there exists a parity base of $(A, L)$ containing $J$.
For a vector $z = \prn{z_j}_{j \in E}$ and $J \subseteq E$, let $\phi_J(z)$ denote a vector whose each component is defined as
\begin{align}
  {\phi_J(z)}_j \defeq \begin{cases}
    0   & (j \in J),            \\
    z_j & (j \in E \setminus J)
  \end{cases}
\end{align}
for $j \in E$.
We also define $\phi_J(z)$ for $z = \prn{z_l}_{l \in L}$ and $J \subseteq L$ in the same way.
Recall that matrices $\Xi(z)$ and $\Phi(z)$ are defined by~\eqref{def:Xi} and by~\eqref{def:Phi}, respectively.

\begin{lemma}[{\cite{Cheung2014,Harvey2009}}]\label{lem:validity_of_derandomization}
  The following hold.
  \begin{enumerate}
    \item Let $(A_1, A_2)$ be an $r \times n$ matrix pair with common column set $E$.
          Suppose that we are given a vector $z = \prn{z_j}_{j \in E} \in \setK^n$ such that $J$ is extensible if and only if $\Xi\prn{\phi_J(z)}$ is nonsingular for every $J \subseteq E$.
          Then we can construct a common base of $(A_1, A_2)$ (if it exists) in deterministic $\Order\prn{nr^{\omega-1}}$-time.\label{item:validity_of_derandomization_pairs}
    \item Let $(A, L)$ be a $2r \times 2n$ matroid parity.
          Suppose that we are given a vector $z = \prn{z_l}_{l \in L} \in \setK^n$ such that $J$ is extensible if and only if $\Phi\prn{\phi_J(z)}$ is nonsingular for every $J \subseteq L$.
          Then we can construct a parity base of $(A, L)$ (if it exists) in deterministic $\Order\prn{nr^{\omega-1}}$-time.\label{item:validity_of_derandomization_parities}
  \end{enumerate}
\end{lemma}
It is shown in~\cite[Theorem~4.4]{Harvey2009} and in~\cite[Theorem~6.4]{Cheung2014} that a vector of distinct indeterminates satisfies the requirements of $z$ in \cref{lem:validity_of_derandomization}.
The algorithms of Harvey~\cite{Harvey2009} and Cheung--Lau--Leung~\cite{Cheung2014} use a random vector over $\setK$ instead of indeterminates to avoid symbolic computations.
For Pfaffian pairs and parities, we can use $\onevec$ for $z$ as follows.

\begin{lemma}\label{lem:derandomizing_for_pfaffian}
  Let $\setK$ be a field of characteristic zero.
  For Pfaffian pairs and parities over $\setK$, we can choose $z = \onevec$ in \cref{lem:validity_of_derandomization}.
\end{lemma}
\begin{proof}
  Let $(A_1, A_2)$ be a Pfaffian pair of constant $c$ with common column set $E$.
  By \cref{prop:counting_formula_for_pfaffian_pairs}, we have
  \begin{align}
    \det \Xi\prn{\phi_J(\onevec)}
    = c^{-1} \sum_{B \in \cbase(A_1, A_2)} \prod_{j \in E \setminus B} {\phi_J(\onevec)}_j
    = c^{-1} \card{\set{B \in \cbase(A_1, A_2)}[J \subseteq B]}
  \end{align}
  for $J \subseteq E$.
  Since $\ch(\setK) = 0$, the set $\set{B \in \cbase(A_1, A_2)}[J \subseteq B]$ is nonempty if and only if its cardinality is nonzero over $\setK$.
  Hence the nonsingularity of $\Xi(\phi_J(\onevec))$ is equivalent to the extensibility of $J$.

  The same argument can be applied to Pfaffian parities by using \cref{prop:counting_formula_for_pfaffian_prities}.
  Let $(A, L)$ be a Pfaffian parity of constant $c$.
  Then
  \begin{align}
    \pf \Phi\prn{\phi_J(\onevec)}
    = c^{-1} \sum_{B \in \pbase(A, L)} \prod_{l \in L \setminus B} {\phi_J(\onevec)}_l
    = c^{-1} \card{\set{B \in \pbase(A, L)}[J \subseteq B]}
  \end{align}
  for $J \subseteq L$.
  Thus $J$ is extensible if and only if $\Phi(\phi_J(\onevec))$ is nonsingular.
\end{proof}

The proof of \cref{lem:derandomizing_for_pfaffian} can also be seen as alternative simple proofs of~\cite[Theorem~4.4]{Harvey2009} and~\cite[Theorem~6.4]{Cheung2014}.
Now \cref{thm:complexity_of_unweighted_counting} is obtained as a conclusion of \cref{thm:complexity_c_known,lem:validity_of_derandomization,lem:derandomizing_for_pfaffian}.

\subsection{Counting on Weighted Pfaffian Pairs}\label{sec:counting_on_weighted_pfaffian_pairs}
Let $(A_1, A_2)$ be an $r \times n$ weighted Pfaffian pair of constant $c$ with column weight $\funcdoms{w}{E}{\setR}$.
In this section, we consider counting the number of minimum-weight common bases of $(A_1, A_2)$ over $\setK$.
While we can compute it by naively expanding $\det A_1 D\prn{\theta^w} \trsp{A_2}$ or $\det \Xi\prn{\theta^w}$ from \cref{prop:algebraic_weighted_pair}, this expansion requires pseudo-polynomial time with respect to the maximum absolute value of a weight (assuming $w$ to be integral).
Instead, we reduce the problem to the counting on an unweighted Pfaffian pair.

We introduce some notions to make our descriptions rigorous.
Since $w$ is real-valued, the determinant and entries of $A_1 D\prn{\theta^w} \trsp{A_2}$ are a formal $\setK$-linear combination $f(\theta)$ of real powers of $\theta$.
Namely, $f(\theta)$ is formally expressed as
\begin{align}
  f(\theta) = \sum_{x \in X} a_x \theta^x
\end{align}
with finite $X \subseteq \setR$ and $a_x \in \setK$ for $x \in X$.
Abusing terminology, we call $f(\theta)$ a \emph{polynomial} in $\theta$.
We define the \emph{degree} $\deg f(\theta)$ and the \emph{order} $\ord f(\theta)$ of $f(\theta)$ as the maximum and the minimum $x \in X$ such that $a_x \ne 0$, respectively.
We set $\deg 0 \defeq -\infty$ and $\ord 0 \defeq +\infty$ for convenience.
The \emph{constant term} of $f(\theta)$ means $a_0$.

We begin to describe the algorithm.
Suppose that $(A_1, A_2)$ has at least one common base and we have obtained a minimum-weight common base $B \in \cbase(A_1, A_2)$ by solving the weighted linear matroid intersection problem.
We first perform row transformations on $A_1$ and $A_2$ so that $A_1[B]$ and $A_2[B]$ become the identity matrix $I_r$.
This operation, called \emph{pivoting}, remains $(A_1, A_2)$ Pfaffian but changes its constant to 1.
Now we can regard the row sets of $A_1$ and $A_2$ as $B$ since $A_1[B]$ and $A_2[B]$ are identity.

Frank's \emph{weight splitting lemma}~\cite{Frank1981} reveals the dual structure of the weighted matroid intersection problem.
It claims that there exist $\funcdoms{w_1, w_2}{E}{\setR}$ such that

\begin{enumerate}[label={(W\arabic*)}]
  \item $w_1(j) + w_2(j) = w(j)$ for $j \in E$, and\label{item:W1}
  \item a common base $B' \in \cbase(A_1, A_2)$ minimizes the weight $w$ if and only if $B'$ minimizes $w_1$ among all bases of $A_1$ and $B'$ minimizes $w_2$ among all bases of $A_2$.\label{item:W2}
\end{enumerate}

Let $w_1, w_2$ be split weights satisfying~\ref{item:W1} and~\ref{item:W2}.
The following observation is easy but important.

\begin{proposition}\label{prop:weighted_intersection_cocircuit}
  For $k = 1, 2$, $u \in B$ and $j \in E$, if the $(u,j)$th entry of $A_k$ is nonzero, then it holds $w_k(u) \le w_k(j)$.
\end{proposition}
\begin{proof}
  By $A_k[B] = I_r$, the set $B' \defeq B \setminus \set{u} \cup \set{j}$ is a base of $A_k$.
  If $w_k(u) > w_k(j)$, we have $w(B') < w(B)$, which contradicts to~\ref{item:W2} and the minimality of $B$.
\end{proof}

For $k = 1,2$, let $A_k^\# \in \setK^{r \times n}$ be the matrix with row set $B$ and column set $E$ whose the $(u,j)$th entry $\pbig{A_k^\#}_{u,j}$ is defined by
\begin{align}\label{def:A_k_sharp}
  \pbig{A_k^\#}_{u,j} \defeq \begin{cases}
    \text{the $(u,j)$th entry of $A_k$} & (w_k(u) = w_k(j)), \\
    0 & (\text{otherwise})
  \end{cases}
\end{align}
for $u \in B$ and $j \in E$.

\begin{lemma}\label{lem:pair_sharp}
  The set of minimum-weight common bases of $(A_1, A_2)$ with respect to $w$ is equal to $\cbase\pbig{A_1^\#, A_2^\#}$.
  In addition, $\pbig{A_1^\#, A_2^\#}$ is Pfaffian with constant 1.
\end{lemma}
\begin{proof}
  We first show that $\base\pbig{A_k^\#}$ is the set of minimum-weight bases of $A_k$ with respect to the weight $w_k$ for $k = 1,2$.
  Then the first claim of the lemma follows from~\ref{item:W2}.

  Define $\funcdoms{p}{B}{\setR}$ by $p(u) \defeq -w_k(u)$ for $u \in B$ and put $\tilde{A}_k(\theta) \defeq D\pbig{\theta^p} A_k D\pbig{\theta^{w_k}}$, where $\theta$ is an indeterminate.
  Note that each nonzero entry in $\tilde{A}_k(\theta)$ is a ``monomial,'' i.e., its degree and order are the same.
  Take $B' \subseteq E$ with $\card{B'} = r$.
  Then we have $\det \tilde{A}_k(\theta)[B'] = \theta^{w_k(B') - w_k(B)} \det A_k[B']$.
  Thus, $B'$ is a minimum-weight base of $A_k$ with respect to $w_k$ if and only if $\det \tilde{A}_k(\theta)[B']$ is in $\setK \setminus \set{0}$.
  From \cref{prop:weighted_intersection_cocircuit}, the degree of each nonzero entry in $\tilde{A}_k(\theta)$ is nonnegative; this implies that any term of positive degree in $\tilde{A}_k(\theta)[B']$ cannot contribute to the constant term of $\det \tilde{A}_k(\theta)[B']$.
  In addition, the constant term of each entry in $\tilde{A}_k(\theta)$ is the same as that of $A_k^\#$ by its definition.
  Hence the constant term of $\det \tilde{A}_k(\theta)[B']$ is equal to $\det A_k^\#[B']$, which means that $B'$ is in $\base\pbig{A_k^\#}$ if and only if $B'$ minimizes $w_k$ among $\base(A_k)$.

  In the above argument, $\det A_k^\#[B'] = \det \tilde{A}_k(\theta)[B'] = \det A_k[B']$ is proved for $B' \in \base\pbig{A_k^\#}$.
  Hence we have $\det A_1^\#[B'] \det A_2^\#[B'] = \det A_1[B'] \det A_2[B'] = 1$ for all $B' \in \cbase\pbig{A_1^\#, A_2^\#}$.
\end{proof}

By \cref{lem:pair_sharp} and \cref{prop:counting_formula_for_pfaffian_pairs}, we have the following corollary, which leads us to an algorithm described in \cref{alg:weighted_pfaffian_pair}.

\begin{corollary}\label{cor:counting_via_sharp}
  The number of minimum-weight common bases of $(A_1, A_2)$ modulo is equal to $\det A_1^\# \trsp{{A_2^\#}}$ over $\setK$.
\end{corollary}

\begin{algorithm}[tbp]
  \caption{Computing the number of minimum-weight common bases of a Pfaffian pair.}\label{alg:weighted_pfaffian_pair}
  \begin{algorithmic}[1]
    \Input{An $r \times n$ Pfaffian pair $(A_1, A_2)$ and a column weight $\funcdoms{w}{E}{\setZ}$}
    \Output{The number of minimum-weight common bases of $(A_1, A_2)$ modulo $\ch(\setK)$}
    %\If{$\cbase(A_1, A_2) = \varnothing$}
    %\State{\Return 0} % chktex 1
    %\EndIf{}
    \State{Compute a minimum-weight common base $B \in \cbase(A_1, A_2)$ and split weights $w_1, w_2$}
    \State{$A_1 \gets {A_1[B]}^{-1}A_1, \, A_2 \gets {A_2[B]}^{-1}A_2$}
    \State{Construct the matrices $A_1^\#$ and $A_2^\#$ defined by~\eqref{def:A_k_sharp}}
    \State{\Return $\det A_1^\# \trsp{{A_2^\#}}$} % chktex 1
  \end{algorithmic}
\end{algorithm}

Now \cref{thm:complexity_of_counting_weighted_pairs} can be proved as follows.

\begin{proof}[{of \cref{thm:complexity_of_counting_weighted_pairs}}]
  The validity of \cref{cor:counting_via_sharp} is proved in the above arguments.
  We analyze its time complexity.
  Frank's weighted matroid intersection algorithm~\cite{Frank1981} can be implemented for linear matroids in $\Order\prn{nr^\omega + nr \log n}$-time (see, e.g.,~\cite[Corollarly~41.10a]{Schrijver2003}).
  Other computations can be done within this time.
\end{proof}

\subsection{Counting on Weighted Pfaffian Parities}\label{sec:counting_on_weighted_pfaffian_parities}

Let $(A, L)$ be a $2r \times 2n$ Pfaffian parity of constant $c$ with line weight $\funcdoms{w}{L}{\setR}$.
We describe an algorithm to count the number of minimum-weight parity bases of $(A, L)$ modulo $\ch(\setK)$.

%On formulating weighted Pfaffian parities algebraically, we use $\Phi(z)$ instead of $A \Delta(t) \trsp{A}$ as follows.
Suppose that $(A, L)$ has at least one parity base.
Let $\zeta$ denote the minimum weight of a parity base of $(A, L)$ and $N$ the number of minimum-weight parity bases modulo $\ch(\setK)$.
Note that $N$ is nonzero if $\ch(\setK) = 0$.
We put $\delta \defeq w(L) - \zeta$.
Then following holds from \cref{prop:algebraic_weighted_parity}.

\begin{lemma}\label{lem:algebraic_form_of_weighted_parity}
  The coefficient of $\theta^{\delta}$ in $\pf \Phi\pbig{\theta^w}$ is equal to $cN$.
  In addition, it holds $\delta \ge \deg \pf \Phi\pbig{\theta^w}$ and the equality is attained if and only if $N \ne 0$.
\end{lemma}

%Iwata--Kobayashi~\cite{Iwata2017} developed the only polynomial-time algorithm for the weighted linear matroid parity problem so far.
We first obtain a minimum-weight parity base $B \in \pbase(A, L)$ applying the algorithm of Iwata--Kobayashi~\cite{Iwata2017}.
Then we perform a row transformation and a line (column) permutation on $A$ so that the left $2r$ columns of $A$ correspond to $B$ and $A[B] = I_{2r}$.
Namely, $A$ is in the form of $A =\begin{pmatrix} I_{2r} & C \end{pmatrix}$ for some matrix $C \in \setK^{2r \times (2n - 2r)}$.
Note that these transformations retain $(A, L)$ Pfaffian but change the constant to constant 1.
We perform the same transformations on $\Phi\pbig{\theta^w}$ (and on $\Delta\pbig{\theta^w}$) accordingly.
Now the polynomial matrix $\Phi\pbig{\theta^w}$ is in the form of
\begin{align}
  \Phi\pbig{\theta^w} \defeq \prn{\begin{array}{c|cc}
      O         & I_{2r}                                                       & C \\\hline
      -I_{2r}   & \multicolumn{2}{c}{\multirow{2}{*}{$\Delta\pbig{\theta^w}$}}     \\
      -\trsp{C} & \multicolumn{2}{c}{}
    \end{array}}
  \begin{array}{l}
    \gets U \\
    \gets B \\
    \gets E \setminus B,
  \end{array}
\end{align}
where $U$ is the row set of $A$ identified with $B$.

Besides the minimum-weight parity base $B$, the algorithm of Iwata--Kobayashi~\cite{Iwata2017} output an extra matrix $C^*$.
Its row set $U^*$ and column set $E^*$ contains $U$ and $E \setminus B$, respectively, and elements in $U^* \setminus U$ and $E^* \setminus E = E^* \setminus (E \setminus B)$ are newly introduced ones.
The Schur complement of $C^*$ with respect to $Y \defeq C^*[U^* \setminus U, E^* \setminus E]$ coincides with $C$, i.e., it holds
\begin{align}\label{eq:C_star_Schur}
  C = C^*[U, E \setminus B] - C^*[U, E^* \setminus E]Y^{-1} C^*[U^* \setminus U, E \setminus B].
\end{align}
In addition, the cardinalities of $U^*$ and $E^*$ are guaranteed to be $\Order(n)$.
We put $W \defeq U^* \cup B \cup E^*$ and $c^* \defeq \det Y$.
Consider the skew-symmetric polynomial matrix $\Phi^*(\theta) = \pbig{\Phi_{u,v}^*(\theta)}_{u,v \in W}$ defined by
\begin{align}\label{def:phi_star_theta}
  \Phi^*(\theta) = \prn{\begin{array}{c|c|c|c|c}
      \multicolumn{2}{c|}{\multirow{2}{*}{$O$}}             & O                      & \multicolumn{2}{c}{\multirow{2}{*}{$C^*$}}                                           \\\cline{3-3}
      \multicolumn{2}{c|}{}                                 & I_{2r}                 & \multicolumn{2}{c}{}                                                                 \\\hline
      O                                                     & -I_{2r}                & \multicolumn{2}{c|}{\multirow{2}{*}{$\Delta\pbig{\theta^w}$}} & \multirow{2}{*}{$O$} \\\cline{1-2}
      \multicolumn{2}{c|}{\multirow{2}{*}{$-\trsp{{C^*}}$}} & \multicolumn{2}{c|}{}  &                                                                                      \\\cline{3-5}
      \multicolumn{2}{c|}{}                                 & \multicolumn{2}{c|}{O} & O
    \end{array}}
  \begin{array}{l}
    \gets U^* \setminus U \\
    \gets U               \\
    \gets B               \\
    \gets E \setminus B   \\
    \gets E^* \setminus E.
  \end{array}
\end{align}
Then we have the following claim, which is essentially the same as Claim~6.2 in the arXiv preprint of~\cite{Iwata2017}.

\begin{lemma}\label{lem:pfaffian_equality_phi_star}
  It holds $\pf \Phi^*(\theta) = c^* \pf \Phi\pbig{\theta^w}$.
\end{lemma}
\begin{proof}
  From the property~\eqref{eq:C_star_Schur} of $C^*$ on the Schur complement, we can transform $\Phi^*(\theta)$ by elementary operations as
  \begin{align}
    \hat{\Phi}(\theta) \defeq
    \prn{\begin{array}{c|c|c|c|c}
        \multicolumn{2}{c|}{\multirow{2}{*}{$O$}} & O                             & O                                                             & Y                    \\\cline{3-5}
        \multicolumn{2}{c|}{}                     & I_{2r}                        & C                                                             & O                    \\\hline
        O                                         & -I_{2r}                       & \multicolumn{2}{c|}{\multirow{2}{*}{$\Delta\pbig{\theta^w}$}} & \multirow{2}{*}{$O$} \\\cline{1-2}
        O                                         & \raisebox{-.3ex}{$-\trsp{C}$} & \multicolumn{2}{c|}{}                                         &                      \\\hline
        \raisebox{-.3ex}{$-\trsp{Y}$}             & O                             & \multicolumn{2}{c|}{O}                                        & O
      \end{array}}
    \begin{array}{l}
      \gets U^* \setminus U \\
      \gets U               \\
      \gets B               \\
      \gets E \setminus B   \\
      \gets E^* \setminus E.
    \end{array}
  \end{align}
  Then we have
  \begin{align}
    \pf \Phi^*(\theta)
    = \pf \hat{\Phi}(\theta)
    = \pf \begin{pmatrix}
      O         & O                   & Y \\
      O         & \Phi\pbig{\theta^w} & O \\
      -\trsp{Y} & O                   & O
    \end{pmatrix}
    =
    \pf \begin{pmatrix}
      O         & Y & O                   \\
      -\trsp{Y} & O & O                   \\
      O         & O & \Phi\pbig{\theta^w}
    \end{pmatrix}
    =
    c^* \pf \Phi\pbig{\theta^w}.
  \end{align}
  Note that the permutation which we applied on the third equality is even since the order of $\Phi\pbig{\theta^w}$ is $2r + 2n$.
  Hence the claim holds.
\end{proof}

\Cref{lem:algebraic_form_of_weighted_parity} can be rephrased in terms of $\Phi^*(\theta)$ by using \cref{lem:pfaffian_equality_phi_star} as follows.

\begin{lemma}\label{lem:algebraic_form_of_weighted_parity_star}
  The coefficient of $\theta^\delta$ in $\pf \Phi^*(\theta)$ is equal to $c^*N$.
  In addition, it holds $\delta \ge \deg \pf \Phi^*(\theta^w)$ and the equality is attained if and only if $N \ne 0$.
\end{lemma}

We next define an undirected graph $G = G(\Phi^*)$ associated with $\Phi^*(\theta)$.
The vertex set of $G$ is $W$ and the edge set is given by
\begin{align}\label{def:graph_of_phi_star_theta}
  F \defeq \set[\big]{\set{u, v}}[u, v \in W, \, \Phi_{u,v}^*(\theta) \ne 0].
\end{align}
We set the weight of every edge $\set{u, v} \in F$ to $\deg \Phi_{u,v}^*(\theta)$.
Let $\hat{\delta}(\Phi^*)$ denote the maximum weight of a perfect matching of $G$.
We set $\hat{\delta}(\Phi^*) \defeq -\infty$ if $G$ has no perfect matching.
Here we put $\hat{\delta} \defeq \hat{\delta}(\Phi^*)$.
From the definition~\eqref{def:pfaffian} of Pfaffian, $\hat{\delta}$ serves as a combinatorial upper bound on $\deg \pf \Phi^*(\theta)$.
For later use, we define $G(S)$ and $\hat{\delta}(S)$ for any skew-symmetric polynomial matrix $S(\theta)$ in the same manner.

The dual problem of the maximum-weight perfect matching problem on $G$ is as follows (see~\cite{Iwata2017} and~\cite[Theorem~25.1]{Schrijver2003}):
\begin{align}
  \text{(D)} \quad
  \begin{array}{|c>{\hspace{-.5em}}l}
    \underset{\pi, \xi}{\text{minimize}} & \begin{array}[t]{>{\displaystyle}l}
      \sum_{u \in W} \pi(u) - \sum_{Z \in \Omega} \xi(Z)
    \end{array} \\
    \text{subject to}                    & \begin{array}[t]{>{\displaystyle}l>{\displaystyle}l}
      \pi(u) + \pi(v) - \sum_{Z \in \Omega_{u,v}} \xi(Z) \ge \deg \Phi^*_{u,v}(\theta) & (\set{u,v} \in F),  \\
      \xi(Z) \ge 0                                                                     & (Z \in \Omega),
    \end{array}
  \end{array}
\end{align}
where $\Omega \defeq \set{Z \subseteq W}[\text{$\card{Z}$ is odd and $\card{Z} \ge 3$}]$ and $\Omega_{u,v} \defeq \set{Z \in \Omega}[\card{Z \cap \set{u,v}} = 1]$ for $u,v \in W$.
The following claim is proved in~\cite{Iwata2017} as a key ingredient of the optimality certification on the weighted linear matroid parity problem.

\begin{proposition}[{\cite[Claim~6.3 in the arXiv preprint]{Iwata2017}}]\label{prop:feasibility_of_D2}
  There exists a feasible solution of (D) having the objective value $\delta$.
\end{proposition}

We make use of \cref{prop:feasibility_of_D2} for the purpose of counting.

\begin{lemma}\label{lem:tightness_of_parity}
  It holds $\delta \ge \hat{\delta} \ge \deg \pf \Phi^*(\theta)$.
  The equalities are attained if $N \ne 0$.
\end{lemma}
\begin{proof}
  We have $\delta \ge \hat{\delta}$ by \cref{prop:feasibility_of_D2} and the weak duality of (D).
  We also have $\hat{\delta} \ge \deg \pf \Phi^*(\theta)$ from the definition of Pfaffian.
  The equality condition is obtained from \cref{lem:algebraic_form_of_weighted_parity_star}.
\end{proof}

By \cref{lem:tightness_of_parity}, it holds $N = 0$ if $\delta > \hat{\delta}$.
Otherwise, our goal is to compute the coefficient of $\theta^\delta = \theta^{\hat{\delta}}$ in $\pf \Phi^*(\theta)$ by \cref{lem:algebraic_form_of_weighted_parity_star}.
This can be obtained by executing Murota's upper-tightness testing algorithm on combinatorial relaxation~\cite[Section~4.4]{Murota1995a} (with $\det$ replaced with $\pf$).

\begin{proposition}[{see~\cite[Section~4.4]{Murota1995a}}]\label{prop:combinatorial_relaxation}
  Let $S(\theta)$ be a $2n \times 2n$ skew-symmetric polynomial matrix.
  We can compute the coefficient of $\theta^{\hat{\delta}(S)}$ in $\pf S(\theta)$ in $\Order(n^3)$-time.
\end{proposition}

\begin{algorithm}[tbp]
  \caption{Computing the number of minimum-weight parity bases of a Pfaffian parity.}\label{alg:weighted_pfaffian_parity}
  \begin{algorithmic}[1]
    \Input{An $2r \times 2n$ Pfaffian parity $(A, L)$ and a line weight $\funcdoms{w}{L}{\setZ}$}
    \Output{The number of minimum-weight common bases of $(A_1, A_2)$ modulo $\ch(\setK)$}
    \State{Compute a minimum-weight parity base $B \in \pbase(A, L)$ and the matrix $C^*$}
    \State{Construct the matrix $\Phi^*\pbig{\theta}$ and the graph $G = G(\Phi^*)$}
    \State{Compute the maximum weight $\hat{\delta} = \hat{\delta}(\Phi^*)$ of a perfect matching of $G$}
    \If{$\delta \defeq w(B) > \hat{\delta}$}
    \State{\Return 0} % chktex 1
    \Else{}
    \State{Compute the coefficient $a$ of $\theta^{\hat{\delta}}$ in $\pf \Phi^*(\theta)$}
    \State{$c^* \defeq \pf C^*[U^* \setminus U, E^* \setminus E]$}
    \State{\Return ${c^*}^{-1}a$} % chktex 1
    \EndIf{}
  \end{algorithmic}
\end{algorithm}

\Cref{alg:weighted_pfaffian_parity} shows the entire procedure of our algorithm.
Its time complexity, which is stated in \cref{thm:counting_minimum_weight_parity_bases}, is analyzed as follows.

\begin{proof}[{of \cref{thm:counting_minimum_weight_parity_bases}}]
  The algorithm of Iwata--Kobayashi~\cite{Iwata2017} runs in $\Order\prn{n^3 r}$-time~\cite[Theorem~11.1]{Iwata2017}.
  The maximum-weight perfect matching problem can be solved in $\Order\prn{n^3}$-time as $\card{W} = \Order\prn{n}$; see~\cite[Section~26.3a]{Schrijver2003}.
  The coefficient of $\theta^{\hat{\delta}}$ in $\pf \theta^*(\theta)$ can also be computed in $\Order\prn{n^3}$-time by \cref{prop:combinatorial_relaxation}.
  Hence the total running time is dominated by $\Order\prn{n^3 r}$.
\end{proof}

In the above arguments, we have assumed that arithmetic operations on $\setK$ can be performed in constant time.
This assumption is reasonable when $\setK$ is a finite field of fixed order.
When $\setK = \setQ$, it has not been proved that a direct application of the algorithm by Iwata--Kobayashi~\cite{Iwata2017} does not swell the bit-lengths of intermediate numbers.
Instead, they showed that one can obtain a minimum-weight parity base $B$ by applying their algorithm over a sequence of finite fields.
However, since our counting algorithm requires not only $B$ but also $C^*$, we cannot directly execute our counting algorithm if we use the weighted linear matroid parity algorithm for $\setK = \setQ$ as a black-box.
Here, we describe a polynomial-time counting algorithm for $\setK = \setQ$, which is based on the same reduction to problems over finite fields as~\cite{Iwata2017}.

Let $(A, L)$ be a Pfaffian parity over $\setQ$ equipped with a line weight $\funcdoms{w}{L}{\setR}$.
Multiplying the product of denominators of entries in $A$, we may assume that entries in $A$ are integral.
Applying the weighted linear matroid parity algorithm for $\setK = \setQ$, we first compute the minimum weight $\eta$ of a parity base and the constant $c \in \setZ$ of $(A, L)$.
Let $\gamma$ be the maximum absolute value of the entries of $A$ and put $K \defeq \ceil{r \log (nr\gamma)} + 2$.
We compute $K$ smallest prime numbers $p_1, \ldots, p_K$ by the sieve of Eratosthenes.
Since $K$ is bounded by a polynomial of the bit-length of $A$ and $p_K = \Order(K \log K)$ by the prime number theorem, this computation can be done in polynomial time.

Let $N$ be the number of minimum-weight parity bases of $(A, L)$.
For $i \in \intset{K}$, we consider the problem of computing $cN$ modulo $p_i$.
We have $cN \equiv 0$ modulo $p_i$ if $p_i$ divides $c$.
Suppose the case when $c$ is not a multiple of $p_i$.
Since $\det A[B] = c$ for all $B \in \pbase(A, L)$, a line subset $B \subseteq L$ is a parity base of $(A, L)$ if and only if $B$ is a parity base of $(A^{p_i}, L)$, where $A^{p_i}$ is a matrix over $\GF(p_i)$ obtained by regarding each entry of $A$ as an element of $\GF(p_i)$.
Therefore, $cN$ modulo $p_i$ is equal to the number of parity bases of $(A^{p_i}, L)$ over $\GF(p_i)$.
We compute this quantity by applying \cref{alg:weighted_pfaffian_parity} to $(A^{p_i}, L)$.
Since arithmetic operations on $\GF(p_i)$ can be performed in polynomial time, this computation can also be done in polynomial time by \cref{thm:counting_minimum_weight_parity_bases}.

Using the Chinese remainder theorem and the Euclidean algorithm, we compute $cN$ modulo $\prod_{i=1}^K p_i$ from $cN$ modulo $p_1, \ldots, p_K$.
Then we have
\begin{align}
  2\abs{cN}+1
  < 4\abs{cN}
  \le 4r! \gamma^r \binom{n}{r}
  \le 4\prn{r\gamma n}^{r}
  \le 2^K
  \le \prod_{i=1}^K p_i,
\end{align}
which implies that $cN$ is uniquely determined from $cN$ modulo $\prod_{i=1}^K p_i$.
These arguments certificate the correctness of \cref{thm:bit_complexity}.

\section*{Acknowledgments}
The authors thank Satoru Iwata for his helpful comments, and Yusuke Kobayashi, Yutaro Yamaguchi, and Koyo Hayashi for discussions.
This work was supported by JST ACT-I Grant Number JPMJPR18U9, Japan, and Grant-in-Aid for JSPS Research Fellow Grant Number JP18J22141, Japan.

\bibliographystyle{abbrv}
\bibliography{references}

\begin{thebibliography}{10}

\bibitem{Aho1974}
A.~V. Aho, J.~E. Hopcroft, and J.~D. Ullman.
\newblock {\em {The Design and Analysis of Computer Algorithms}}.
\newblock Addison-Wesley, Reading, MA, 1974.

\bibitem{Anari2018}
N.~Anari, S.~O. Gharan, and C.~Vinzant.
\newblock {Log-concave polynomials, entropy, and a deterministic approximation
  algorithm for counting bases of matroids}.
\newblock In {\em Proceedings of the 59th Annual IEEE Symposium on Foundations
  of Computer Science (FOCS '18)}, pages 35--46, 2018.

\bibitem{Anari2019}
N.~Anari, K.~Liu, S.~O. Gharan, and C.~Vinzant.
\newblock {Log-concave polynomials II: High-dimensional walks and an FPRAS for
  counting bases of a matroid}.
\newblock In {\em Proceedings of the 51st Annual ACM Symposium on Theory of
  Computing (STOC '19)}, pages 1--12, 2019.

\bibitem{Bouchet1995}
A.~Bouchet.
\newblock {Coverings and delta-coverings}.
\newblock In E.~Balas and J.~Clausen, editors, {\em Proceedings of the 4th
  International Conference on Integer Programming and Combinatorial
  Optimization (IPCO '95)}, volume 920 of {\em Lecture Notes in Computer
  Science}, pages 228--243, Berlin Heidelberg, 1995. Springer.

\bibitem{Bouchet1998}
A.~Bouchet, J.~F. Geelen, and W.~H. Cunningham.
\newblock {Principally unimodular skew-symmetric matrices}.
\newblock {\em Combinatorica}, 18(4):461--486, 1998.

\bibitem{Broder1997}
A.~Z. Broder and E.~W. Mayr.
\newblock {Counting minimum weight spanning trees}.
\newblock {\em Journal of Algorithms}, 24(1):171--176, 1997.

\bibitem{Cai2017}
J.-Y. Cai and X.~Chen.
\newblock {\em {Complexity Dichotomies for Counting Problems}}, volume~1.
\newblock Cambridge University Press, Cambridge, 2017.

\bibitem{Cheung2014}
H.~Y. Cheung, L.~C. Lau, and K.~M. Leung.
\newblock {Algebraic algorithms for linear matroid parity problems}.
\newblock {\em ACM Transactions on Algorithms}, 10(3):1--26, 2014.

\bibitem{Colbourn1995}
C.~J. Colbourn, J.~S. Provan, and D.~Vertigan.
\newblock {The complexity of computing the Tutte polynomial on transversal
  matroids}.
\newblock {\em Combinatorica}, 15(1):1--10, 1995.

\bibitem{Edmonds1968}
J.~Edmonds.
\newblock {Matroid partition}.
\newblock In G.~B. Dantzig and A.~F. {Veinott, Jr.}, editors, {\em Mathematics
  of the Decision Sciences: Part I}, volume~11 of {\em Lectures in Applied
  Mathematics}, pages 335--345. AMS, Providence, RI, 1968.

\bibitem{Edmonds1970}
J.~Edmonds.
\newblock {Submodular functions, matroids, and certain polyhedra}.
\newblock In R.~Guy, H.~Hanani, N.~Sauer, and J.~Sch{\"{o}}nheim, editors, {\em
  Combinatorial Structures and Their Applications}, pages 69--87. Gordon and
  Breach, New York, NY, 1970.

\bibitem{Frank1981}
A.~Frank.
\newblock {A weighted matroid intersection algorithm}.
\newblock {\em Journal of Algorithms}, 2(4):328--336, 1981.

\bibitem{Frank2011}
A.~Frank.
\newblock {\em {Connections in Combinatorial Optimization}}.
\newblock Oxford Lecture Series in Mathematics and Its Applications. Oxford
  University Press, New York, NY, 2011.

\bibitem{Gabow1986}
H.~N. Gabow and M.~Stallmann.
\newblock {An augmenting path algorithm for linear matroid parity}.
\newblock {\em Combinatorica}, 6(2):123--150, 1986.

\bibitem{Gabow1996}
H.~N. Gabow and Y.~Xu.
\newblock {Efficient theoretic and practical algorithms for linear matroid
  intersection problems}.
\newblock {\em Journal of Computer and System Sciences}, 53(1):129--147, 1996.

\bibitem{Gallai1964}
T.~Gallai.
\newblock {Maximum-Minimum S{\"{a}}tze und verallgemeinerte Faktoren von
  Graphen}.
\newblock {\em Acta Mathematica Academiae Scientiarum Hungaricae}, 12:131--173,
  1964.

\bibitem{Geelen1995}
J.~F. Geelen.
\newblock {\em {Matchings, Matroids and Unimodular Matrices}}.
\newblock PhD thesis, University of Waterloo, Waterloo, ON, 1995.

\bibitem{Geelen2001}
J.~F. Geelen.
\newblock {\em {Matching Theory}}.
\newblock Lecture Notes from the Euler Institute for Discrete Mathematics and
  Its Applications. 2001.

\bibitem{Geelen2005}
J.~F. Geelen and S.~Iwata.
\newblock {Matroid matching via mixed skew-symmetric matrices}.
\newblock {\em Combinatorica}, 25(2):187--215, 2005.

\bibitem{Gessel1985}
I.~Gessel and G.~Viennot.
\newblock {Binomial determinants, paths, and hook length formulae}.
\newblock {\em Advances in Mathematics}, 58(3):300--321, 1985.

\bibitem{Goodall2011}
A.~Goodall and A.~{De Mier}.
\newblock {Spanning trees of 3-uniform hypergraphs}.
\newblock {\em Advances in Applied Mathematics}, 47(4):840--868, 2011.

\bibitem{Harvey2009}
N.~J.~A. Harvey.
\newblock {Algebraic algorithms for matching and matroid problems}.
\newblock {\em SIAM Journal on Computing}, 39(2):679--702, 2009.

\bibitem{Hayashi2018}
K.~Hayashi and S.~Iwata.
\newblock {Counting minimum weight arborescences}.
\newblock {\em Algorithmica}, 80(12):3908--3919, 2018.

\bibitem{Hirschman2004}
S.~Hirschman and V.~Reiner.
\newblock {Note on the Pfaffian matrix-tree theorem}.
\newblock {\em Graphs and Combinatorics}, 20(1):59--63, 2004.

\bibitem{Ishikawa1995}
M.~Ishikawa and M.~Wakayama.
\newblock {Minor summation formula of Pfaffians}.
\newblock {\em Linear and Multilinear Algebra}, 39(3):285--305, 1995.

\bibitem{Iwata2017}
S.~Iwata and Y.~Kobayashi.
\newblock {A weighted linear matroid parity algorithm}.
\newblock In {\em Proceedings of the 49th Annual ACM SIGACT Symposium on Theory
  of Computing (STOC '17)}, pages 264--276. ACM Press, 2017.
\newblock arXiv: \href{https://arxiv.org/abs/1905.13371}{\texttt{1905.13371}}.

\bibitem{Kasteleyn1961}
P.~W. Kasteleyn.
\newblock {The statistics of dimers on a lattice: I. The number of dimer
  arrangements on a quadratic lattice}.
\newblock {\em Physica}, 27(12):1209--1225, 1961.

\bibitem{Kirchhoff1847}
G.~Kirchhoff.
\newblock {Ueber die Aufl{\"{o}}sung der Gleichungen, auf welche man bei der
  Untersuchung der linearen Vertheilung galvanischer Str{\"{o}}me gef{\"{u}}hrt
  wird}.
\newblock {\em Annalen der Physik}, 148(12):497--508, 1847.

\bibitem{Lawler1976}
E.~L. Lawler.
\newblock {\em {Combinatorial Optimization: Networks and Matroids}}.
\newblock Holt, Rinehart and Winston, New York, NY, 1976.

\bibitem{LeGall2014}
F.~{Le Gall}.
\newblock {Powers of tensors and fast matrix multiplication}.
\newblock In {\em Proceedings of the 39th International Symposium on Symbolic
  and Algebraic Computation (ISSAC '14)}, pages 296--303, New York, NY, 2014.
  ACM.

\bibitem{Lindstrom1973}
B.~Lindstr{\"{o}}m.
\newblock {On the vector representations of induced matroids}.
\newblock {\em Bulletin of the London Mathematical Society}, 5(1):85--90, 1973.

\bibitem{Little1974}
C.~H.~C. Little.
\newblock {An extension of Kasteleyn's method of enumerating the 1-factors of
  planar graphs}.
\newblock In D.~A. Holton, editor, {\em Proceedings of the 2nd Australian
  Conference on Combinatorial Mathematics (ACCM '74)}, volume 403 of {\em
  Lecture Notes in Mathematics}, pages 63--72, Berlin Heidelberg, 1974.
  Springer.

\bibitem{Lovasz1979}
L.~Lov{\'{a}}sz.
\newblock {On determinants, matchings, and random algorithms}.
\newblock In L.~Budach, editor, {\em Fundamentals of Computation Theory}.
  Akademie-Verlag, Berlin, 1979.

\bibitem{Lovasz1980}
L.~Lov{\'{a}}sz.
\newblock {Matroid matching and some applications}.
\newblock {\em Journal of Combinatorial Theory, Series B}, 28(2):208--236,
  1980.

\bibitem{Mader1978}
W.~Mader.
\newblock {{\"{U}}ber die Maximalzahl kreuzungsfreier $H$-Wege}.
\newblock {\em Archiv der Mathematik}, 31(1):387--402, 1978.

\bibitem{Masbaum2002}
G.~Masbaum and A.~Vaintrob.
\newblock {A new matrix-tree theorem}.
\newblock {\em International Mathematics Research Notices},
  2002(27):1397--1426, 2002.

\bibitem{Maurer1976}
S.~B. Maurer.
\newblock {Matrix generalizations of some theorems on trees, cycles and
  cocycles in graphs}.
\newblock {\em SIAM Journal on Applied Mathematics}, 30(1):143--148, 1976.

\bibitem{Murota1995a}
K.~Murota.
\newblock {Computing the degree of determinants via combinatorial relaxation}.
\newblock {\em SIAM Journal on Computing}, 24(4):765--796, 1995.

\bibitem{Murota2000}
K.~Murota.
\newblock {\em {Matrices and Matroids for Systems Analysis}}, volume~20 of {\em
  Algorithms and Combinatorics}.
\newblock Springer, Berlin, 2010.

\bibitem{Orlin2008}
J.~B. Orlin.
\newblock {A fast, simpler algorithm for the matroid parity problem}.
\newblock In A.~Lodi, A.~Panconesi, and G.~Rinaldi, editors, {\em Proceedings
  of the 13th International Conference on Integer Programming and Combinatorial
  Optimization (IPCO '08)}, volume 5035 of {\em Lecture Notes in Computer
  Science}, pages 240--258, Berlin, 2008. Springer.

\bibitem{Oxley2011}
J.~G. Oxley.
\newblock {\em {Matroid Theory}}.
\newblock Oxford Graduate Texts in Mathematics. Oxford University Press, New
  York, NY, second edition, 2011.

\bibitem{Plummer1986}
M.~D. Plummer and L.~Lov{\'{a}}sz.
\newblock {\em {Matching Theory}}, volume~29 of {\em Annals of Discrete
  Mathematics}.
\newblock Elsevier, North Holland, 1986.

\bibitem{Robertson1999}
N.~Robertson, P.~D. Seymour, and R.~Thomas.
\newblock {Permanents, Pfaffian orientations, and even directed circuits}.
\newblock {\em Annals of Mathematics}, 150(3):929--975, 1999.

\bibitem{Schrijver2003}
A.~Schrijver.
\newblock {\em Combinatorial Optimization}, volume~24 of {\em Algorithms and
  Combinatorics}.
\newblock Springer, Berlin, 2003.

\bibitem{Snook2012}
M.~Snook.
\newblock {Counting bases of representable matroids}.
\newblock {\em Electronic Journal of Combinatorics}, 19(4), 2012.

\bibitem{Temperley1961}
H.~N.~V. Temperley and M.~E. Fisher.
\newblock {Dimer problem in statistical mechanics-an exact result}.
\newblock {\em Philosophical Magazine}, 6(68):1061--1063, 1961.

\bibitem{Tomizawa1974}
N.~Tomizawa and M.~Iri.
\newblock {An algorithm for determining the rank of a triple matrix product
  $AXB$ with application to the problem of discerning the existence of the
  unique solution in a network}.
\newblock {\em Electronics and Communications in Japan}, 57(11):50--57, 1974.

\bibitem{Tutte1948}
W.~T. Tutte.
\newblock {The dissection of equilateral triangles into equilateral triangles}.
\newblock {\em Mathematical Proceedings of the Cambridge Philosophical
  Society}, 44(4):463--482, 1948.

\bibitem{Tutte1965}
W.~T. Tutte.
\newblock {Menger's theorem for matroids}.
\newblock {\em Journal of Research of the National Bureau of Standards Section
  B}, 69(1--2):49--53, 1965.

\bibitem{Valiant1977}
L.~G. Valiant.
\newblock {The complexity of computing the permanent}.
\newblock {\em Theoretical Computer Science}, 8(2):189--201, 1979.

\bibitem{Vazirani1989}
V.~V. Vazirani.
\newblock {NC algorithms for computing the number of perfect matchings in
  $K_{3,3}$-free graphs and related problems}.
\newblock {\em Information and Computation}, 80(2):152--164, 1989.

\bibitem{Webb2004}
K.~P. Webb.
\newblock {\em {Counting Bases}}.
\newblock PhD thesis, University of Waterloo, Waterloo, ON, 2004.

\bibitem{Yamaguchi2016}
Y.~Yamaguchi.
\newblock {Shortest disjoint $\mathcal{S}$-paths via weighted linear matroid
  parity}.
\newblock In S.-H. Hong, editor, {\em Proceedings of the 27th International
  Symposium on Algorithms and Computation (ISAAC '16)}, volume~64 of {\em
  Leibniz International Proceedings in Informatics}, pages 63:1--63:13. Schloss
  Dagstuhl--Leibniz-Zentrum f{\"{u}}r Informatik, 2016.

\end{thebibliography}

\appendix
\section{Appendix}

\subsection{Counting on Weighted Pfaffian Pairs Revisited}

In this section, we derive \cref{alg:weighted_pfaffian_pair} in a different manner than \cref{sec:counting_on_weighted_pfaffian_pairs}.
While it is a bit roundabout approach, it might provide a connection between \cref{alg:weighted_pfaffian_pair} and the counting algorithm for weighted Pfaffian parities, explained in \cref{sec:counting_on_weighted_pfaffian_parities}.

Let $(A_1, A_2)$ be a Pfaffian pair with constant $c$ and column weight $\funcdoms{w}{E}{\setR}$.
Suppose that $(A_1, A_2)$ has at least one common base.
Let $\zeta$ be the minimum weight of a common base and $N$ the number of minimum-weight common bases modulo $\ch(\setK)$.
Define $P(\theta) = \pbig{P_{u,v}(\theta)}_{u \in U_1, v \in U_2} \defeq A_1 D\pbig{\theta^w} \trsp{A_2}$, where $\theta$ is an indeterminate and $U_1$ and $U_2$ are row sets of $A_1$ and $A_2$, respectively.
From \cref{prop:algebraic_weighted_pair}, we have the following algebraic characterization on $\zeta$ and $N$.

\begin{lemma}\label{lem:algebraic_form_of_weighted_intersection}
  The coefficient of $\theta^{\zeta}$ in $P(\theta)$ is equal to $cN$.
  In addition, it holds $\zeta \le \ord P(\theta)$ and the equality is attained if and only if $N \ne 0$.
\end{lemma}

Let $B \in \cbase(A_1, A_2)$ be a minimum-weight common base.
As in \cref{sec:counting_on_weighted_pfaffian_pairs}, we first apply pivoting to $A_1$ and $A_2$ so that $A_1[B] = A_2[B] = I_r$.
We perform the same row and column operation on $P(\theta)$ accordingly.
Now we can identify $U_1$ and $U_2$ with $B$ since $A_1[B]$ and $A_2[B]$ are identity.

Next, we construct a bipartite graph $G = (U_1 \cup U_2, F)$ from $P(\theta)$ as follows.
The vertex set $U_1 \cup U_2$ is bipartitioned as $\set{U_1, U_2}$.
The edge set $F$ is given by
\begin{align}
  F \defeq \set{(u,v)}[u \in U_1, v \in U_2, P_{u,v}(\theta) \ne 0]
\end{align}
and we set the weight of every edge $(u,v) \in F$ as $\ord P_{u,v}(\theta)$.
Let $\hat{\zeta}$ denote the minimum weight of a perfect matching of $G$.
If $G$ has no perfect matching, we let $\hat{\zeta} \defeq +\infty$.
Then it is easily observed from the definition~\eqref{def:determinant} of the determinant that $\hat{\zeta}$ serves as a combinatorial lower bound on $\ord \det P(\theta)$ (see, e.g.,~\cite[Proposition~2.1]{Murota1995a}), and hence on $\zeta$ if $N \neq 0$ by \cref{lem:algebraic_form_of_weighted_intersection}.
Indeed, these quantities satisfy the following relation.

\begin{lemma}\label{lem:tightness_of_intersection}
  It holds $\zeta \le \hat{\zeta} \le \ord \det P(\theta)$.
  The equalities are attained if $N \ne 0$.
\end{lemma}
\begin{proof}
  By \cref{lem:algebraic_form_of_weighted_intersection} and $\hat{\zeta} \le \ord \det P(\theta)$, it suffices to show $\zeta \le \hat{\zeta}$.
  To show the claim, we use the dual problem of the minimum-weight perfect bipartite matching problem on $G$, which is formulated as follows:
  \begin{align}
    \text{(DB)} \quad
    \begin{array}{|c>{\hspace{-.5em}}l}
      \underset{p_1,p_2}{\text{maximize}} & \begin{array}[t]{>{\displaystyle}l}
        \sum_{u \in B} p_1(u) + \sum_{v \in B} p_2(v)
      \end{array} \\
      \text{subject to}                   & \begin{array}[t]{>{\displaystyle}l>{\displaystyle}l}
        p_1(u) + p_2(v) \le \ord P_{u,v}(\theta) & ((u,v) \in F). \\
      \end{array}
    \end{array}
  \end{align}
  See~\cite[Theorem~17.5]{Schrijver2003} for example.
  Note again that $U_1$ and $U_2$ are identified with $B$.

  Using the split weight $w_1, w_2$ satisfying~\ref{item:W1} and~\ref{item:W2}, we define $p_1(u) \defeq w_1(u)$ and $p_2(v) \defeq w_2(v)$ for $u,v \in B$.
  We show that this $p_1$ and $p_2$ are feasible on (DB).
  For every $(u, v) \in F$, there exists $j \in E$ with $w(j) = \ord P_{u,v}(\theta)$ such that both the $(u,j)$th entry of $A_1$ and the $(v,j)$th entry of $A_2$ are nonzero.
  By \cref{prop:weighted_intersection_cocircuit}, we have $w_1(u) \le w_1(j)$ and $w_2(v) \le w_2(j)$.
  Thus $p_1(u) + p_2(v) = w_1(u) + w_2(v) \le w_1(j) + w_2(j) = w(j) = \ord P_{u,v}(\theta)$, where the third equality follows from~\ref{item:W1}.
  Hence $(p_1, p_2)$ is feasible.

  The value of the objective function with respect to $p_1$ and $p_2$ is
  \begin{align}
    \sum_{u \in B} p_1(u) + \sum_{v \in B} p_2(v)
    = \sum_{u \in B} w_1(u) + \sum_{v \in B} w_2(v)
    = w(B) = \zeta.
  \end{align}
  This equality means that $\zeta$ is no more than the optimal value of the dual program, and thus than $\hat{\zeta}$ by the weak duality of the linear program.
\end{proof}

If $\zeta < \hat{\zeta}$, then $N$ must be zero by \cref{lem:tightness_of_intersection}.
Otherwise, it holds $\zeta = \hat{\zeta} = \ord \det P(\theta)$ and $N$ is equal to the coefficient of $\theta^{\hat{\zeta}}$ in $\det P(\theta)$.
In the definition~\eqref{def:determinant} of the determinant of $P(\theta)$, every minimum-weight perfect matching of $G$ contributes to the coefficient of $\theta^{\hat{\zeta}}$ in $\det P(\theta)$.
In this sense, we can regard the computation of $N$ as a kind of ``counting operation'' on minimum-weight perfect matchings of the bipartite graph $G$.

Murota~\cite{Murota1995a} gave a characterization of the coefficient of $\theta^{\hat{\zeta}}$ in $\det P(\theta)$ (for a general polynomial matrix $P(\theta)$) as follows.
Let $(p_1, p_2)$ be a feasible solution of (DB).
The \emph{tight coefficient matrix} $P^\# = \prn{P_{u,v}^\#}_{u,v \in B}$ of $P(\theta)$ with respect to $(p_1, p_2)$ is defined by
\begin{align}
  P_{u,v}^\# \defeq \text{the coefficient of $\theta^{p_1(u)+p_2(v)}$ in $P_{u,v}(\theta)$}
\end{align}
for $u,v \in B$.

\begin{proposition}[{\cite[Propositions~2.4 and~2.6]{Murota1995a}}]\label{prop:combinatorial_relaxation_det}
  Let $(p_1, p_2)$ be a feasible solution of $\mathrm{(DB)}$.
  If $(p_1, p_2)$ is not optimal, then $P^\#$ is singular.
  If $(p_1, p_2)$ is optimal, then $\det P^\#$ is equal to the coefficient of $\theta^{\hat{\zeta}}$ in $\det P(\theta)$.
\end{proposition}

\Cref{prop:combinatorial_relaxation_det} essentially follows from the complementarity of the linear program.
At this point, we have obtained a polynomial-time algorithm for computing $N$ since $P^\#$ can be calculated in polynomial-time.
We need one more argument to reach to \cref{alg:weighted_pfaffian_pair}.

Let $w_1$ and $w_2$ be a split weight satisfying~\ref{item:W1} and~\ref{item:W2}.
For $k = 1,2$, let $A_k^\#$ be the matrix obtained from $A_k$ and $w_k$ by~\eqref{def:A_k_sharp}.
We also put $p_1(u) \defeq w_1(u)$ and $p_2(u) \defeq w_2(u)$ for $u \in B$.
Note that $(p_1, p_2)$ is feasible on (DB) as shown in the proof of \cref{lem:tightness_of_intersection}.

\begin{lemma}\label{lem:from_A_sharp_to_tcm}
  The matrix $A_1^\# \trsp{{A_2^\#}}$ is equal to the tight coefficient matrix of $P(\theta)$ with respect to $(p_1, p_2)$.
\end{lemma}
\begin{proof}
  Fix $u, v \in B$ and $j \in E$.
  Let $a$ and $b$ be the $(u,j)$th entry of $A_1$ and the $(v,j)$th entry of $A_2$, respectively.
  Assume that $ab \ne 0$.
  Then $ab$ contributes to $P^\#_{u,v}$ if and only if $w_1(u)+w_2(v) = w(j)$, which is equivalent to $w_1(u) = w_1(j)$ and $w_2(u) = w_2(j)$ by \cref{prop:weighted_intersection_cocircuit} and~\ref{item:W1}.
  Hence $ab$ contributes to $P^\#_{u,v}$ if and only if it contributes to the $(u,v)$th entry of $A_1^\# \trsp{{A_2^\#}}$.
\end{proof}

We prove \cref{cor:counting_via_sharp} using \Cref{lem:from_A_sharp_to_tcm}, which guarantees the validity of \cref{alg:weighted_pfaffian_pair}.

\begin{proof}[of \cref{cor:counting_via_sharp}]
  Suppose that $\zeta = \hat{\zeta}$.
  Then $(p_1, p_2)$ is optimal on (DB) since the associated objective value is $\zeta = \hat{\zeta}$.
  By \cref{lem:algebraic_form_of_weighted_intersection}, $N$ is equal to the coefficient of $\theta^{\hat{\zeta}}$ in $\det P(\theta)$, which is the same as $\det A_1^\# \trsp{{A_2^\#}}$ by the latter part of \cref{prop:combinatorial_relaxation_det} and \cref{lem:from_A_sharp_to_tcm}.

  If $\zeta < \hat{\zeta}$, then $N = 0$ by \cref{lem:tightness_of_intersection}.
  In addition, $(p_1, p_2)$ is not optimal on (DB).
  Hence $\det A_1^\# \trsp{{A_2^\#}}$ must be zero by the former claim of \cref{prop:combinatorial_relaxation_det}.
  Thus $\det A_1^\# \trsp{{A_2^\#}} = N$ holds on both cases.
\end{proof}

\begin{remark}
  In the above arguments, we have constructed the bipartite graph $G$ from $P(\theta) = A_1 D\prn{\theta^w} \trsp{A_2}$.
  On the other hand, \cref{alg:weighted_pfaffian_parity} builds a graph from $\Phi^*\prn{\theta}$ instead of $Q(\theta) \defeq A \Delta\prn{\theta^w} \trsp{A}$ for a Pfaffian parity $(A, L)$.
  Assuming that $A$ is pivoted so that $A[B] = I_{2r}$ with a minimum-weight parity base $B$, we conjecture that the number of minimum-weight parity bases is equal to the coefficient of $\theta^{\hat{\delta}(Q)}$ in $\pf Q(\theta)$, where $\hat{\delta}(Q)$ is defined in \cref{sec:counting_on_weighted_pfaffian_parities}.
  If this is the case, we can improve the running time of the weighted counting algorithm for Pfaffian parities.
  Moreover, we can use an arbitrary algorithm to output $B$ since $C^*$ is no longer needed, which might further improve the running time for specific instances.
  This conjecture is true for linear matroid intersection by \cref{lem:tightness_of_intersection,lem:algebraic_form_of_weighted_intersection} and for the matching problem by the definition of Pfaffian.
\end{remark}

\end{document}